\documentclass[a4paper,prx,notitlepage,onecolumn,nofootinbib,superscriptaddress,longbibliography,floatfix,11pt,accepted=2026-05-22]{quantumarticle}
\pdfoutput=1
\usepackage[numbers,sort&compress]{natbib}
\usepackage{caption}
\usepackage{subcaption}
\usepackage{amsmath, amssymb, amsfonts, amsthm}
\usepackage{tabularx,graphicx}
\usepackage{epstopdf}
\usepackage{graphicx}
 \usepackage{spectralsequences} 
\usepackage{latexsym}
\usepackage{tikz-cd}
\usepackage{color, colortbl}
\usepackage{psfrag}
\usepackage{bbm}
\usepackage{bm}
\usepackage{dsfont}
\usepackage{feynmp}
\usepackage{slashed}
\usepackage{multirow}
\usepackage{framed}
\usepackage{braket}

\makeatletter
\def\l@subsubsection#1#2{}
\makeatother

\newcommand{\beq}{\begin{eqnarray}}
	\newcommand{\eeq}{\end{eqnarray}}

\newcommand{\Z}{\mathbb{Z}}

\newcommand{\la}{\langle}
\newcommand{\ra}{\rangle}

\DeclareMathOperator{\tr}{tr}

\newcommand{\di}{{\rm dist}}
\newcommand{\ind}{{\rm ind}}
\newcommand{\bsp}{\begin{split}}
	\newcommand{\esp}{\end{split}}

\newcommand{\hc}{{\rm h.c.}}

\newcommand{\ie}{{i.e., }}
\newcommand{\eg}{{e.g., }}

\newcommand{\mO}{{\mathcal{O}}}

\newcommand{\Le}{\mathcal{L}}
\newcommand{\Ri}{\mathcal{R}}
\newcommand{\C}{\mathcal{C}}
\newcommand{\Hi}{\mathcal{H}}
\newcommand{\sa}{\mathcal{S}}
\newcommand{\ZZ}{\mathcal{Z}}
\newcommand{\Al}{\mathcal{A}}
\newcommand{\au}{\mathrm{Aut}}
\newcommand{\B}{\mathcal{B}}
\newcommand{\Su}{\mathrm{Supp}}
\newcommand{\U}{\mathrm{U}(1)}
\usepackage{color}
\definecolor{darkblue}{rgb}{0.,0.,0.4}
\definecolor{darkred}{rgb}{0.5,0.,0.}
\definecolor{BlueViolet}{RGB}{138,43,226}
\definecolor{SkyBlue}{RGB}{30,144,255}
\definecolor{DarkGreen}{RGB}{0,100,0}
\usepackage[pdftex,colorlinks=true,linkcolor=darkblue,citecolor=blue,urlcolor=darkred]{hyperref}
\usepackage[normalem]{ulem}

\theoremstyle{plain}
\newtheorem*{theorem*}{Theorem}
\newtheorem{theorem}{Theorem}
\newtheorem{lemma}{Lemma}
\newtheorem{corollary}{Corollary}
\newtheorem{Proposition}{Proposition}

\newtheorem{definition}{Definition}

\newcommand{\add}[1]{ {   #1 }}

\begin{document}
	\title{
Quantum Cellular Automata on Symmetric Subalgebras}
	
	\author{Ruochen Ma}
    \affiliation{Kadanoff Center for Theoretical Physics \& Enrico Fermi Institute, University of Chicago, Chicago, IL 60637, USA}

    \author{Yabo Li}
    \affiliation{Center for Quantum Phenomena, Department of Physics, New York University, 726 Broadway, New York, New York 10003, USA}

\author{Meng Cheng}
\affiliation{Department of Physics, Yale University, New Haven, Connecticut 06511-8499, USA}

\begin{abstract}
We investigate quantum cellular automata (QCA) on one-dimensional spin systems defined over a \emph{subalgebra} of the full local operator algebra—the symmetric subalgebra under a finite Abelian group symmetry $G$. For systems where each site carries a regular representation of $G$, we establish a complete classification of such subalgebra QCAs based on two topological invariants: (1) a surjective homomorphism from the group of subalgebra QCAs to the group of anyon permutation symmetries in a $(2+1)d$ $G$ gauge theory; and (2) a generalization of the Gross-Nesme-Vogts-Werner (GNVW) index that characterizes the flow of the symmetric subalgebra. Specifically, two subalgebra QCAs correspond to the same anyon permutation and share the same index if and only if they differ by a finite-depth unitary circuit composed of $G$-symmetric local gates. We also identify a set of operations that generate all subalgebra QCAs through finite compositions. As an example, we examine the Kramers-Wannier duality on a $\mathbb{Z}_2$ symmetric subalgebra, demonstrating that it maps to the $e$-$m$ permutation in the two-dimensional toric code and has an irrational index of $\sqrt{2}$. Therefore, it cannot be extended to a QCA over the full local operator algebra and mixes nontrivially with lattice translations.
\end{abstract}

\maketitle

\tableofcontents

\section{Introduction}

Quantum Cellular Automata (QCA) serve as models for the discrete-time evolution of quantum many-body systems, incorporating two essential principles of quantum many-body physics: unitarity and locality. Originally proposed as an alternative paradigm for quantum computation \cite{watrous1995one}, QCA have recently been linked to a broad range of topics, including the study of invertible topological phases \cite{2020FHHbeyondcoho}, periodically driven (Floquet) quantum dynamics \cite{2016FloquetPo,2017fQCA,2019fQCA,2016AFAI}, topological phases in many-body localized systems \cite{2024longelse}, and the simulation of quantum field theories \cite{2007QFT,2020QFT,2020QFT2}.

To be concrete, the general setup of QCA, as adopted in most existing literature, is as follows \cite{2019review,2020review}. We consider a spatial lattice with a finite-dimensional Hilbert space (either bosonic or fermionic) at each site. The full Hilbert space and the algebra of observables are simply the tensor product of those on each site. The QCA evolves the system over discrete time steps via a unitary operator or, equivalently, as an automorphism of the algebra of observables in the Heisenberg picture. Crucially, each local operator is transformed by the QCA into an operator with finite support on a nearby neighborhood, and the distance between the evolved operator and the initial one is uniformly bounded by a fixed finite constant across the entire system. In other words, QCA models discrete unitary dynamics with a strict causal cone.

In one spatial dimension (1D), our understanding of QCAs has reached a largely mature stage. QCAs on a spin system are completely classified by a group homomorphism to the GNVW index \cite{2009GNVW}, which takes values in positive rational numbers, with the group of finite-depth circuits (FDCs) as the kernel (having a trivial index). This means that two QCAs differ only by an FDC if and only if they have the same index. A representative example of a QCA with a nontrivial index is a lattice shift. Further enrichment of the classification arises when imposing the additional constraint that the QCA conserves a global symmetry \cite{2017MPUCirac,2018sMPU,2021U1QCA} or when considering the case where the underlying lattice system is fermionic \cite{2019fQCA}. QCAs in higher spatial dimensions remain a topic with many open questions that attract active research. Recently, significant progress has been made in the classification and construction of nontrivial examples \cite{2023FHH,Haah:2019fqd, 2020FH,2023Haah,2024FHH}.

A natural question arises: Since most current discussions focus on QCAs acting on the tensor product algebra of a lattice—\ie the product of operator algebras at each site (with a few exceptions \cite{2023JonesDHR,2024Jonesindex,2024JW}, discussed below)—can we generalize the notion of QCAs to be defined on a \emph{subalgebra} of this algebra, which may not have a tensor product structure? What would be the classification and characterization of such QCAs? In this work, we examine the simplest scenario by considering QCAs defined on the subalgebra that is symmetric under a finite Abelian group symmetry.

To be precise, we start with a 1D lattice chain, assigning a spin Hilbert space to each site. The chain is acted upon by a finite Abelian symmetry $G$. We assume the symmetry is implemented as a tensor product of onsite symmetry generators, and each site carries a regular representation of the symmetry group. The full algebra of observables is defined as all local operators on the chain—\ie operators with finite support for an infinite chain or operators whose size is much smaller than the system size for a finite lattice with a periodic boundary condition. We consider a subalgebra of this full local operator algebra, namely, the one invariant under the $G$ symmetry. A $G$-symmetric subalgebra QCA is then defined as an automorphism of the $G$-symmetric subalgebra that maps each local operator to another local operator in the same subalgebra, supported on a nearby finite neighborhood. Note that this definition includes $G$-symmetric FDCs (i.e. those that commute with the $G$ symmetries) as special cases.

One motivation for considering QCAs on symmetric subalgebras arises from recent development in generalized (non-invertible) symmetries \cite{2015GGS,2023ShaoTasi,2024noninvICTP}. Notably, from the viewpoint of local operator algebras, a broad class of conventional group-like symmetries—including tensor product symmetries, anomalous symmetries \cite{2014elsenayak}, and lattice translations—can themselves be viewed as ordinary QCAs, i.e., locality-preserving automorphisms of the full local operator algebra. A simple example of a QCA acting on a symmetric subalgebra is provided by the Kramers-Wannier (KW) duality. On a qubit chain with $\Z_2$ symmetry generated by the product of onsite Pauli-$X$ operators, the KW duality defines a locality-preserving automorphism on the $\Z_2$ symmetric subalgebra via $X_i \mapsto Z_i Z_{i+1}$ and $Z_i Z_{i+1} \mapsto X_{i+1}$, while it annihilates all $\Z_2$-odd local operators. This can be understood as the KW duality implementing a gauging of the $\Z_2$ global symmetry \cite{2019dualweb}. This fact makes it a QCA (specifically, an automorphism) defined \emph{only} on the symmetric subalgebra. Although non-invertible symmetries are an active area of research, most existing literature focuses on case-by-case studies of specific examples (e.g., KW duality, the Kennedy-Tasaki transformation \cite{kennedy1992hidden}). A systematic discussion starting from the axiomatic definition of QCA is lacking, making the full landscape of such transformations—especially a general understanding of their equivalence classes—unclear. For instance, one may ask: what are all possible QCAs on a symmetric subalgebra?

In this work, we present a complete classification of QCA defined on a subalgebra symmetric under a finite Abelian group symmetry\footnote{This encompasses all Tambara-Yamagami (TY) fusion categorical symmetries \cite{tambara1998tensor} in 1D.}, where each site carries a regular representation of the group. We prove that such QCAs, up to FDCs with $G$-symmetric local gates, are completely classified by two ``invariants".

    The first invariant is defined by how the QCA transforms two distinguished types of operators: (1) a symmetry operator restricted to a large but finite region, which creates a pair of $G$ domain walls; and (2) a product of two distant (but still finitely separated) charged operators, such that the total charge is neutral.  
    
    We demonstrate that by arranging these operators into distinct geometries, their commutation relations can be mapped to the exchange statistics of identical anyons and the braiding statistics of distinct anyons in a 2D $G$-gauge theory. Consequently, a QCA on the symmetric subalgebra corresponds to an anyon permutation symmetry. We further show that, even without assuming translation invariance, the corresponding anyon permutation symmetry is a global invariant determined solely by the QCA.
In fact, this correspondence establishes a surjective homomorphism from the group of subalgebra QCAs to the group of anyon permutations. Notably, since anyon permutation symmetry can be non-Abelian, the group of subalgebra QCAs can also be non-Abelian. This contrasts with the group of QCAs (modulo finite-depth circuits) defined on the full operator algebra (referred to as unitary QCAs, or uQCAs), which is always Abelian \cite{2022FrHH}.
    
    The second invariant is the generalization of the GNVW index to effectively capture how the symmetric subalgebra is shifted under a QCA. The GNVW index itself can be interpreted as quantifying the flow of an ``incompressible conserved quantity"—quantum information \cite{2021Gonginfoflow} or, more precisely, the dimension of the local operator algebra \cite{2009GNVW}—and it can be computed locally at any point, yielding a universal value. In Ref. \cite{2009GNVW}, this property is termed a ``locally computable global invariant." In this work, we introduce an index (denoted by $\ind$) that is locally computable and reduces to the GNVW index when applied to uQCAs. We demonstrate that this index shares many properties with the GNVW index and can be interpreted as characterizing the flow of the subalgebra. In summary, for a chain with a regular representation of a finite Abelian group $G = \prod_j \Z_{n_j}$ per site {(here all $n_j$ are prime powers, and the decomposition of $G$ is unique up to permutation of $n_j$)}, the index has the following properties: 
\begin{itemize} 
\item The index of a QCA can take values in \begin{equation} 
\ind (\alpha) \in \Big\{ \prod_j (\sqrt{n_j})^{m_j} \prod_x p_x^{q_x} \bigg| q_x \in \Z,\ m_j = 0,1 \Big\}, 
\end{equation} 
where $\{ p_x \}$ are the prime factors of $|G| = \prod_j n_j$. This value is a global invariant that is independent of the specific position where it is computed. 
\item For two QCAs $\alpha$ and $\beta$, we have $\ind(\alpha \circ \beta) = \ind(\alpha)\ind(\beta)$. 
\item An FDC with $G$-symmetric gates (sFDC) has index 1. {Moreover, a QCA that can be extended to a uQCA on the full operator algebra has an index equal to its GNVW index, while those with an irrational index are not extendable to a uQCA on the full algebra. }
\end{itemize}

We also conclude that any two QCAs associated with the same anyon permutation and the same index differ only by an FDC composed of symmetric gates.

More concretely, we further show that each QCA can be reduced to a finite composition of the following four operations:
\begin{itemize} 
\item A lattice translation of a sub-Hilbert space with prime dimension, implemented in a way that commutes with the $G$ symmetry;

\item An FDC that, when acting on a $G$-symmetric product state, entangles it into a $G$-symmetry protected topological (SPT) state \cite{2011chenspt1d}. These first two operations also appear in the classification of uQCAs. With the inclusion of additional ancilla qubits, these two operations enumerate all possible $G$-symmetric uQCAs \cite{2018sMPU}, up to FDCs with symmetric gates;

\item A product of onsite unitaries that implements a nontrivial outer automorphism of the group $G$. Since this operation swaps between different elements of $G$, it is not a symmetric uQCA. However, it constitutes an automorphism of the symmetric subalgebra, which by definition contains operators with only trivial charge under $G$;

\item Finally, noting that $G$ can always be expressed as a direct product of cyclic subgroups, $G = \prod_j \mathbb{Z}_{n_j}$, the last ingredient is the KW duality for each cyclic component, $\Z_{n_j}$.
\end{itemize}

We note that the two invariants are not entirely independent. This is evident from the four ``generating" QCAs: the first three classes of QCAs have integer index because they can all be lifted to uQCAs on the full algebra. As a result, a subalgebra QCA with square root ind must involve operations from the last class and can not be lifted to a uQCA. In other words, QCAs corresponding to certain anyon permutations must have square root ind. 

\subsection{Subalgebra QCAs and non-invertible symmetries}

As discussed earlier, an important motivation for this work is the connection between subalgebra QCAs and non-invertible symmetries. Let us now elaborate on it more. If a subalgebra QCA 
 $\alpha$ can not be lifted to a uQCA, one can still define it as a non-invertible operator on the full algebra, by simply composing it with the projector to the symmetric subspace. This way $\alpha$ can be thought of as a non-invertible symmetry. If we have a local Hamiltonian invariant under $\alpha$, the low-energy states of the Hamiltonian typically admit a continuum field-theoretic description, and $\alpha$ becomes a fusion category symmetry in the field theory. However, one should keep in mind that there is no one-to-one correspondence between subalgebra QCAs (or generally, lattice non-invertible symmetries) and fusion categories, as discussed in \cite{2024Seibergshaoseif}. For example, the same KW non-invertible symmetry may flow to different fusion category symmetries~\cite{2024Seibergshaoseif}. 

 Recently, there has been considerable discussion on whether a non-invertible symmetry mixes with lattice translations \cite{2024seibergshao,2024Seibergshaoseif}. Our index offers a quantitative characterization of this mixing, which is otherwise difficult to determine when the transformation is not a simple mapping between products of Pauli operators.

\subsection{$\Z_2$ subalgebra}

As an illustration of our results, let us consider a familiar example of a QCA defined on a symmetric subalgebra: the Kramers-Wannier (KW) duality on a qubit chain. The KW duality is an automorphism of the subalgebra that commute with the $\Z_2$ symmetry $X=\prod_i X_i$. It is generated by the standard set $\{ Z_i Z_{i+1},\, X_i \}$. The KW transformation acts on this symmetric subalgebra as follows: 
\begin{equation}
    \mathrm{KW}:\, Z_i Z_{i+1} \to X_{i+1}, \quad X_i \to Z_i Z_{i+1},
    \label{eq:KWdual}
\end{equation}
We will show below in Sec. \ref{sec:index} that $\ind({\rm KW})=\sqrt{2}$. This makes sense as ${\rm KW}^2$ is a unit translation $\rm T$, which carries GNVW index $\ind({\rm T})=2$.

Our classification theorem shows that KW is the only nontrivial $\Z_2$ subalgebra QCA, up to lattice translations and sFDCs. More precisely:

\begin{theorem*}
    For any QCA $\alpha$ on the $\Z_2$ symmetric subalgebra $\B$, its action on $\B$ can be implemented as ${\rm KW}^q\circ {\rm T}^p \circ W$,
 where $W$ is an FDC with symmetric gates, $\mathrm{T}$ denotes the lattice translation operator to the right, and $q=0,1, p\in \Z$; the index of the QCA is given by $\sqrt{2}^q2^p$. 
\end{theorem*}  

From the 2D perspective, KW corresponds to the only nontrivial anyon permutation symmetry in the $\Z_2$ gauge theory (also known as $\Z_2$ toric code). Denote the $\Z_2$ electric charge by $e$, and the magnetic charge by $m$, then $e\leftrightarrow m$ is a symmetry of the topological order. One can view the $\Z_2$ symmetric subalgebra as the boundary algebra of a $\Z_2$ toric code on a disk (where the bulk is in the ground state), and then KW exactly maps to the $e\leftrightarrow m$ symmetry.

\subsection{$\Z_2\times\Z_2$ subalgebra}

Let us now discuss a second example: a qubit chain with a $\Z_2^e \times \Z_2^o$ symmetry generated by 
\begin{equation} 
\eta^e = \prod_{i:\mathrm{even}} X_i,\quad \eta^o = \prod_{i:\mathrm{odd}} X_i, \end{equation} 
which act on all even and all odd sites, respectively.

Besides translations, QCAs on this subalgebra are generated by the following classes: 
\begin{itemize} 
\item The KW transformations: There are two KW transformations associated with the $\Z_2^e$ and $\Z_2^o$ symmetries, respectively. We denote them as KW$_e$ and KW$_o$.
\item $spt$: In 1D, there is a nontrivial SPT phase protected by $\Z_2^e \times \Z_2^o$ symmetry, specifically the cluster chain \cite{2001cluster,2020MPSreview,2013SPTChen}, labeled by the nontrivial element of $H^2(\Z_2 \times \Z_2, \U) = \Z_2$. The entangler $spt$ acts on local operators as
\begin{equation}
    spt: \, X_i \leftrightarrow Z_{i-1} X_i Z_{i+1}.
\end{equation}
$spt$ transforms a product state with $X_i=1$ into the cluster state.

\item $out$: The group of outer automorphisms of $\Z_2 \times \Z_2$ is $S_3$, corresponding to the permutation of the three non-identity elements in $\Z_2 \times \Z_2$.
\end{itemize}

{We remark that, since the anyon permutation symmetry group in the 2D $\mathbb{Z}_2 \times \mathbb{Z}_2$ gauge theory is $(S_3 \times S_3)\rtimes \mathbb{Z}_2$, which is non-Abelian, the group of QCA on the symmetric subalgebra is therefore also non-Abelian; see Sec.~\ref{sec:Z2Z2}, in particular Fig.~\ref{Fig:anyon}, for details.}

Some cases of the QCAs have been considered in the context of non-invertible symmetries~\cite{2024seifshao}. In particular, we can now have examples of QCAs with integer indices but are nevertheless not uQCAs (which is impossible for $\Z_2$ subalgebra QCAs) 

\begin{itemize} 
\item 
First we define
\begin{equation} \alpha_1 = \mathrm{T}^{-1} \circ \mathrm{KW}_e \circ \mathrm{KW}_o. 
\label{eq:repD8-intro}
\end{equation} 

On the generators of the symmetric algebra $\B$, its action simply exchanges $X_i$ and $Z_{i-1} Z_{i+1}$. Computing the index of $\mathrm{D}$ gives \begin{equation} 
\ind(\alpha_1) = 1. 
\end{equation} 

Viewed as a non-invertible operator, $\alpha_1$ provides a realization of the Rep(D$_8$) fusion category symmetry in some lattice models. The fact that $\ind(\alpha_1)=1$ provides a quantitative characterization of the fact that the non-invertible operator in the Rep(D$_8$) fusion category does not mix with lattice translation \cite{2024seifshao}.

\item 
  We can also define
\begin{equation}
    \alpha_2 = \mathrm{KW}_e \circ \mathrm{KW}_o,
\end{equation}
which implements the transformation $X_i \mapsto Z_i Z_{i+2} \mapsto X_{i+2}$ on the standard generating set of the symmetric subalgebra $\B$. A direct computation yields
\begin{equation}
    \ind(\alpha_2) = 2.
\end{equation}

Viewed as a non-invertible operator, $\alpha_2$ realizes the Rep(H$_8$) fusion category symmetry on lattice. 

\end{itemize}

\subsection{Structure of the paper}
The rest of the paper is organized as follows. 

In Sec. \ref{sec:qca_general}, we define the setup and the notion of QCA, as well as the equivalence relations among them.

In Sec. \ref{sec:index} and \ref{sec:indexsym}, we introduce the index and establish their fundamental properties,  including its invariance under composition with FDCs consisting of symmetric gates, as well as its independence from the specific spatial location where it is computed. 

In Sec. \ref{sec:Z2subalg}, we study the simplest scenario—QCAs on the $\Z_2$ symmetric subalgebra—in greater detail and demonstrate that they can be completely classified by the index. In Sec. \ref{sec:finiteabelian}, we generalize the discussion to a general finite Abelian group. 

Sec. \ref{sec:Z2Z2} examines a concrete example of QCAs on the $\Z_2 \times \Z_2$ symmetric subalgebra and makes connections to recent discussions on non-invertible symmetries. We conclude our work in Sec. \ref{sec:discussion}.

\section{QCAs over a symmetric subalgebra}
\label{sec:qca_general}

In this section, we provide a definition of a quantum cellular automaton (QCA) over a subalgebra.

Consider an infinitely long one-dimensional lattice system $\Z$. On each lattice site $i \in \Z$, we assign a local Hilbert space $\Hi_i$, which supports an operator algebra $\Al_i$ that is isomorphic to a matrix algebra over the complex number. The full algebra of operators, $\Al(\Z)$, on the spin chain is defined as the algebra of all \emph{local} operators. Specifically, for any finite subset $\Lambda \subset \Z$, we denote the operators supported entirely on $\Lambda$ (acting as the identity elsewhere) by $\Al_\Lambda = \otimes_{i \in \Lambda} \Al_i$. The full operator algebra is then given by the union of such local algebras over all finite subsets $\Lambda$:\footnote{The quasi-local algebra on the chain can be obtained by taking the norm completion of $\Al(\Z)$ \cite{bratteli2012operator}. This added complexity, however, is irrelevant to the subsequent discussion in this work and is therefore not addressed here.} 
\begin{equation} 
\Al(\Z) = \cup_\Lambda \Al_\Lambda. 
\end{equation}
In this work, we study QCAs defined on a unital $*$-subalgebra $\B$, \ie a subalgebra $\B \subset \Al(\Z)$ that contains the identity element of $\Al(\Z)$ and is closed under the adjoint (Hermitian conjugation). Specifically, we require the subalgebra $\B$ to satisfy the following assumptions:

\begin{itemize} \item For any two subsets $I$ and $J$ of $\Z$, if $I\subset J$, then we have $\B_I \subset \B_J$, where $\B_I$ is the subalgebra of $\B$ with support fully on the interval $I$, and acts as identity in $\Z \setminus I$. 
\item If $I\cap J = \varnothing$, then $[\B_I,\B_J]=0$. \end{itemize}
Importantly, we do not assume that $\B$ factorizes as a tensor product of operator algebras on each site.

Our main focus is on the subalgebra of operators invariant under a global symmetry group $G$, referred to as $G$-symmetric subalgebra below. To illustrate, let us consider the example of a spin chain with $G=\Z_n$ symmetry. On each site, we place a local Hilbert space $\mathcal{H}_{\rm loc}=\mathbb{C}^n$, which we think of as the regular representation of $\Z_n$.\footnote{ In this work the symmetry group $G$ is taken to act onsite. This assumption is not essential, as shown in Ref.~\cite{disentanglesym}, since any anomaly-free finite group symmetry in 1D can be transformed into an onsite form by adding ancilla sites in the regular representation and then conjugating by a finite-depth circuit. On the other hand, we leave the study of QCA on subalgebras symmetric under an \emph{anomalous} symmetry to future work.
} The full operator algebra on a single site, $\Hi_{loc}$, is generated by the set $\{ X, Z \}$, where $X$ and $Z$ are generalized Pauli operators (e.g. $X$ is the shift operator and $Z$ is the clock operator). The $\Z_n$ symmetry is generated by $X=\prod_i X_i$.
The symmetric subalgebra $\B$ on the lattice is generated by a set of local unitaries $\{ Z_i^\dagger Z_{i+1}, X_i \}$, where $i$ labels a lattice site $i \in \Z$.  It is straightforward to generalize this definition to any finite Abelian group.

We introduce some additional notation. For any two sites on the chain, $i$ and $j$, the natural distance, denoted as ``dist", is given by $|i-j|$. We define the support of an operator $O$ as the smallest connected subset $\Lambda \subset \Z$ such that $O \in \Al_\Lambda$, denoted by $\Su(O)$. 

Furthermore, for any subset $\Lambda \subset \Z$, we define an ``enlarged" region, denoted by $\Lambda^{+l}$, which includes all sites within a distance no greater than $l$ from at least one site in $\Lambda$. With these definitions, we can now define a QCA on the subalgebra $\B$ as follows:
{\definition A QCA $\alpha$ on the subalgebra $\B$ is a locality-preserving (LP) $*$-automorphism of $\B$. Specifically, it satisfies the following conditions:
\begin{itemize}
    \item It preserves linearity, multiplication, and Hermitian conjugation:
    \begin{equation}
    \begin{split}
        \alpha(c_1 O_1+c_2O_2)&=c_1\alpha(O_1)+c_2\alpha(O_2), \\
        \alpha( O_1 O_2) & = \alpha(O_1) \alpha( O_2),\\
        \alpha(O^\dagger) & = \alpha( O)^\dagger,
    \end{split}
    \end{equation}
    where $O$, $O_1$ and $O_2$ are local operators within $\B$, and $c_1$, $c_2$ are complex numbers. 
    \item It has an inverse $\alpha^{-1}$ defined on the subalgebra $\B$.
    \item Both $\alpha$ and $\alpha^{-1}$ are LP, which means:
    \begin{equation}
        \Su(\alpha(O)) \subseteq \Su(O)^{+l},
    \end{equation}
    where $O \in \B$, and this condition similarly holds for $\alpha^{-1}$. The parameter $l$ is referred to as the spread of $\alpha$.
\end{itemize}
\label{def:QCA}
}

We remark that, although we work here on an infinite lattice, one can also consider a finite chain with a periodic boundary condition. The full operator algebra $\Al$ then includes all local operators on the chain with support much smaller than the system size $L$. We can similarly define a QCA on a subalgebra $\B \subset \Al$, to be an LP $*$-automorphism of $\B$. This definition remains well-defined as long as the system size $L$ is much larger than the spread $l$.

It is clear from the definition that QCAs on $\B$ form a group, where the group multiplication is given by the composition, and the identity element is the identity automorphism. 

A large class of QCAs on $\B$ is given by FDC made of local unitary operators in $\B$. Motivated by the definition of equivalence classes of QCAs in spin chains, we propose the following definition of equivalence relation between QCAs on $\B$:

{\definition Two QCAs on $\B$, $\alpha$ and $\beta$, are equivalent if
\begin{equation} 
\alpha(O) = W \circ \beta(O), \quad \forall O\in \B,
\end{equation} 
where $W$ is an FDC with local unitary operators in $\B$. \label{def:equivalence} }

It is straightforward to confirm that this definition is both symmetric and transitive, thereby establishing an equivalence relation.

It is also insightful to consider QCAs on the subalgebra $\B$ that can be embedded into FDCs on the full algebra. Obviously, they include all FDCs made of local unitary gates in $\B$. However, the converse is not always true. In fact, if the subalgebra 
 is defined by a finite symmetry $G$ in a spin chain where each site carries a regular representation of the group, a complete classification of  $G$-symmetric FDCs (i.e. FDCs commuting with $G$) has been obtained in \cite{2021Gonginfoflow}.  We briefly review this result as it will play an important role in our study. 
 
 Each $G$-symmetric FDC is associated with an element in the second cohomology group $H^2(G, \U)$. Physically, the cohomology class characterizes the SPT order created by applying the FDC to a $G$-symmetric product state. The following theorem shows that the cohomology class completely classifies $G$-symmetric FDCs:

 \begin{theorem}
 Let $G$ be a finite Abelian group, and suppose each site has a local Hilbert space in the regular representation of $G$. Two $G$-symmetric FDCs are equivalent (namely, they differ by a FDC with $G$-symmetric gates), if and only if they have the same cohomology class in $H^2(G, \U)$. 
 \label{thm:GsymFDC}
 \end{theorem}

The proof of Theorem \ref{thm:GsymFDC} is given in Appendix \ref{app:sFDC}.

\section{The index}
\label{sec:index}

We now define a mathematical object, referred to as the \emph{index}, which will be shown to characterize the information flow induced by a QCA $\alpha$. To do so, we first review the definition of $\eta(\Al_1,\Al_2)$ \cite{2009GNVW}, a measure of the overlap between two operator algebras $\Al_1$ and $\Al_2$, both of which are subalgebras of $\Al(\Z)$. Let $\Al_1$ and $\Al_2$ be two operator algebras consisting of operators acting non-trivially on a finite-dimensional Hilbert space. Denote a complete orthonormal basis of $\Al_1$ by $\{ O_a \}$, \ie $\{ O_a \}$ satisfies $\tr(O_a^\dagger O_{a'}) = \delta_{a a'}$, where $\tr$ denotes the normalized trace such that $\tr(\mathbf{1}) = 1$, with $\mathbf{1}$ being the identity operator on the Hilbert space where $\Al_1$ and $\Al_2$ are defined. Similarly, let $\{ O_b \}$ be a complete orthonormal basis for $\Al_2$. The overlap $\eta(\Al_1,\Al_2)$ is defined as
\begin{equation}
    \eta(\Al_1,\Al_2) = \sqrt{ \sum_{O_a\in \Al_1, O_b\in \Al_2} |\tr(O_a^\dagger O_b)|^2},
\end{equation}
which effectively measures the dimension of the intersection of the two algebras, \ie $\eta(\Al_1,\Al_2) = \sqrt{\dim(\Al_1\cap \Al_2)}$. An important property of $\eta(\Al_1, \Al_2)$ is that it is basis-independent, \ie when $\Al_1$ and $\Al_2$ are viewed as vector spaces over the complex numbers, the overlap remains invariant under arbitrary unitary basis transformations. This can be verified directly as follows: let $V$ be a unitary that relates two bases, $\{ A_a \}$ and $\{ \tilde{A}_{a} \}$, of $\Al_1$. Then we have
\begin{equation}
\begin{split}
    & \eta(\Al_1,\Al_2)  = \sqrt{ \sum_{\tilde{A}_a, O_b} |\tr(\tilde{A}_a^\dagger O_b)|^2} \\
     = & \sqrt{ \sum_{a, a',a'', O_b} V_{aa'}^* V_{aa''} \tr(A_{a'}^\dagger O_b)  \tr(A_{a''} O_b^\dagger)} = \sqrt{ \sum_{A_a, O_b} |\tr(A_a^\dagger O_b)|^2}.
\end{split}    
\end{equation}
The same holds for any arbitrary basis change of $\Al_2$.

{\definition
The index of a QCA $\alpha$ on the subalgebra $\B$ is defined as 
\begin{equation} 
\ind(\alpha) = \frac{\eta(\alpha(\B_-), \B_+)}{\eta(\B_-, \B_+)}, 
\end{equation} 
where $\B_+$ and $\B_-$ are subalgebras of $\B$, supported non-trivially only within the intervals $I_+$ and $I_-$ shown in Fig.\ref{Fig:index}, respectively.
\label{def:ind}
}

\begin{figure}
\begin{center}
  \includegraphics[width=.55\textwidth]{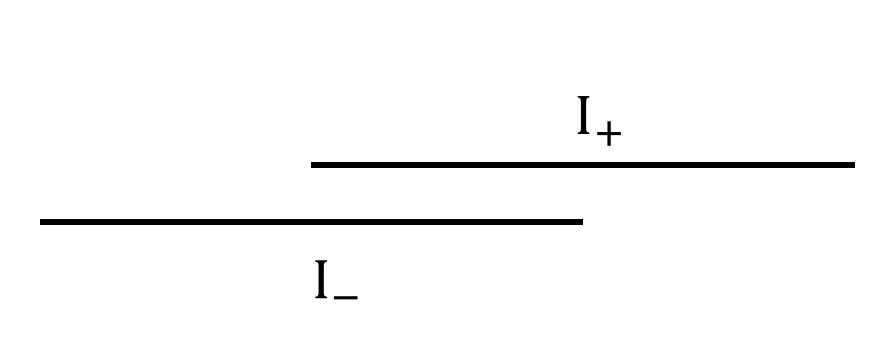} 
\end{center}
\caption{The two intervals used to define the index in Def. \ref{def:ind}. Each of the three segments, $I_-\cap\overline{I_+}$, $I_-\cap I_+$, and $I_+\cap\overline{I_-}$, is chosen to have length greater than $l$, for reasons that will be clarified shortly.
}
\label{Fig:index}
\end{figure}

Notice that when $\B = \Al(\Z)$, \ie the full operator algebra, the index coincides with the GNVW index defined in Ref.\cite{2009GNVW}. A quick verification can be done: consider the case where each site is assigned a $d$-dimensional Hilbert space, and $\alpha$ is the translation by one site to the right. A direct computation shows that $\ind (\alpha) = d$. For a more detailed discussion, see Ref.\cite{2023bulkboundary}.

We establish a key property of the index, stating that the index of a QCA $\alpha$ remains invariant when composed with a finite-depth circuit (FDC) with gates contained \emph{within} the subalgebra $\B$, provided the intervals are chosen so that $\partial I_+$ and $\partial I_-$ are sufficiently distant from each other. We denote the boundary sites of an interval $I$ by $\partial I$. We define the $\xi$-boundary of a set $I$, denoted by $\partial_\xi I$, to be
\begin{equation}
    \partial_\xi I  = \{ r \in \Z: \di(r,I)\leq \xi \, \mathrm{and} \, \di(r,\overline{I})\leq \xi \},  
\end{equation}
{where $\di(r,I)$ denotes the minimal distance from site $r$ to the set $I$ (with $\di(r,I)=0$ if $r\in I$), and similarly for $\di(r,\overline{I})$, where $\overline{I}$ denotes the complement of $I$.} An FDC is defined as a finite-depth circuit with $n$ layers of disjoint unitary gates, each with a diameter bounded above by a constant $\lambda$.

{\Proposition Let $\alpha$ be a QCA on the subalgebra $\B$ with a spread $l$, and let $W$ be an FDC whose gates lie within $\B$. We then have:
\begin{itemize}
    \item $\ind(\alpha\circ W) = \ind(\alpha\circ W')$, where $W'$ is obtained from $W$ by removing all gates except those supported non-trivially within $\partial_{n\lambda} I_- \cap \partial_l I_+$;
    \item $\ind(W\circ\alpha) = \ind(W' \circ\alpha)$, where $W'$ is obtained from $W$ by removing all gates except those supported non-trivially within $\partial_{l} I_- \cap \partial_{n\lambda} I_+$. 
\end{itemize}
Here an FDC $W$ acts on the algebra of operators by conjugation, \ie $W(O) = W^\dagger O W$.
\label{prop:circuitinv}
}

\begin{proof} The proposition follows straightforwardly from the illustration in Fig. \ref{Fig:Prop1}. The gates that are removed perform only a unitary basis transformation of $\B_-$ or $\B_+$ and, as such, do not influence the index. Let us consider the QCA $\alpha \circ W$. The blue gates in Fig.\ref{Fig:Prop1}, which lie outside $\partial_{n\lambda} I_-$, induce only a unitary basis rotation of $\B_-$ through conjugation and can therefore be removed. For the red gates (collectively referred to as $W'$), we have \begin{equation} 
\alpha \circ W' = 
\alpha(W') \circ \alpha, 
\end{equation} 
where $\alpha(W')$ remains a unitary in $\B$. If $W'$ is positioned outside $\partial_l I_+$, then by Def.\ref{def:ind}, $\alpha(W')$ induces only a unitary basis rotation of $\B_+$ through conjugation and thus does not affect the index. This argument similarly applies to the QCA $W \circ \alpha$. 
\end{proof}

\begin{figure}
\begin{center}
  \includegraphics[width=.55\textwidth]{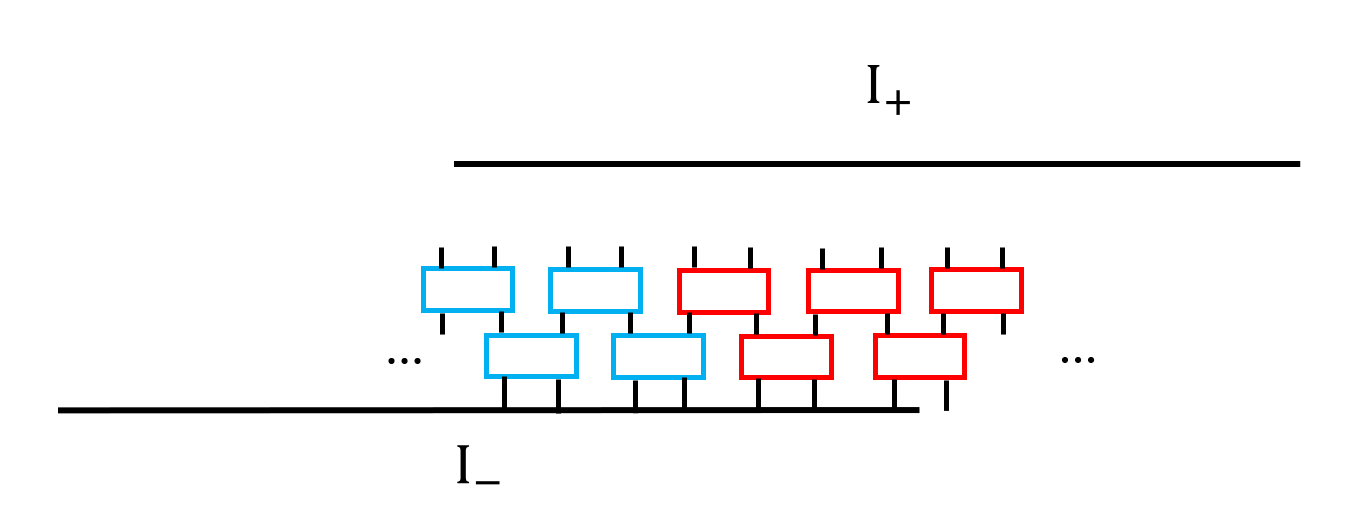} 
\end{center}
\caption{A two-layer ($n=2$) FDC. Neither the blue nor the red-colored gates affect the index as they only perform a unitary basis transformation of $\B_-$ or $\B_+$.
}
\label{Fig:Prop1}
\end{figure}

Proposition \ref{prop:circuitinv} leads to an immediate corollary, which states that when the boundaries of $I_+$ and $I_-$ are sufficiently far apart, the index remains invariant under an FDC with gates in $\B$. Moreover, such an FDC itself has an index of 1, since the identity QCA has an index of 1 according to Def. \ref{def:ind}.

{\corollary If the minimum distance between $\partial I_+$ and $\partial I_-$ is greater than $(n+1)\lambda + l$, the index is invariant under a finite-depth circuit $W$ with gates in $\B$. In particular, such an FDC itself has an index of 1. 
\label{cor:FDCinvariant}
}

\subsection{Index for QCAs over the subalgebra symmetric under a finite Abelian group}
\label{sec:indexsym}

Proposition \ref{prop:circuitinv} establishes that the index is invariant under FDCs with gates contained in $\B$. In this section, we demonstrate that if $\B$ is the subalgebra symmetric under a finite Abelian group (referred to as the symmetric subalgebra below), then, given an additional assumption on the onsite symmetry representation, the index of an LP automorphism is a global invariant. Specifically, this implies that the index is independent of the exact length and positions of the two intervals on $\Z$, as long as $\partial I_+$ and $\partial I_-$ are sufficiently separated. We first highlight two key properties of the symmetric subalgebra, which will be crucial for our subsequent results.

To be specific, we consider a finite Abelian group $G$, which is isomorphic to a direct product of cyclic groups, $G \cong \otimes_j \Z_{n_j}$, where ${ n_j }$ labels the order of each subgroup. On each site, we place a local Hilbert space $\Hi_{loc} = \otimes_j V_j$, such that $\text{dim}(V_j) = n_j$, where each $V_j$ is in the regular representation of the corresponding subgroup $\Z_{n_j}$. For simplicity in notation, we present the following observations using the example $G = \Z_n$, which extends naturally to each cyclic subgroup and thus applies to a general finite Abelian group. The full operator algebra on a single site, $\Hi_{loc}$, is generated by the set $\{ X, Z \}$, where $X$ and $Z$ are generalized Pauli operators.
{The symmetric subalgebra $\B$ satisfies the following properties: 
\begin{itemize}
\item The symmetric subalgebra $\B$ on the lattice is generated by a set of local unitaries $\{ Z_i^\dagger Z_{i+1}, X_i \}$, where $i$ labels a lattice site $i \in \Z$. An important property of $\B$ is that it has a complete orthonormal basis, with each element being a product of elements from this generating set. Hereafter we refer to the above generating set as the standard generating set of the symmetric subalgebra.
\item In this generating set, the elements commute up to a phase.  

\item  For any finite subset $I \subset \mathbb{Z}$, denote by $\B_I$ and $\Al_I$ the subalgebras of $\B$ and $\Al(\mathbb{Z})$, respectively, that are fully supported on $I$. The commutant of $\B_I$ in $\Al_I$ is generated by the symmetry operator restricted to the region $I$, \ie $g_I=\prod_{i \in I} X_i$.\footnote{ This is because, by definition, $\B_I$ is the commutant of $g_I$ in $\Al_I$. Its commutant in $\Al_I$ is therefore the double commutant of $g_I$, which, by the von Neumann double commutant theorem, coincides with the algebra generated by $g_I$.
} 
\end{itemize}
}

It is important to note the significance of the specific choice of local Hilbert space $\Hi_{loc}$, for which we provide an illustrative counterexample: Consider a chain with a three-dimensional Hilbert space per site, which is a direct sum of $\Z_2$ even and odd components: $\Hi_i = \Hi_i^e \oplus \Hi_i^o$, where dim$(\Hi_i^e)=1$ and dim$(\Hi_i^o)=2$ for all sites $i \in \Z$. We denote such a local Hilbert space as $\Hi_i = \mathbb{C}^{1|2}$, and a single-site $\Z_2$ odd operator at site $i$ as $O_i^o$. 
On each site, there are no $\Z_2$-odd unitary operators. Consequently, within this specific Hilbert space, the symmetric subalgebra—in particular, the two-site operator $O_i^o O_j^o$ ($i \neq j$), which clearly lies in the $\Z_2$ symmetric subalgebra—cannot be generated by a set of $\Z_2$-symmetric unitaries. Hereafter, we focus on the case where each site carries a regular representation of the symmetry, and we comment on where this assumption is important.

In Def.\ref{def:QCA}, we define a QCA such that both $\alpha$ and its inverse $\alpha^{-1}$ are LP. In fact, if $\alpha$ is a QCA on a symmetric algebra $\B$ with spread $l$, it is automatically guaranteed that $\alpha^{-1}$ is also a QCA with the same spread. This can be seen as follows: Consider a finite subset $I \subset \Z$, and denote the symmetric algebra on $I$ as $\B_I$. Let $J$ be the set of sites at least distance $l$ away from $I$, i.e., $J = \Z \setminus I^{+l}$, and define the symmetric algebra on $J$ as $\B_J$. Since $\alpha$ has a spread $l$, we have 
\begin{equation} [\B_I, \alpha(\B_J)] = 0. \end{equation} 
Because $\alpha$ is an automorphism of the symmetric algebra, it follows that 
\begin{equation} [\alpha^{-1}(\B_I), \B_J] = 0. 
\end{equation} 
This implies that $\alpha^{-1}(\B_I)$ either acts on $J$ as the identity or as $\prod_{i\in J} X_i$. However, the latter possibility is excluded because $\alpha^{-1}$ must map local operators to local operators within $\B$, and $\prod_{i \in J} X_i$ is a nonlocal operator acting on the entire region $J$. Therefore, $\alpha^{-1}$ is also an LP automorphism with spread $l$.

The following lemma will be useful in the subsequent discussion:
{
\lemma A QCA $\alpha$ on a symmetric subalgebra $\B$ preserves the Hilbert-Schmidt inner product, defined in terms of the normalized trace. Specifically, we have 
\begin{equation} 
\tr[\alpha(O_1)^\dagger \alpha(O_2)] = \tr(O_1^\dagger O_2), 
\end{equation} 
for operators $O_1, O_2 \in \B$.
\label{lem:tracepreserving}
}

\emph{Proof:} By the linearity of $\alpha$, and the fact that the standard generating set can generate $\B$, it suffices to prove the lemma for $O_1$ and $O_2$ that are products of elements in the standard generating set. We consider two cases:
(1) When $O_1 = O_2$, we need to show that $\tr[\alpha(O_1)^\dagger \alpha(O_2)] = 1$. This follows directly since 
\begin{equation} \tr[\alpha(O_1)^\dagger \alpha(O_2)] = \tr[\alpha(O_1^\dagger O_2)] = \tr(\mathbf{1}) = 1. 
\end{equation}

(2) When $O_1 \neq O_2$, their inner product is zero. This is because there always exists an element $O$ in the standard generating set such that $O^\dagger O_1^\dagger O_2 O = \omega O_1^\dagger O_2$, where $\omega$ is a nontrivial phase factor, i.e., $\omega \neq 1$. Thus, we have \begin{equation} 
\begin{split} 
\tr[\alpha(O_1)^\dagger \alpha(O_2)] & = \tr[\alpha(O)^\dagger \alpha(O_1)^\dagger \alpha(O_2) \alpha(O)]\\ & = \tr[\alpha(O^\dagger O_1^\dagger O_2 O)] = \omega \tr[\alpha(O_1)^\dagger \alpha(O_2)]. \end{split} 
\end{equation} 
Since $\omega \neq 1$, it follows that $\tr[\alpha(O_1)^\dagger \alpha(O_2)] = 0$, completing the proof of the lemma. \qed

We now establish two fundamental properties of the index for a QCA on the symmetric subalgebra $\B$. The first property asserts that the index of a QCA $\alpha$ on $\B$ is a global invariant; that is, its value is independent of the specific choice of intervals $I_+$ and $I_-$, provided they are sufficiently large. Let us denote the index calculated on the intervals $I_-$ and $I_+$ as $\ind(\alpha)[I_-, I_+]$.
{\theorem
If $\B$ is a symmetric subalgebra, and $\alpha$ is a QCA defined on $\B$ with operator spreading length $l$, then $\ind(\alpha)[I_-,I_+] = \ind(\alpha)[I_-\setminus a,I_+]$, where $a$ denotes sites outside of $\partial_{l} I_+$. Similarly, we have $\ind(\alpha)[I_-,I_+] = \ind(\alpha)[I_-,I_+ \setminus b]$, where $b$ denotes sites outside of $\partial_{l} I_-$.
\label{Thm:global}
}

\begin{proof}
The two intervals are illustrated in Fig. \ref{Fig:Thm1}, where we aim to show that the index remains invariant after truncating $a$ from $I_-$, where $a$ is the union of sites in $a_-$ and $a_+$. Let us denote the standard generating sets on $I_-$ and $I_+$ by $A$ and $B$, respectively. The subset of $A$ supported non-trivially on $a_-$ is labeled $A_-$, the subset supported non-trivially on $a_+$ is labeled $A_+$, and the remainder of the set, $A \setminus (A_- \cup A_+)$, is denoted by $A'$. The index of $\alpha$, computed using the intervals $I_-$ and $I_+$, can be expressed as $\ind(\alpha)[I_-, I_+] = \frac{\eta(\alpha(\langle A \rangle), \langle B \rangle)}{\eta(\langle A \rangle, \langle B \rangle)}$, where $\langle A \rangle$ denotes the algebra generated by the set $A$. The numerator can be evaluated as follows:
\begin{equation}
\begin{split}
    \eta(\alpha(\langle A \rangle), \langle B \rangle) & =\sqrt{\sum_{O_-\in \langle A_- \rangle,O\in \langle A' \rangle, O_+\in \langle A_+ \rangle, M\in \langle B \rangle} | \tr [\alpha(O_-)^\dagger \alpha(O)^\dagger \alpha(O_+)^\dagger M]|^2} \\
    & = \sqrt{\mathrm{dim}(\langle A_+ \rangle)\times\sum_{O\in \langle A' \rangle, M\in \langle B \rangle} | \tr [\alpha(O)^\dagger M]|^2}
\end{split},
\label{eq:truncateinterval}
\end{equation}
where we have employed the following observations:
\begin{itemize}
    \item The overlap can be computed using the complete orthonormal basis of $\langle A \rangle$, where each element is a product of elements in $A$, and similarly for $\langle B \rangle$. The term inside the square root is nonzero only when $O_-$ is the identity. This is because $M$ can be written as $M = \alpha(M')$, where $M' = \alpha^{-1}(M)$ is trivially supported on $a_-$, given that $\alpha^{-1}$ has a spread of $l$. Consequently, using Lemma \ref{lem:tracepreserving}, we have \begin{equation} 
    \tr [\alpha(O_-)^\dagger \alpha(O)^\dagger \alpha(O_+)^\dagger M] = \tr (O_-^\dagger O^\dagger O_+^\dagger M'), 
    \end{equation} 
    which is nonzero only when $O_-$ is the identity.
    \item The operator $\alpha(O_+)^\dagger$ is a unitary contained within the algebra $\langle B \rangle$, since $\alpha$ is a QCA on $\B$ with a spread $l$. When multiplied by $M$, it results only in a unitary basis transformation of $\langle B \rangle$ through the left regular representation. Therefore, the term inside the square root remains the same for any unitary operator $O_+ \in  \langle A_+ \rangle$, contributing a factor of $\text{dim}(\langle A_+ \rangle)$, which represents the dimension of the algebra $\langle A_+ \rangle$ as a vector space.
\end{itemize}
Computing the denominator in the same way results in
\begin{equation}
    \eta(\langle A \rangle, \langle B \rangle) = \sqrt{\mathrm{dim}(\langle A_+ \rangle)\times\sum_{O\in \langle A' \rangle, M\in \langle B \rangle} | \tr (O^\dagger M)|^2}.
\end{equation}
The first part of the theorem follows by noting that $A'$ generates the symmetric subalgebra on $I_-\setminus a$. Observing that
\begin{equation}
    \ind(\alpha)[I_-, I_+] = \frac{\eta[\alpha(\langle A \rangle), \langle B \rangle]}{\eta(\langle A \rangle, \langle B \rangle)} = \frac{\eta[\langle A \rangle,\alpha^{-1}(\langle B \rangle)]}{\eta(\langle A \rangle, \langle B \rangle)} ,
\end{equation}
and by repeating the argument presented above, we establish the second part of the theorem.
\end{proof}

\begin{figure}
\begin{center}
  \includegraphics[width=.55\textwidth]{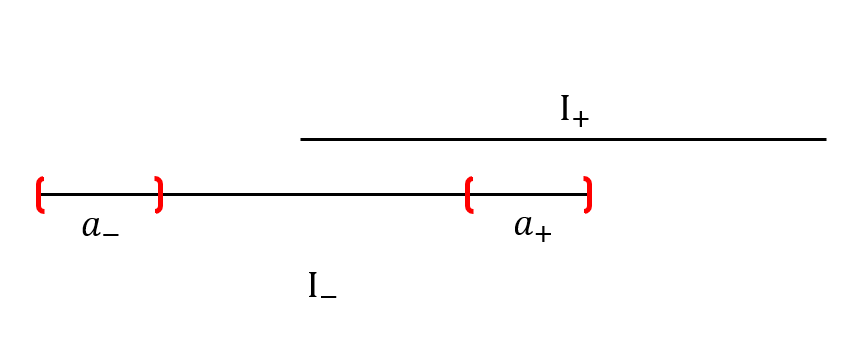} 
\end{center}
\caption{
}
\label{Fig:Thm1}
\end{figure}
A direct consequence of Theorem \ref{Thm:global} is that, although the index of a QCA is computed locally, it serves as a \emph{global} invariant. Specifically, one can expand, shrink, or move the intervals arbitrarily, as long as $\partial_l I_-$ and $\partial_l I_+$ do not intersect, and still obtain the same index value, determined solely by the QCA $\alpha$. This property is reminiscent of the GNVW index of a unitary QCA (uQCA), referred to as a ``locally computable global invariant" in Ref. \cite{2009GNVW}.

Secondly, the index is inverted under spatial reflection. The intuition behind this result is that the index quantifies a ``flow" of the subalgebra $\B$ under $\alpha$ to the right; hence, it becomes inverted when measured to the left.
{\Proposition If we define a reflected index for a QCA on the symmetric subalgebra $\B$ by exchanging $I_-$ and $I_+$, specifically:
\begin{equation}
    \ind_r(\alpha) = \frac{\eta(\alpha(\B_+), \B_-)}{\eta(\B_+,\B_-)},
\end{equation}
then the following relation holds: 
\begin{equation}
    \ind(\alpha)\ind_r(\alpha) = 1.
\end{equation}
\label{prop:reflection}
}

\begin{proof}
Consider the geometry shown in Fig. \ref{Fig:Prop2}. Let $B$, $A_-$, and $A_+$ represent the standard generating sets fully supported on the intervals $J$, $I_-$, and $I_+$, respectively. Define $A_m$ as the elements in the standard generating set that are nontrivially supported within the interval $I_m$ (so that the generators on $\partial I_m$ also have support in $I_-$ or $I_+$). Let $I = I_- \cup I_m \cup I_+$. The intervals $I_-$, $I_m$, and $I_+$ are adjacent and non-overlapping, chosen such that $J \cap \partial_l I = \varnothing$, $\partial J \cap \partial_l I_m = \varnothing$, and the length of $I_m$ exceeds $2l+1$. Clearly, $A = A_- \cup A_+ \cup A_m$ generates the symmetric algebra on $I = I_- \cup I_m \cup I_+$. We then have:
\begin{equation}
    \eta(\alpha(\la A \ra),\la B \ra) = \eta[\la A \ra ,\alpha^{-1}(\la B \ra)] = \eta(\la A \ra, \la B \ra),
    \label{eq:reflection}
\end{equation}
due to (1) $\alpha$ preserving the Hilbert-Schmidt inner product as per Lemma \ref{lem:tracepreserving}, and (2) $\alpha^{-1}(\langle B \rangle) \subset \langle A \rangle$, because $\alpha$ is a QCA with spread $l$. Further, by a same argument as that leading to Eq.\eqref{eq:truncateinterval}, we find: 
\begin{equation} 
\eta(\alpha(\langle A \rangle), \langle B \rangle) = \sqrt{\mathrm{dim}(\langle A_m \rangle)} \times \eta[\alpha(\langle A_- \rangle), \langle B \rangle] \times \eta[\alpha(\langle A_+ \rangle), \langle B \rangle], 
\end{equation} 
where we used the facts that (1) $\alpha(\langle A_m \rangle)$ is supported in $J$ and, when multiplied with $\langle B \rangle$, only induces a unitary basis transformation of $\langle B \rangle$, contributing a factor $\sqrt{\mathrm{dim}(\langle A_m \rangle)}$; and (2) the distance between $\mathrm{Supp}[\alpha(\langle A_- \rangle)]$ and $\mathrm{Supp}[\alpha(\langle A_+ \rangle)]$ exceeds 1 (since $I_m$ has length greater than $2l+1$), so they overlap with two subalgebras of $\la B \ra$ that are spatially disjoint, resulting in
\begin{equation}
    \eta[\alpha(\la A_- \ra \otimes \la A_+ \ra), \la B \ra] =\eta[\alpha(\la A_- \ra  ), \la B \ra] \times \eta[\alpha(\la A_+ \ra  ), \la B \ra] .
\end{equation}
Calculating the right-hand side of Eq. \eqref{eq:reflection} yields: 
\begin{equation} \eta(\langle A \rangle, \langle B \rangle) = \sqrt{\mathrm{dim}(\langle A_m \rangle)} \times \eta[\langle A_- \rangle, \langle B \rangle] \times \eta[\langle A_+ \rangle, \langle B \rangle], 
\end{equation} 
from which Prop. \ref{prop:reflection} follows.
\end{proof}

\begin{figure}
\begin{center}
  \includegraphics[width=.60\textwidth]{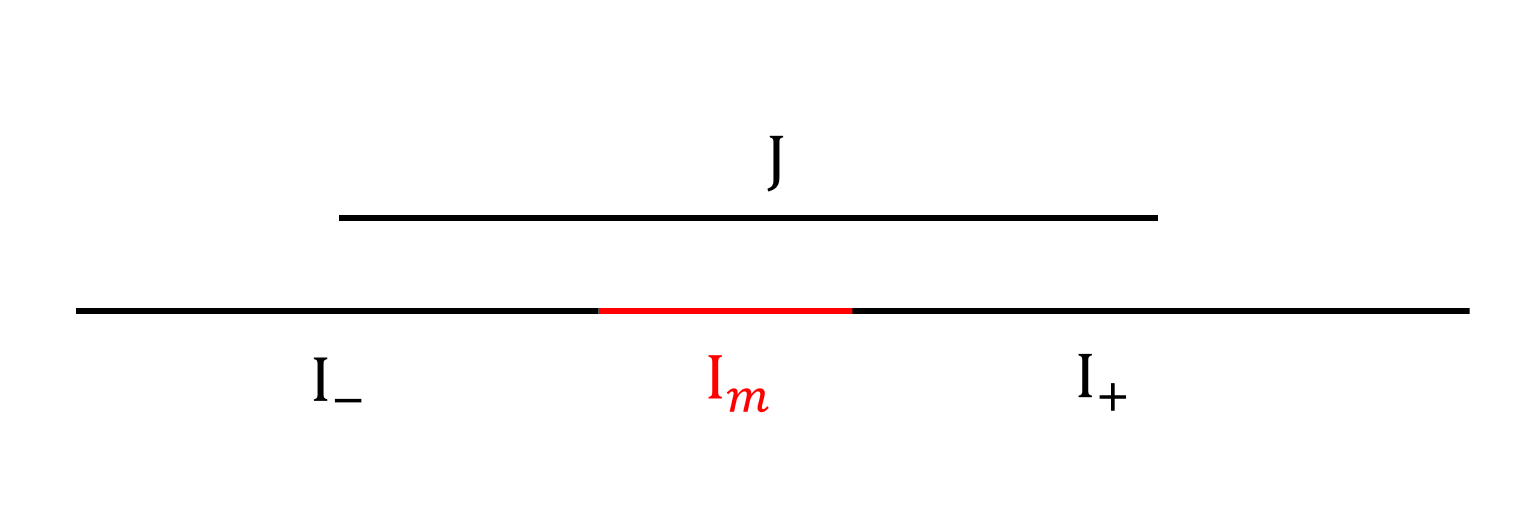} 
\end{center}
\caption{The middle region $I_m$, shown in red, is selected to have a length greater than $2l+1$.}
\label{Fig:Prop2}
\end{figure}

Let us now illustrate the definition using the example of the KW duality in Eq. \eqref{eq:KWdual}, viewed as a QCA on the $\Z_2$ symmetric subalgebra generated by the standard set $\{ Z_i Z_{i+1},\, X_i \}$. Recall that the KW transformation acts on this symmetric subalgebra as follows: 
\begin{equation}
    \mathrm{KW}:\, Z_i Z_{i+1} \to X_{i+1}, \quad X_i \to Z_i Z_{i+1},
\end{equation}
which preserves locality with a spread $l=1$. To compute the index, we take two intervals $I_-$ and $I_+$, overlapping on $\xi>1$ sites, and denote the standard generating sets fully supported on these intervals as $A_-$ and $A_+$, respectively. According to Def. \ref{def:ind}, the index is given by
\begin{equation}
    \ind(\mathrm{KW}) = \frac{\eta[\mathrm{KW}(\la A_- \ra),\la A_+ \ra]}{\eta(\la A_- \ra, \la A_+\ra)},
\end{equation}
where 
\begin{itemize} 
\item $A_-$ and $A_+$ share $2\xi - 1$ common elements (specifically, $\xi$ $X_i$ operators and $(\xi-1)$ $Z_i Z_{i+1}$ operators), each of which has order 2. Thus, the denominator is $\eta(\la A_- \ra, \la A_+\ra) = \sqrt{2^{2\xi-1}}$; 
\item $\mathrm{KW}(A_-)$ and $A_+$ share $2\xi$ common elements—the additional one being $Z_j Z_{j+1} \in A_-$, where $j+1$ is the leftmost site of $I_+$, which transforms into $X_{j+1} \in A_+$. Therefore, the numerator is $\eta[\mathrm{KW}(\la A_- \ra), \la A_+\ra] = \sqrt{2^{2\xi}}$. 
\end{itemize} 
Combining these, we find $\ind(\mathrm{KW}) = \sqrt{2}$. The fact that the index is not 1 provides a quantitative characterization of how the KW transformation intertwines with lattice translation \cite{2024Seibergshaoseif}.

Before closing this section, we note that recent work \cite{2024Jonesindex} have proposed generalizations of the GNVW index to algebras without a tensor product structure, using abstract von Neumann algebra theory and the Jones index for subfactors. However, in these approaches, the proposed index relies on analyzing the transformation of a subalgebra on a half-infinite chain, which is not locally computable and is challenging to extend to systems with periodic boundary conditions \cite{2024Jonesindex}. We will examine the relationship between these two indices in a forthcoming work \cite{index}.

\section{QCAs over the $\Z_2$ symmetric subalgebra}
\label{sec:Z2subalg}

In this section, we delve deeper into QCA's defined on the $\Z_2$ symmetric subalgebra. Specifically, we demonstrate that QCA's on the $\Z_2$ symmetric subalgebra are completely classified by their index. Moreover, we establish a homomorphism from the group of QCA's to a $\Z_2$ invariant—which measures whether the index is rational or not—that determines whether the QCA on $\B$ can be extended to a uQCA on the full operator algebra $\Al(\Z)$. We also highlight the connection of our index to one defined on a 1D fermionic system. As we will discuss, many of the arguments presented here can be generalized to other symmetric subalgebras in a straightforward manner.

\subsection{Transformations of string operators}
We begin by defining two types of string operators that play a central role in the subsequent discussions: (1) the symmetry string $S_X = \prod_{i=m}^n X_i$, confined to a finite region with length $|n - m|$ much larger than the spread $l$ of the QCA $\alpha$ (specifically, $|n - m| > 2(2l + 1)^2+2l$, for reasons that will become clear shortly); and (2) the long-range charge-hopping operator $S_Z = Z_m Z_n$ with $|n - m| > 2(2l + 2)^2+2l$, which we refer to as the charge string.\footnote{Hereafter, we always assume the length of a string operator to be larger than $2(2l + 2)^2+2l$.} These two types of string operators are referred to as patch operators in Ref.\cite{2019patch,2023patch}, and physically correspond to creating a pair of charges or domain walls at large separations. The assumption that $\alpha$ is a QCA on $\B$ imposes strong constraints on the behavior of these operators.
{\Proposition 
Under $\alpha$, the possible transformations of the operator $S_X$ are: (1) $\alpha(S_X) = O_L S_X^{-l} O_R$, where $O_L$ and $O_R$ are unitary operators supported within a distance $l$ from $m$ and $n$, respectively, and $S_X^{-l} = \prod_{i=m+l+1}^{n-l-1} X_i$, or (2) $\alpha(S_X) = O_L O_R$, which is a product of local unitaries supported within a distance $l$ from $m$ and $n$, respectively.

Similarly, the charge string $S_Z$ is transformed, under $\alpha$, to either (1) $\alpha(S_Z) = V_L V_R$, or (2) $\alpha(S_Z) = V_L S_X^{-l} V_R$, where $V_L$ and $V_R$ are unitaries supported within a distance $l$ from $m$ and $n$, respectively.
\label{prop:stringtransform}
}

\begin{proof} We first examine the symmetry string $S_X$. We denote the standard generating set on the interval $\Lambda = [m,n]$ by $A$, and the subset supported solely on $\Lambda^{-l} = [m+l+1,n-l-1]$ by $B$. Since $\alpha$ is an automorphism of $\B$, it holds that \begin{equation} \alpha(S_X) \alpha(O) = \alpha(O) \alpha(S_X) \end{equation} for any arbitrary $\Z_2$-even operator $O$ with $\Su(O) \subset \Lambda$. Since $\alpha^{-1}$ is also a QCA on $\B$ with the same spread $l$, it follows that $\langle B \rangle \subset \alpha(\langle A \rangle)$. Thus, $\alpha(S_X)$ must commute with $\langle B \rangle$, the subalgebra of $\B$ supported solely on $\Lambda^{-l}$. Consequently, it must act on $\Lambda^{-l}$ either as the identity or as $S_X^{-l}$:
\begin{equation}
\alpha(S_X) = O_1 \otimes \mathbf{1} + O_2 \otimes S_X^{-l},\label{string}
\end{equation}
where $O_1$ and $O_2$ are operators supported only within $\partial_l \Lambda$, due to the LP nature of $\alpha$. The unitarity of $S_X$ further imposes the conditions
\begin{equation}
    \begin{split}
        O_1^\dagger O_1 + O_2^\dagger O_2 &= \mathbf{1}, \\
        O_1^\dagger O_2 + O_2^\dagger O_1 &= 0.
    \end{split}
    \label{eq:unitarityofstring}
\end{equation}

Moreover, note that $S_X$ is a depth-1 FDC consisting of gates within the symmetric algebra. Since $\alpha$ is a locality-preserving automorphism of $\B$, each $\alpha(X_i)$ is a unitary supported on at most $(2l+1)$ sites. Therefore, one can arrange $\alpha(S_X)$ as an FDC with at most $(2l+1)$ layers of disjoint gates. Consequently, $\alpha(S_X)$ is a uQCA with spread upper bounded by $(2l+1)^2$. Hence, we have
\begin{equation}
\Su[\alpha(S_X) O \alpha(S_X)^\dagger]\subset \Su(O)^{+(2l+1)^2},
\label{eq:stringFDC}
\end{equation}
for any operator $O$ with $\Su(O) \subset \Lambda^{-l}$. Specifically, when $|n - m| > 2(2l + 1)^2 + 2l$, one can choose $O$ as a $\Z_2$-odd operator supported on $\Lambda^{-l}$ whose distance to $\partial_l \Lambda$ is greater than $(2l + 1)^2$. As a result, Eq.\eqref{eq:stringFDC} requires that
\begin{equation}
    \begin{split}
        O_1^\dagger O_1 - O_2^\dagger O_2 & = \pm \mathbf{1}, \\
        O_1^\dagger O_2 - O_2^\dagger O_1 &= 0.
    \end{split}
    \label{eq:stringtransform}
\end{equation}
Combining these results with Eq. \eqref{eq:unitarityofstring}, we deduce that one of $O_1$ and $O_2$ is a unitary operator, while the other is zero.

Finally, we demonstrate that $O_1$ and $O_2$ must factorize as a product of unitary operators supported near $m$ and $n$, respectively. Suppose $O_1 \neq 0$. Since $\alpha(S_X)$ is a FDC with spread $(2l+1)^2$, when $|n-m| > (2l+1)^2+2l$, we must have $\Su(O_1 O O_1^\dagger) \subset \partial_l \Lambda \cap [m-(2l+1)^2-l, m+(2l+1)^2+l] = [m-l, m+l]$ for arbitrary $O$ with $\Su(O) \subset [m-l, m+l]$, and similarly for $\Su(O) \subset [n-l, n+l]$. We therefore conclude that $O_1 = O_L O_R$ with the desired property stated in Prop. \ref{prop:stringtransform}.

Since $S_Z$ is a depth-2 FDC, \ie $S_Z = \prod_{i=m}^{n-1} Z_i Z_{i+1}$, the operator $\alpha(S_Z)$ can be arranged as an FDC with $(2l+2)$ layers of gates, each with a size bounded above by $(2l+2)$. The transformation rule for the charge string $S_Z$ then follows similarly.
\end{proof}

Prop. \ref{prop:stringtransform} imposes a constraint on the transformation of strings defined on a specific large interval $\Lambda = [m,n]$.\footnote{The proof of Prop.\ref{prop:stringtransform} relies on the two string operators, $S_X$ and $S_Z$, being FDCs, a property that may not hold when the onsite Hilbert space carries a general representation of $G$. This is one reason we focus on an onsite Hilbert space in the regular representation.} As we will now demonstrate, even without assuming translation invariance, the form of the transformation (i.e., either possibility (1) or (2) in Prop. \ref{prop:stringtransform}) is a global invariant, meaning it does not depend on the specific choice of interval, as long as its length is sufficiently large. Furthermore, by considering the composition rules and commutation relations between two strings, we can further constrain the transformation rules:
\begin{figure}
\begin{center}
  \includegraphics[width=.50\textwidth]{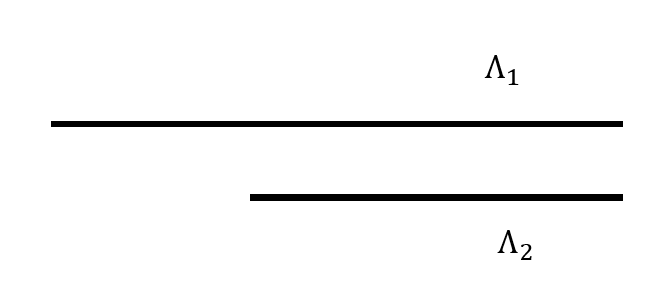} 
\end{center}
\caption{Both intervals, $\Lambda_1 = [a,b]$ and $\Lambda_2$, have lengths greater than $2(2l+2)^2+2l$, with their right endpoints coinciding at the site $b$. The distance between their left endpoints exceeds $2l$. 
}
\label{Fig:exchange}
\end{figure}

{\Proposition Under $\alpha$, the boundary unitary operators $O_L$ and $O_R$ (and similarly for $V_L$ and $V_R$) defined in Prop. \ref{prop:stringtransform} must have well-defined parity. Moreover, there are two consistent transformation rules: 
\begin{itemize} 
\item $\alpha(S_X) = O_L S_X^{-l} O_R$, with $O_L$ and $O_R$ being $\Z_2$ even, and $\alpha(S_Z) = V_L V_R$, with $V_L$ and $V_R$ being $\Z_2$ odd, or 
\item $\alpha(S_X) = O_L O_R$, with $O_L$ and $O_R$ being $\Z_2$ odd, and $\alpha(S_Z) = V_L S_X^{-l} V_R$, with $V_L$ and $V_R$ being $\Z_2$ even, 
\label{prop:Z2auto}
\end{itemize} 
where the form of the transformation is independent of the specific interval on which the string is supported.

}
\begin{proof} Let us first consider the symmetry string $S_X$. Since $\alpha$ is a QCA on the $\Z_2$ symmetric subalgebra, both $O_L$ and $O_R$ must commute with the $\Z_2$ global symmetry up to a phase, with their phases canceling out. Further, note that 
\begin{equation} 
\mathbf{1} = \alpha(S_X) \alpha(S_X) = O_L^2 O_R^2, 
\end{equation} 
which requires the phase to be $\pm 1$. Therefore, $O_L$ and $O_R$ must have well-defined and identical $\Z_2$ parity. Similarly, the same reasoning shows that $V_L$ and $V_R$ must have well-defined, same $\Z_2$ parity. {Since $\alpha$ is a $*$-automorphism of $\B$, the product $O_L O_R$ (and similarly $V_L V_R$) is Hermitian. We adopt the convention that each of these boundary operators is individually Hermitian, which can always be achieved by an opposite phase rotation between $O_L$ and $O_R$.}

To demonstrate that the form of the transformation is a global invariant, consider two symmetry strings arranged as depicted in Fig. \ref{Fig:exchange}, where the right endpoints of the two strings coincide exactly at the same lattice site. Since
\begin{equation}
    \alpha(S_X^{\Lambda_1}) \alpha(S_X^{\Lambda_2}) = \alpha(S_X^{\Lambda_1\setminus \Lambda_2}),
\end{equation}
and this is supported within $(\Lambda_1 \setminus \Lambda_2)^{+l}$ due to the LP nature of $\alpha$, it follows that (1) $S_X^{\Lambda_1}$ and $S_X^{\Lambda_2}$ must have the same transformation form as in Prop. \ref{prop:stringtransform}. Specifically,
\begin{itemize}
    \item $\alpha(S_X^{\Lambda_1}) = O_L S_X(\Lambda_1^{-l}) O_R$ and $\alpha(S_X^{\Lambda_2}) = O_L' S_X(\Lambda_2^{-l}) O_R'$, where $S_X(\Lambda_1^{-l}) = \prod_{i=a+l+1}^{i=b-l-1} X_i$ is the symmetry string supported in the interior of $\Lambda_1$, and similarly for $S_X(\Lambda_2^{-l})$. Here, $O_L$ and $O_R$ are unitaries supported within a distance $l$ from the left and right boundaries of $\Lambda_1$, respectively, and similarly for $O_L'$ and $O_R'$, or
    \item $\alpha(S_X^{\Lambda_1}) = O_L O_R$ and $\alpha(S_X^{\Lambda_2}) = O_L' O_R'$; 
\end{itemize}
and (2) $O_R = O_R'$. By considering two symmetry strings with identical left endpoints and applying the same argument to the charge string, we establish that the transformation rule of string operators is a global invariant determined solely by the QCA $\alpha$, \ie one can arbitrarily expand, shrink, or move the strings, and their transformations under $\alpha$ remain invariant.

\begin{figure}
\begin{center}
  \includegraphics[width=.50\textwidth]{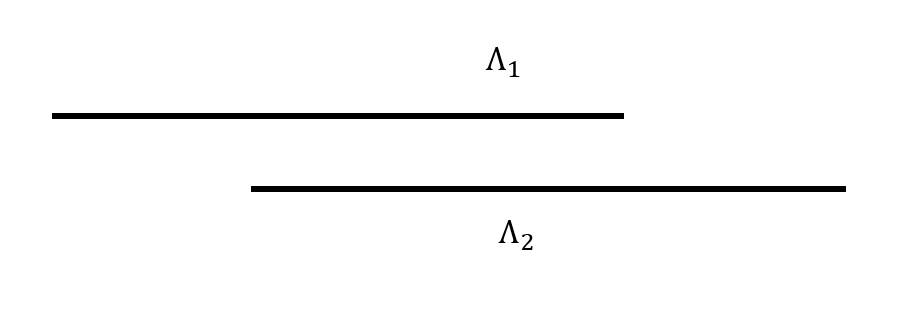} 
\end{center}
\caption{Both intervals, $\Lambda_1$ and $\Lambda_2$, have lengths greater than $2(2l+2)^2+2l$, with $\partial_l \Lambda_1 \cap \partial_l \Lambda_2 = \varnothing$.
}
\label{Fig:braiding}
\end{figure}

Finally, the two consistent charge assignments can be determined by analyzing the commutation relations between strings, which can be categorized into two cases:
\begin{itemize}
    \item For two strings of the same type, \ie either both symmetry strings or both charge strings, we arrange them as shown in Fig. \ref{Fig:exchange}. The fact that $\alpha$ is an automorphism of the symmetric algebra imposes the following constraint:
    \begin{equation}
    \begin{split}
        \alpha(S_X^{\Lambda_1}) \alpha(S_X^{\Lambda_2}) &= \alpha(S_X^{\Lambda_2}) \alpha(S_X^{\Lambda_1}), \\ 
        \alpha(S_Z^{\Lambda_1}) \alpha(S_Z^{\Lambda_2}) &= \alpha(S_Z^{\Lambda_2}) \alpha(S_Z^{\Lambda_1}).
    \end{split}
    \end{equation}
    \item For two strings of distinct types, \ie one symmetry string and one charge string, we arrange them as depicted in Fig. \ref{Fig:braiding}. The fact that $\alpha$ is an automorphism of the symmetric algebra imposes the following constraint: 
    \begin{equation}
        \alpha(S_X^{\Lambda_1}) \alpha(S_Z^{\Lambda_2}) = -\alpha(S_Z^{\Lambda_2}) \alpha(S_X^{\Lambda_1}).
    \end{equation}
\end{itemize}
By combining these two constraints, one arrives at the two consistent transformations stated in the Proposition. 
\end{proof}

In the following section, we will highlight the connection between the transformation rules of strings and anyon permutation symmetries in 2D $G$ gauge theory. For now, simply note that the two consistent transformations are either leaving the string type unchanged or swapping the charge string with a symmetry string, and vice versa. An example of a QCA satisfying the first transformation rule in Prop. \ref{prop:Z2auto} is an FDC with symmetric gates, i.e., each gate commutes with the $\Z_2$ global symmetry. In contrast, the KW duality provides an instance of the second transformation rule in Prop. \ref{prop:Z2auto}. Specifically, we have
\begin{equation}
    \begin{split}
        \mathrm{KW}(\prod_{i=m}^n X_i)& = Z_m Z_{n+1},\\
        \mathrm{KW}(Z_m Z_n) & = \prod_{i=m+1}^n X_i. 
    \end{split}
    \label{eq:KWduality}
\end{equation}

\subsection{The classification}

It is known that for uQCA's in 1D, i.e., QCA's defined on the entire operator algebra $\Al(\Z)$, the complete classification is given by the GNVW index \cite{2009GNVW}. In particular, 1D uQCA's form a group under finite composition, where finite-depth circuits constitute a normal subgroup. The group of 1D uQCA's defined on a spin system, modulo finite-depth circuits, is isomorphic to a group of rational numbers under multiplication: 
\begin{equation}
    \mathrm{uQCA}/\mathrm{FDC} \cong \prod_i p_i^{q_i},
    \label{eq:valueGNVW}
\end{equation}
where the set $\{p_i \}$ includes all prime divisors of the local Hilbert space dimension $d$ (assuming the local Hilbert spaces are uniform across lattice sites), and $\{ q_i \}$ are integers. The physical interpretation of Eq.\eqref{eq:valueGNVW} is that a $p_i$-dimensional subspace at each site is shifted by $q_i$ sites.

In the remainder of this section, we will demonstrate that, similarly, QCA's defined on the $\Z_2$ symmetric subalgebra are completely classified by the index in Def. \ref{def:ind}.

{\theorem QCA's on the $\Z_2$ symmetric subalgebra $\B$ form a group under finite composition, with the subgroup of FDC's with $\Z_2$ symmetric gates being normal. Furthermore, for any QCA $\alpha$, its action on $\B$ can be implemented in one of the following forms:
\begin{itemize} 
\item $\mathrm{T}^q \circ W^\dagger$, where $W$ is an FDC with symmetric gates, $\mathrm{T}$ denotes the lattice translation operator to the right, and $q \in \Z$; 
\item $\mathrm{KW} \circ \mathrm{T}^q \circ W^\dagger$, where $W$ is an FDC with symmetric gates. \end{itemize}
\label{thm:Z2QCA}
}

\begin{proof}
It is straightforward to verify that LP automorphisms of the symmetric algebra form a group under finite composition. We now demonstrate that sFDC's form a normal subgroup. For a given sFDC $W$, we have $\alpha \circ W \circ \alpha^{-1} = \alpha(W)$. Therefore, it suffices to show that $\alpha(W)$ remains an sFDC. This follows directly, as illustrated in Fig. \ref{Fig:sFDC}, where we depict a single layer of gates, with the spread of $\alpha$ taken as 1. Under $\alpha$, each unitary gate is mapped to another local unitary with larger (but still finite) support. The resulting gates can then be rearranged into multiple, but finitely many, layers, with gates in the same layer having disjoint supports. Since the original gates in a single layer commute with each other, the transformed ones do as well, so their order remains irrelevant. As $W$ consists of finitely many such layers, we conclude that $\alpha(W)$ is an sFDC.
\begin{figure}
\begin{center}
  \includegraphics[width=.75\textwidth]{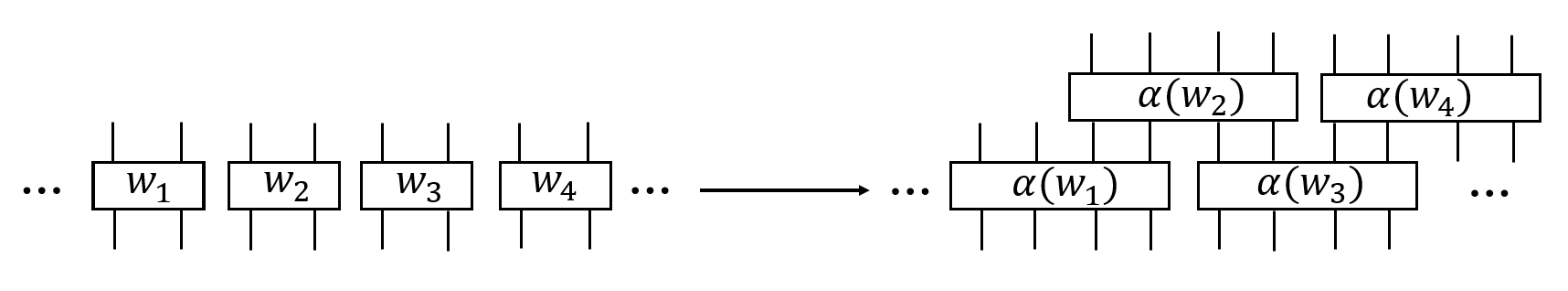} 
\end{center}
\caption{If $W$ is an sFDC, then $\alpha(W)$ is also an sFDC.
}
\label{Fig:sFDC}
\end{figure}

Now we demonstrate the second part of the theorem. The approach involves analyzing the action of the QCA on a subalgebra of $\B$, which is isomorphic to the full operator algebra on a finite interval. The argument is organized into five steps.

\emph{Step 1:} To begin, we coarse-grain the system by grouping every $l$ sites into a new unit cell, effectively reducing the spread of the QCA $\alpha$ to 1. Here, we focus on a subalgebra of $\B$ that is isomorphic to the full operator algebra supported on a finite interval. Specifically, for each single-site $Z$ operator in the unit cells labeled by $2n$ and $2n+1$ in Fig.\ref{Fig:partition} (denoted as $I_{2n}$ and $I_{2n+1}$, respectively), we define an operator
\begin{equation} 
z_i = Z_i Z_m, 
\label{eq:locallyodd} 
\end{equation} where $m - i > 2(2l+2)^2+2l$. The set of all such operators with $i\in I_{2n}\cup I_{2n+1}$ is denoted by $A^o_{2n}$. Additionally, we denote the set of single-site $X$ operators in $I_{2n}\cup I_{2n+1}$ as $A_{2n}$. It is important to note that, unlike operators in $A_{2n}$, operators in $A^o_{2n}$ have a right component fixed on site $m$, located well outside the interval $I_{2n}\cup I_{2n+1}$. Two observations will be crucial in the following discussion. 
\begin{itemize} 
\item The algebra $\B_{2n} = \langle A_{2n}, A^o_{2n} \rangle$ forms a subalgebra of the symmetric subalgebra $\B$, thus having a well-defined transformation under $\alpha$. For each generator of $\B_{2n}$, we can further define a \emph{local} $\Z_2$ parity by examining the portion supported within $I_{2n} \cup I_{2n+1}$: operators in $A_{2n}$ are classified as locally even, while those in $A^o_{2n}$ are classified as \emph{locally} odd. By construction, any operator in $\B_{2n}$ can then be expressed as a linear combination of both locally even and locally odd operators.

\item The algebra $\B_{2n}$, generated by $A_{2n}^o = \{ z_i | i \in I_{2n} \cup I_{2n+1} \}$ and $A_{2n} = \{ X_i | i \in I_{2n} \cup I_{2n+1} \}$, is isomorphic to the algebra of all operators supported within $I_{2n} \cup I_{2n+1}$, denoted by $\mO_{2n}\otimes \mO_{2n+1}$. This isomorphism is given by the map:
\begin{equation}
    \begin{split}
        att: &Z_i \mapsto z_i = Z_i Z_m, \\
        &X_i \mapsto X_i.
    \end{split}
    \label{eq:attach}
\end{equation}
\end{itemize}

We then repeat this construction for nearby unit cells $[I_{2k}, I_{2t+1}]$ ($k+3 < n < t-3$), where $[I_{2k}, I_{2t+1}]$ denotes all sites covered by the unit cells from $I_{2k}$ to $I_{2t+1}$. We require these unit cells to be at a distance greater than $2(2l+2)^2+2l$ from the site $m$. For example, we define $\B_{2n-2} = \langle A_{2n-2}, A^o_{2n-2} \rangle$, where $A_{2n-2} = \{ X_i | i \in I_{2n-2} \cup I_{2n-1} \}$ and $A^o_{2n-2}$ is a set of operators defined by 
\begin{equation} z_i = Z_i Z_m, \, \forall i \in I_{2n-2} \cup I_{2n-1}. \end{equation} 
Importantly, all locally odd operators (i.e., those formed by a product of an odd operator in $[I_{2k}, I_{2t+1}]$ and another odd operator far away) share the same right endpoint $Z_m$. This approach allows us to construct a sequence of mutually commuting algebras, denoted $\B_{2k}, \dots, \B_{2t}$. {By construction, these algebras commute with each other.} It is also straightforward to verify that the algebra generated by this sequence, \begin{equation} 
\mathcal{M} = \langle \B_{2k},\dots, \B_{2t} \rangle, 
\end{equation} 
is isomorphic to the full operator algebra supported on the interval $[I_{2k}, I_{2t+1}]$. In particular, $\mathcal{M}$ contains the $\Z_2$ symmetric subalgebra supported within the interval $[I_{2k}, I_{2t+1}]$.

\emph{Step 2:} We now construct another sequence of algebras, which can be regarded as the ``image" of the algebra sequence from Step 1 under the QCA $\alpha$. Let us consider the transformation of a locally odd operator as defined in Eq.\eqref{eq:locallyodd}. Provided that the distance between $i$ and $m$ is sufficiently large, this operator must transform in accordance with one of the cases in Prop.\ref{prop:Z2auto}. For the moment, we focus on the first case, specifically, 
\begin{equation} 
\alpha(z_i) = V_L V_R, 
\end{equation} 
where $V_L$ and $V_R$ are $\Z_2$-odd unitaries with supports in $[i-l,i+l]$ and $[m-l,m+l]$, respectively, each chosen to be Hermitian. Since all locally odd operators defined in Step 1 share the right endpoint operator $Z_m$, Prop.\ref{prop:Z2auto} implies that the right endpoint component of $\alpha(z_i)$, namely $V_R$, is identical for all these operators. This observation motivates us to define a new set of locally odd operators as \begin{equation} 
z'_i = Z_i V_R, 
\end{equation} 
where the right-hand component is now the transformed operator. Since $V_R$ is $\Z_2$-odd, each $z_i'$ is contained in $\B$. We then define $\mathcal{C}_{2n-1} = \langle A_{2n-1}, A^o_{2n-1} \rangle$, where $A_{2n-1}$ denotes the set of single-site $X$ operators on $I_{2n-1} \cup I_{2n}$, and $A^o_{2n-1} = \{ z'_i  |  i \in I_{2n-1} \cup I_{2n} \}$. Similarly, we repeat this construction for nearby unit cells $[I_{2k-1}, I_{2t+2}]$. For example, we define $\mathcal{C}_{2n+1} = \langle A_{2n+1}, A^o_{2n+1} \rangle$, where $A_{2n+1}$ consists of $X$ operators on $I_{2n+1} \cup I_{2n+2}$, and $A^o_{2n+1} = \{ z'_i  |  i \in I_{2n+1} \cup I_{2n+2} \}$. This procedure generates a new sequence of algebras, denoted by $\mathcal{C}_{2k-1},\dots ,\mathcal{C}_{2t+1}$.

Again, by construction, the algebras within the sequence $\mathcal{C}_{2k-1}, \dots, \mathcal{C}_{2t+1}$ have the following properties: 
(1) they commute with each other; and (2) the algebra generated by this sequence, \begin{equation} 
\mathcal{N} = \langle \mathcal{C}_{2k-1}, \dots, \mathcal{C}_{2t+1} \rangle, \end{equation} 
is isomorphic to the algebra of all operators supported within the interval $[I_{2k-1}, I_{2t+2}]$, and, in particular, contains the subalgebra of $\B$ that is fully supported within $[I_{2k-1}, I_{2t+2}]$.

\emph{Step 3:} In this step, we analyze the transformation of $\B_{2n}$ under the QCA $\alpha$. Let us begin by selecting an element $b \in \B_{2n}$. Under the action of the QCA, $\alpha(b)$ clearly lies within $\mathcal{C}_{2n-1} \otimes \mathcal{C}_{2n+1}$, since $b$ can be expressed as a linear combination of locally even and locally odd operators, all of which are mapped into $\mathcal{C}_{2n-1} \otimes \mathcal{C}_{2n+1}$.

To proceed further, we introduce the concept of a support algebra \cite{2000Zanardi,2009GNVW}. Generally, given a subalgebra $\Al \subset \B_1 \otimes \B_2$ of a tensor product of two finite-dimensional algebras, and selecting linearly independent bases $\{ a_\nu \}$ for $\Al$ and $\{ b^2_\mu \}$ for $\B_2$, each element of the basis $\{ a_\nu \}$ can be uniquely expanded as \begin{equation} 
a_\nu = \sum_\mu b^{1}_{\nu\mu} \otimes b^2_{\mu}, 
\label{eq:supp0} 
\end{equation} 
with $b^{1}_{\nu\mu} \in \B_1$. The support algebra of $\Al$ in $\B_1$, denoted $\sa(\Al, \B_1)$, is the algebra generated by the set $\{ b^1_{\nu\mu} \}$ arising in this expansion, where $\nu$ ranges over all basis elements of $\Al$. This support algebra has two useful properties \cite{2009GNVW}: (1) $\sa(\Al, \B_1)$, as constructed, is independent of the choice of basis, and is the smallest subalgebra $\mathcal{C}$ of $\B_1$ such that $\Al \subset \mathcal{C} \otimes \B_2$; and (2) if subalgebras $\Al_1 \subset \B_1 \otimes \B_2$ and $\Al_2 \subset \B_2 \otimes \B_3$ commute within $\B_1 \otimes \B_2 \otimes \B_3$, then $\sa(\Al_1, \B_2)$ and $\sa(\Al_2, \B_2)$ commute within $\B_2$.

Given that $\alpha(\B_{2n}) \subset \mathcal{C}_{2n-1} \otimes \mathcal{C}_{2n+1}$, we can define the support algebras of $\alpha(\B_{2n})$ within $\mathcal{C}_{2n-1}$ and $\mathcal{C}_{2n+1}$, respectively: 
\begin{equation} 
\begin{split} \Le_{2n} & = \sa(\alpha(\B_{2n}), \mathcal{C}_{2n-1}), \\ 
\Ri_{2n} & = \sa(\alpha(\B_{2n}), \mathcal{C}_{2n+1}), 
\end{split} 
\end{equation} for $n \in [k, t]$. A geometric illustration of these algebras is provided in Fig \ref{Fig:partition}. Note that all support algebras thus obtained are mutually commuting, which follows from the mutual commutativity of distinct $\B_{2n}$'s and property (2) from the preceding paragraph (needed to show that $\mathcal{L}_{2n}$ and $\mathcal{R}_{2n-2}$ commute). Moreover, since $\B_{2n}$ is closed under Hermitian conjugation (adjoint) and $\alpha$ preserves this operation, the support algebras are also closed under taking the adjoint.

\begin{figure}
\begin{center}
  \includegraphics[width=.45\textwidth]{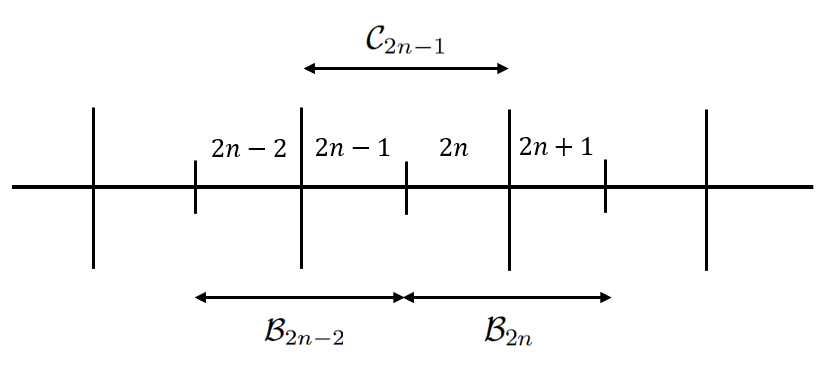} 
\end{center}
\caption{The unit cell structure of the chain: The locally even piece and the left portion of the locally odd piece in $\B_{2n}$ are supported within $I_{2n} \cup I_{2n+1}$.
}
\label{Fig:partition}
\end{figure}

We now show that all support algebras $\Le_{2n}$ and $\Ri_{2n}$ for $n\in[k+1,t-1]$ are isomorphic to matrix algebras, and having a basis in which each element has a well-defined local $\Z_2$ parity: in mathematical terms, these algebras are $\Z_2$ graded algebras \cite{knus1969algebras}. The argument can be divided into four parts:
\begin{itemize}
    \item We first prove that 
    \begin{equation} \alpha(\B_{2n}) = \Le_{2n} \otimes \Ri_{2n}, \, n\in[k+1,t-1]. 
    \label{eq:supp1} 
    \end{equation} 
By the definition of support algebras, we have $\alpha(\B_{2n}) \subset \Le_{2n} \otimes \Ri_{2n}$. If $\alpha(\B_{2n})$ were a strict subset, \ie $\alpha(\B_{2n}) \subsetneq \Le_{2n} \otimes \Ri_{2n}$, then there would exist a non-identity element $\epsilon \in \Le_{2n} \otimes \Ri_{2n}$ in the relative commutant of $\alpha(\B_{2n})$\footnote{Because $\alpha(\B_{2n})$ and $\Le_{2n} \otimes \Ri_{2n}$ are both closed under taking adjoint, they are semisimple and, by the Wedderburn-Artin theorem, are isomorphic to direct sums of matrix algebras. Consequently, if $\alpha(\B_{2n}) \subsetneq \Le_{2n} \otimes \Ri_{2n}$, the relative commutant of $\alpha(\B_{2n})$ must contain non-identity elements.} that commutes with $\alpha(\B_{2j})$ for all $j\in [k,t]$. In fact, $\alpha(\B_{2j})$ with $j\neq n$ lies within other support algebras and therefore commutes with $\Le_{2n} \otimes \Ri_{2n}$. Consequently, $\alpha^{-1}(\epsilon)$ commutes with the algebra $\mathcal{M}$. For $n \in [k+1, t-1]$, it also holds that $\alpha^{-1}(\epsilon) \in \mathcal{M}$. This implies that $\alpha^{-1}(\epsilon)$, and hence $\epsilon$, must be the identity, as $\mathcal{M}$, being a matrix algebra, has no nontrivial center \cite{bratteli2012operator}. This establishes Eq.\eqref{eq:supp1}.

\item We observe that 
\begin{equation} 
\mathcal{C}_{2n-1} = \Ri_{2n-2} \otimes \Le_{2n}, 
\label{eq:supp2} 
\end{equation} 
for $n\in [k+1,t-1]$. Clearly, $\Ri_{2n-2} \otimes \Le_{2n} \subset \mathcal{C}_{2n-1}$. If $\Ri_{2n-2} \otimes \Le_{2n}$ were a strict subset, \ie $\Ri_{2n-2} \otimes \Le_{2n} \subsetneq \mathcal{C}_{2n-1}$, then there would be a non-identity element $\epsilon \in \mathcal{C}_{2n-1}$ that commutes with all support algebras. This implies that $\alpha^{-1}(\epsilon)\in\mathcal{M}$ commutes with the algebra $\mathcal{M}$, which is isomorphic to a matrix algebra. However, this is impossible by the same argument as in the previous paragraph.

\item The support algebras $\Le_{2n}$ and $\Ri_{2n}$ for $n \in [k+1, t-1]$ are each isomorphic to a matrix algebra. This follows because, for these values of $n$, we see from Eq.\eqref{eq:supp1} that 
\begin{equation} 
 \B_{2n} = \alpha^{-1}(\Le_{2n} \otimes \Ri_{2n}) = \alpha^{-1}(\Le_{2n}) \otimes \alpha^{-1}(\Ri_{2n}). 
\label{eq:supp3}
\end{equation} 
If any one of these support algebras for $n \in [k+1, t-1]$ had a non-trivial center, then $\B_{2n}$ would also have a non-trivial center. This is not possible since $\B_{2n}$ is isomorphic to a matrix algebra, \ie the algebra of all operators on $[I_{2n}, I_{2n+1}]$. Thus, each of $\Le_{2n}$ and $\Ri_{2n}$ for $n \in [k+1, t-1]$ must have a trivial center and therefore be isomorphic to a matrix algebra (by the Wedderburn-Artin Theorem \cite{bratteli2012operator}).

\item Finally, we demonstrate that each support algebra is a $\Z_2$-graded algebra, \ie each has a basis in which every element possesses a well-defined local parity. This is straightforward: in Eq.\eqref{eq:supp0}, we can select the basis $\{ a_\nu \}$ for $\alpha(\B_{2n})$ and $\{ b_\mu^2 \}$ for $\mathcal{C}_{2n\pm 1}$ such that each element in both bases has a well-defined local parity. Therefore, each element in $\{ b_{\nu\mu}^1 \}$ also has a well-defined parity. Consequently, the support algebra $\sa(\alpha(\B_{2n}), \mathcal{C}_{2n\mp 1})$ possesses a basis in which each element exhibits a specific local parity. By the construction of $\mathcal{C}_{2n\pm 1}$, we see that each $\Le_{2n}$ has a basis $\{ M_{2n}^{L,\mu}, N_{2n}^{L,\sigma}\otimes V_R \}$, where $\{ M_{2n}^{L,\mu} \}$ represents the locally even operators, fully supported within $I_{2n-1} \cup I_{2n}$, and $\{ N_{2n}^{L,\sigma}\otimes V_R \}$ represents the locally odd operators, with $N_{2n}^{L,\sigma}$ being $\Z_2$ odd and supported within $I_{2n-1} \cup I_{2n}$. A similar argument applies to $\Ri_{2n}$, for which we obtain a basis $\{ M_{2n}^{R,\mu}, N_{2n}^{R,\sigma}\otimes V_R \}$.

\end{itemize}

Combining these results: from Eqs. \eqref{eq:supp1} and \eqref{eq:supp2}, we find that for $n \in [k+1, t-1]$, \begin{equation} 
\begin{split} 
\dim(\B_{2n}) & = \dim(\Le_{2n}) \times \dim(\Ri_{2n}), \\ 
\dim(\mathcal{C}_{2n-1}) & = \dim(\Le_{2n}) \times \dim(\Ri_{2n-2}), \end{split} 
\end{equation} 
which implies that $\{ \Le_{2n} \}$ are isomorphic to $2^{d_L} \times 2^{d_L}$ matrix algebras, and $\{ \Ri_{2k} \}$ are isomorphic to $2^{d_R} \times 2^{d_R}$ matrix algebras, with $d_L + d_R = 2l$. Notably, the dimensions of the matrix algebras, $d_L$, are identical for all $n \in [k+1, t-1]$, and the same holds for $d_R$. 

\emph{Step 4:} In the fourth step, we demonstrate that a QCA with the first type of transformation rule, as described in Prop. \ref{prop:Z2auto}, acts on the symmetric subalgebra, up to an sFDC, as a translation. We begin by constructing local gates supported within $[I_{2k}, I_{2t+1}]$, collectively denoted as $W_\Lambda$, which satisfy 
\begin{equation} 
\begin{split} 
&W_\Lambda \circ \mathrm{T}^{-q} \circ \alpha(X_i) = X_i, \\ 
&W_\Lambda \circ \mathrm{T}^{-q} \circ \alpha(Z_i Z_m) = Z_i V_R', 
\end{split} 
\label{eq:recover} 
\end{equation} 
where $i \in [I_{2k+2}, I_{2t-1}]$, and $V_R'$ is an odd operator supported near the distant site $m$, which is both unitary and Hermitian. Eq. \eqref{eq:recover} implies that the standard generating set of $\B$ within $[I_{2k+2}, I_{2t-1}]$, and hence its subalgebra supported on this interval, is invariant under $W_\Lambda \circ \mathrm{T}^{-q} \circ \alpha$. By extending this construction to an FDC $W$ over the entire chain, we conclude that $\alpha$, when acting on the symmetric subalgebra, is equivalent to \begin{equation} 
\alpha(\B) = \mathrm{T}^q \circ W^\dagger(\B). 
\end{equation}

To begin, we divide each interval $I_{2n-1} \cup I_{2n}$, with $n \in [k+1, t-1]$, into two adjacent segments: the left segment containing $d_R$ sites, and the right segment containing $2l - d_R = d_L$ sites. We denote the algebra of all operators on the left $d_R$ sites as $\mathcal{T}_{2n-2}^R$ and on the right $d_L$ sites as $\mathcal{T}_{2n}^L$, corresponding to $2^{d_R} \times 2^{d_R}$ and $2^{d_L} \times 2^{d_L}$ matrix algebras, respectively. We demonstrate below that there exists a unitary $W_{2n-1}$, supported entirely within $I_{2n-1} \cup I_{2n}$, which maps the support of $\Le_{2n}$ and $\Ri_{2n-2}$ onto $\mathcal{T}_{2n}^L$ and $\mathcal{T}_{2n-2}^R$, respectively. Here, for locally odd operators in $\Le_{2n}$ and $\Ri_{2n-2}$ (those with a right component at a distant site), we refer only to the support of the left component.

Indeed, there is an algebra isomorphism from $\Le_{2n}$ to another algebra, $\tilde \Le_{2n} = \langle M_{2n}^{L, \mu}, N_{2n}^{L, \sigma} \rangle$, via the following transformation $\rho$ on the basis elements: 
\begin{equation} 
\begin{split} \rho: & M_{2n}^{L, \mu}  \mapsto M_{2n}^{L, \mu} \\ & N_{2n}^{L, \sigma} \otimes V_R  \mapsto N_{2n}^{L, \sigma}, 
\end{split} 
\label{eq:truncateend} 
\end{equation} 
thus making $\tilde\Le_{2n}$ isomorphic to a $2^{d_L} \times 2^{d_L}$ matrix algebra. More precisely, $\rho$ is an isomorphism from $\mathcal{C}_{2n-1}=(\mO_{2n-1} \otimes \mO_{2n})^e \oplus [(\mO_{2n-1} \otimes \mO_{2n})^o \otimes V_R]$ to $\mO_{2n-1} \otimes \mO_{2n}$ for all $n \in [k+1, t-1]$, where $\mO^{e(o)}$ denotes the $\Z_2$ even (odd) component of an algebra $\mO$; specifically, it acts by truncating the pattern $V_R$, an order 2 unitary, from the $\Z_2$ odd component.\footnote{The transformation $\rho^{-1}$ differs from $att$ as defined in Eq.\eqref{eq:attach} only by the operator attached: for $att$, the attached operator is $Z_m$, while for $\rho^{-1}$, it is $V_R$, which is also an order-2 unitary localized near site $m$.} A similar argument applies to obtain $\tilde \Ri_{2n-2} = \langle M_{2n-2}^{R, \mu}, N_{2n-2}^{R, \sigma} \rangle$, which is also isomorphic to a $2^{d_R} \times 2^{d_R}$ matrix algebra.

An observation follows: the algebras $\tilde\Le_{2n}$ and $\tilde\Ri_{2n-2}$ are now fully supported within $I_{2n-1} \cup I_{2n}$, respectively. Since $\mathcal{T}_{2n}^L$ and $\tilde\Le_{2n}$, as well as $\mathcal{T}_{2n-2}^R$ and $\tilde\Ri_{2n-2}$, are matrix algebras of identical dimensions, there exists a local unitary $W_{2n-1}$, supported on $I_{2n-1} \cup I_{2n}$, such that \begin{equation} 
\begin{split} 
W_{2n-1} \tilde\Ri_{2n-2} W_{2n-1}^\dagger &= \mathcal{T}_{2n-2}^R, \\ 
W_{2n-1} \tilde\Le_{2n} W_{2n-1}^\dagger &= \mathcal{T}_{2n}^L, 
\end{split} 
\end{equation} 
as desired. { Specifically, (1) since $\tilde\Le_{2n}$ and $\mathcal{T}_{2n}^L$ are two representations of the same finite-dimensional matrix algebra on the Hilbert space $I_{2n-1}\cup I_{2n}$, they differ only by a unitary change of basis $W_{2n-1}$ due to the Skolem–Noether theorem; (2) under the identification $\mO_{2n-1}\otimes\mO_{2n} \cong \tilde\Le_{2n}\otimes \tilde\Ri_{2n-2} $, the commutant of $\tilde\Le_{2n}$ within $\mO_{2n-1}\otimes\mO_{2n}$ is $\tilde\Ri_{2n-2}$, and similarly the commutant of $\mathcal{T}_{2n}^L$ in $\mO_{2n-1}\otimes\mO_{2n}$ is $\mathcal{T}_{2n-2}^R$. Therefore
\begin{equation}
    W_{2n-1} \tilde\Ri_{2n-2} W_{2n-1}^\dagger = W_{2n-1} (\tilde\Le_{2n})'  W_{2n-1}^\dagger = (W_{2n-1} \tilde\Le_{2n}  W_{2n-1}^\dagger)'= (\mathcal{T}_{2n}^L)' = \mathcal{T}_{2n-2}^R,
\end{equation}
so the same unitary $W_{2n-1}$ simultaneously send $\tilde \Ri_{2n-2}$ to $\mathcal{T}_{2n-2}^R$}. This construction applies to all $n\in [k+1,t-1]$.

We then perform a translation to the right by $-q = d_L - l$ sites (or to the left if $-q < 0$), which shifts the algebras $\mathcal{T}_{2n}^L \otimes \mathcal{T}_{2n}^R$ back into the interval $I_{2n} \cup I_{2n+1}$. Due to Eq.\eqref{eq:supp1}, we have 
\begin{equation} 
\mathcal{T}_{2n}^L \otimes \mathcal{T}_{2n}^R \cong \tilde\Le_{2n} \otimes \tilde\Ri_{2n} \cong \Le_{2n} \otimes \Ri_{2n} \cong \B_{2n} \cong \mathcal{O}_{2n} \otimes \mathcal{O}_{2n+1}, 
\label{eq:chain} 
\end{equation} 
where $\mathcal{O}_{2n} \otimes \mathcal{O}_{2n+1}$ denotes the algebra of all operators on $[I_{2n}, I_{2n+1}]$, and each tensor product in Eq.\eqref{eq:chain} is a matrix algebra. Consequently, the transformation \begin{equation} \mathrm{T}^{-q} \circ W_{2n-1} \circ W_{2n+1} \circ \rho \circ \alpha \circ att \end{equation} 
restricts to an automorphism of $\mO_{2n} \otimes \mO_{2n+1}$, which must be inner, as it acts on a matrix algebra (by the Skolem-Noether's Theorem), and can therefore be implemented by a local unitary $Q_{2n}^\dagger \in \mO_{2n} \otimes \mO_{2n+1}$ for $n \in [k+1, t-1]$. Thus, we conclude that
\begin{equation} 
W_\Lambda \circ \mathrm{T}^{-q} \circ \rho \circ \alpha \circ att 
\label{eq:recoverlocal} 
\end{equation} acts as the identity on $\otimes_{n=k+1}^{t-1} (\mO_{2n}\otimes \mO_{2n+1})$, where 
\begin{equation} 
W_\Lambda = \left(\prod_{n=k+1}^{t-1} Q_{2n}\right) \otimes \mathrm{T}^{-q} \left(\prod_{n=k+1}^{t-1} W_{2n-1}\right) \mathrm{T}^{q} 
\end{equation} is a unitary operator acting non-trivially within $[I_{2k}, I_{2t+1}]$.

Finally, Eq.\eqref{eq:recoverlocal} indicates that 
\begin{equation} 
\begin{split} 
&\alpha(Z_i Z_m) = \rho^{-1} \circ \mathrm{T}^q \circ W_\Lambda^\dagger(Z_i), \\ 
&\alpha(X_i) = \rho^{-1} \circ \mathrm{T}^q \circ W_\Lambda^\dagger(X_i). 
\end{split} 
\label{eq:FDCconstruction1}
\end{equation} 
Since $\mathrm{T}^q \circ W^\dagger_\Lambda$ is a uQCA, matching the supports of operators on both sides of these equalities, we must have \begin{equation} 
\begin{split} 
&\alpha(Z_i Z_m) = [\mathrm{T}^q \circ W_\Lambda^\dagger(Z_i)] \otimes V_R, \\ &\alpha(X_i) = \mathrm{T}^q \circ W_\Lambda^\dagger(X_i), 
\end{split} 
\label{eq:FDCconstruction2}
\end{equation} 
which directly leads to Eq.\eqref{eq:recover}, where $V_R' = \mathrm{T}^{-q} V_R \mathrm{T}^q$.

\emph{Step 5:} We extend the construction in the previous step across the entire system, obtaining an FDC acting on the chain given by \begin{equation} 
W = (\prod_n Q_{2n}) \otimes \mathrm{T}^{-q} (\prod_n W_{2n-1}) \mathrm{T}^{q}. 
\end{equation} 
As in Eq.\eqref{eq:recover}, the transformation $W \circ \mathrm{T}^{-q} \circ \alpha$ leaves the standard generating set of $\B$ invariant, thus acting as the identity on $\B$. We now demonstrate that $W$, as constructed, must be an FDC with symmetric gates. Since 
\begin{equation} 
\alpha(O) = \mathrm{T}^q \circ W^\dagger(O),\quad \forall O\in \B, 
\label{eq:WsymFDC}
\end{equation} 
we know that $\mathrm{T}^q \circ W^\dagger$ is a uQCA that also acts on the symmetric subalgebra as an automorphism. Generally, if a uQCA implements an automorphism of a subalgebra symmetric under a finite Abelian group $G$, it must map a global symmetry generator $g$ to another global symmetry generator $h$ (since these are the only operators that commute with the entire symmetric subalgebra $\B$). Furthermore, the group composition law must be preserved, \ie it acts as an automorphism of the symmetry group $G$, which is trivial in the case of $\Z_2$. Thus, $\mathrm{T}^q \circ W^\dagger$ is a uQCA that commutes with the $\Z_2$ symmetry (see Remark 1 below), and the $\Z_2$ symmetric FDC $W$ is classified by $H^2(\Z_2, \U)$ \cite{2018sMPU,2013Hastings,2016else}, which physically labels the 1D SPT it can entangle. For $G=\Z_2$, this cohomology group is trivial. According to Theorem \ref{thm:GsymFDC}, any such symmetric FDC (\ie a FDC that commutes with the symmetry as a whole, that entangles a trivial SPT) can always be written as an FDC with symmetric local gates. This completes the proof for the first transformation rule listed in Prop.\ref{prop:Z2auto}: its action on $\B$ is identical to $\mathrm{T}^q \circ W^\dagger$, where $W$ is an sFDC.

For a QCA $\alpha$ on the $\Z_2$ symmetric subalgebra $\B$ with the second type of transformation rule, \ie it swaps the symmetry string with the charge string, we can always combine it with a KW transformation as defined in Eq.\eqref{eq:KWduality}, such that $\mathrm{KW}^{-1} \circ \alpha$ now has a transformation rule of the first type in Prop.\ref{prop:Z2auto}. By the preceding discussion, we have $\mathrm{KW}^{-1} \circ \alpha(\B) = \mathrm{T}^q \circ W^\dagger (\B)$, which implies $\alpha(\B) = \mathrm{KW} \circ \mathrm{T}^q \circ W^\dagger (\B)$, where $W$ is a $\Z_2$ sFDC. This completes the proof of the theorem.
\end{proof}

\textbf{Remark 1:} In principle, it is possible that the uQCA $\mathrm{T}^q \circ W$, as constructed in the proof of Theorem \ref{thm:Z2QCA}, commutes with the $\Z_2$ symmetry up to a phase, representing the $\Z_2$ charge of this uQCA. In this case, we can compose it with a single-site operator $Z_i$, such that the resulting uQCA, \ie $\mathrm{T}^q \circ W \circ Z_i$, is $\Z_2$ symmetric. The additional $Z_i$ operator only introduces a minus sign to one local generator, $X_i$, in $\B$, which we ignore for simplicity.

 Since $H^2(\Z_2, \U)=\Z_1$, following Theorem \ref{thm:GsymFDC}, we will hereafter use the terms symmetric FDC and FDC with symmetric gates interchangeably. Given Theorem \ref{thm:Z2QCA}, which serves as a structural result for QCAs on the $\Z_2$ symmetric subalgebra, we can characterize their equivalence classes.

{\theorem \begin{enumerate}
    \item The index of a QCA defined on the $\Z_2$ symmetric subalgebra can take values within 
\begin{equation} 
\ind (\alpha)\in \big\{ (\sqrt{2})^q 2^p \mid p \in \Z,  q=0,1 \big\}. 
\end{equation} 
Moreover, for two QCA's defined on $\B$, the index is multiplicative under composition, \ie \begin{equation} 
\ind(\alpha\circ \beta) = \ind(\alpha)\ind(\beta). 
\end{equation} 
\item 
QCA's defined on the $\Z_2$ symmetric subalgebra are completely classified by the index, up to FDC's with symmetric gates. 
\end{enumerate}

\label{thm:Z2completeclassification}
}

\begin{proof} By Theorem \ref{thm:Z2QCA}, we know that the index of any QCA defined on $\B$ is equal to the index of $\mathrm{KW}^q \circ \mathrm{T}^p$, where $q \in \{ 0,1 \}$ and $p \in \Z$. This follows from Corollary \ref{cor:FDCinvariant}, which states that the index remains invariant under an FDC with symmetric gates. Additionally, since both $\mathrm{KW}$ and $\mathrm{T}$ map elements of the standard generating set of $\B$ to other elements in the same set, calculating their index is straightforward: $\ind (\mathrm{KW} \circ \mathrm{T}^p) = \sqrt{2}^q 2^p$, establishing the possible values of the index.

For two QCA's $\alpha$ and $\beta$ on $\B$, we have $\alpha(\B) = \mathrm{KW}^{q_1} \circ \mathrm{T}^{p_1} \circ W_1$ and $\beta(\B) = \mathrm{KW}^{q_2} \circ \mathrm{T}^{p_2} \circ W_2$, where $W_1$ and $W_2$ are FDC's with symmetric gates. This implies \begin{equation} \alpha \circ \beta (\B) = \mathrm{KW}^{q_1 + q_2} \circ \mathrm{T}^{p_1 + p_2} \circ W', \end{equation} where $W' = (\mathrm{KW}^{q_2} \circ \mathrm{T}^{p_2})^{-1} W_1 (\mathrm{KW}^{q_2} \circ \mathrm{T}^{p_2}) \circ W_2$ is also an FDC with symmetric gates, according to Theorem \ref{thm:Z2QCA}, as FDC's with symmetric gates form a normal subgroup of QCA's on $\B$. Thus, we find that \begin{equation} 
\ind(\alpha \circ \beta) = (\sqrt{2})^{q_1 + q_2} 2^{p_1 + p_2} = \ind(\alpha)\ind(\beta). 
\end{equation}

Finally, for a QCA $\alpha$ with $\ind(\alpha) = \sqrt{2}^q 2^p$, where $q \in \{ 0,1 \}$ and $p \in \Z$, Theorem \ref{thm:Z2QCA} implies that $\alpha$ must act on $\B$ as $\alpha(\B) = \mathrm{KW}^q \circ \mathrm{T}^p \circ W$, where $W$ is an FDC with symmetric gates. This establishes statement (2) in the theorem. 
\end{proof}

Theorem \ref{thm:Z2completeclassification} implies the existence of a homomorphism from the group of QCAs on the $\Z_2$ symmetric algebra to a $\Z_2$ invariant that measures whether the index is rational. The kernel of this homomorphism acts on any $O \in \B$ identically to a uQCA and can thus be naturally extended to the full operator algebra by defining $\alpha(\Al(\Z)) = \mathrm{T}^q \circ W^\dagger (\Al(\Z))$. In contrast, QCAs with an irrational index swap the two types of string operators and cannot be extended to QCAs acting on the full operator algebra. Readers may also recall the index theory for QCAs on fermionic chains, where the same irrational index, \ie $\sqrt{2}$, can arise \cite{2019fQCA,2017fQCA,2020fMPU}. The fermionic QCA exhibiting this irrational index corresponds to the Majorana translation. In Appendix \ref{app:fermionize}, we illustrate a more direct connection between the QCA on the $\Z_2$ symmetric subalgebra and a fermionic QCA.

We include an additional remark before concluding this section:

\textbf{Remark 2:} In the classification of uQCA with symmetries, it is common to consider stable equivalence \cite{2017MPUCirac,2018sMPU,2019fQCA}, where one is permitted to add ancillas in nontrivial (potentially arbitrary) representations of the symmetry. However, this approach appears unnatural in our context: unlike the full local operator algebra, the symmetric subalgebra does not factorize as a tensor product between the system and the ancilla. Therefore, in this work, we fix the onsite representation to be regular and do not incorporate additional ancillas.

\section{General finite Abelian group}
\label{sec:finiteabelian}
We now extend the preceding discussion to the subalgebra that is symmetric under a general finite Abelian group $G$, which can always be decomposed as a direct product of cyclic groups. We begin by examining the transformation rules for symmetry and charge strings, analogous to Prop.\ref{prop:stringtransform}:

{\Proposition Let $\alpha$ be a QCA on the $G$ symmetric subalgebra. We have:

\begin{enumerate}
    \item  The transformation of a symmetry string $S_g = \prod_{i=m}^n u_i(g)$ supported on $\Lambda = [m,n]$, where $n - m > 2(2l+2)^2+2l$, with $u_i(g)$ as the on-site regular representation of an element $g \in G$, is given by $\alpha(S_g) = O_L(g) S_{\sigma(g)}^{-l} O_R^\dagger(g)$, where: \begin{itemize} 
\item $S_{\sigma(g)}^{-l} = \prod_{i=m+l+1}^{n-l-1} u_i[\sigma(g)]$ represents a symmetry string supported on $\Lambda^{-l}$, with $\sigma$ as an endomorphism of $G;$ 
\item $O_L(g)$ and $O_R(g)$ are unitaries that satisfy $O_L(g)^{|g|}= O_R(g)^{|g|} = 1$ (with $|g|$ denoting the order of an element $g \in G$) supported within a distance $l$ of site $m$ and $n$, respectively. They both carry the same, well-defined charge under $G$, determined by a homomorphism $\mu: G \mapsto \hat{G} = H^1(G, U(1))$. 
\end{itemize} 
\item The transformation of a charge string $S_{\hat{g}} = Z_m^{\hat{g}} (Z_n^{\hat{g}})^\dagger$, where $Z_m^{\hat{g}}$ denotes a single-site unitary charged operator associated with the character $\hat{g} \in H^1(G, U(1))$, is given by $\alpha(S_{\hat{g}}) = V_L(\hat{g}) S_{\gamma(\hat{g})}^{-l} V^\dagger_R(\hat{g})$, where $\gamma: \hat{G} \mapsto G$ is a homomorphism. Here, $V_L(\hat{g})$ and $V_R(\hat{g})$ are unitaries that satisfy $V_L(\hat{g})^{|\hat{g}|} = V_R(\hat{g})^{|\hat{g}|} =1$, supported within a distance $l$ of $m$ and $n$, respectively, each carrying the same charge under $G$, determined by an endomorphism $\nu: \hat{G} \mapsto \hat{G}$.

\item The transformation rules for the strings, specifically $\sigma$, $\mu$, $\gamma$, and $\nu$, are global invariants for a QCA $\alpha$. That is, they remain independent of the particular choice of intervals supporting the strings, provided the interval length is sufficiently large.

\end{enumerate}
\label{prop:Gtransform}
}

\begin{proof}
By the same reasoning that leads to Eq. \eqref{string}, we find 
\begin{equation} 
\alpha(S_g) = \sum_{h \in G} O_h \otimes S_h^{-l}, 
\end{equation} 
where $S_h$ is the symmetry generator associated with $h \in G$, restricted to the interval $\Lambda^{-l}$, and $\{ O_h \}$ are operators supported within $\partial_l \Lambda$. Since $\alpha(S_g)$ is an FDC with a spread upper bounded by $(2l+1)^2$, analogous to Eq. \eqref{eq:stringtransform}, we have \begin{equation} 
\sum_h \chi(h) O_h O_h^\dagger = e^{i\theta_\chi}\mathbf{1}, \end{equation} 
for any character $\chi \in \hat{G} = H^1(G,\U)$, with $e^{i\theta_\chi}$ a pure phase. The orthogonality of characters, \ie $\sum_{\chi\in \hat{G}} \chi(g)^* \chi(h) = |G|\delta_{g,h}$, where $|G|$ is the order of the group, implies that each $O_h O_h^\dagger$ must be a real, non-negative $c$-number. Consequently, one of $\{ O_h \}$ must be a unitary operator, with all others being zero. The fact that $\alpha(S_g)$ is an FDC with spread $(2l+1)^2$ further implies that the nonzero operator factorizes as a product of two unitaries, supported near $m$ and $n$ respectively, leading to 
\begin{equation} 
\alpha(S_g) = O_L(g) S_{\sigma(g)}^{-l} O_R^\dagger(g), 
\end{equation} 
where $\sigma: G \mapsto G$ is an endomorphism since $\alpha$ is an automorphism of the symmetric algebra, \ie $\alpha(S_g) \alpha(S_h) = \alpha(S_g S_h) = \alpha(S_{gh})$. Furthermore, since $\alpha(S_g)^{|g|} = 1$, it follows that $O_L(g)^{|g|} [O_R(g)^\dagger]^{|g|} = 1$. By convention, we can choose $O_L(g)^{|g|} = O_R(g)^{|g|} = 1$ individually.

Now we demonstrate that $O_L(g)$ and $O_R(g)$ carry an identical charge under $G$, determined by a homomorphism $\mu: G\mapsto \hat{G}$. Since $\alpha(S_g)$ lies within the symmetric algebra, it must commute with each element $h \in G$. Consequently, each endpoint unitary operator must commute with $h$ up to a $\U$ phase factor, 
\begin{equation} O_L^\dagger(g) h O_L(g) h^\dagger = \mu_g(h) \in \U, \label{eq:character} \end{equation}
where the phase factor obtained from $O_L(g)$ and $O_R(g)$ must be identical. Clearly, $\mu_g(\cdot)$ is a character of $G$, as, by the definition in Eq.\eqref{eq:character}, we have $\mu_g(kh) = \mu_g(k)\mu_g(h)$ for $k, h \in G$.

Furthermore, $\mu$ defines a homomorphism from $G$ to its character group $\hat{G} = H^1(G, \U)$. This follows from the fact that $\alpha$ is an automorphism of the symmetric algebra, which implies 
\begin{equation} 
O_L(g) O_L(h) = f(g,h) O_L(gh), 
\end{equation} 
with $f: G \times G \mapsto \U$ being a phase factor.\footnote{ In fact, $f$ is a group 2-cocycle due to the relation
\begin{equation}
    [O_L(g)O_L(h)]O_L(k) = O_L(g) [ O_L(h) O_L(k)], \, \forall g,h,k\in G.
\end{equation}
} This leads directly to 
\begin{equation} \mu_g(\cdot) \mu_h(\cdot) = \mu_{gh}(\cdot), \end{equation} 
as required. This completes the proof of part (1) of the Proposition. The same argument applies to charge strings $S_{\hat{g}} = Z_m^{\hat{g}} (Z_n^{\hat{g}})^\dagger$ without further modification, thus completing the proof of part (2) of the Proposition. 

To demonstrate part (3), consider two strings of the same type, \eg two symmetry strings associated with the same element $g \in G$, arranged in the geometry shown in Fig. \ref{Fig:exchange}. We denote these strings as $S_{\Lambda_1}$ and $S_{\Lambda_2}$. Since $\alpha$ is a QCA on $\B$, we have \begin{equation} \alpha(S_{\Lambda_1})\alpha(S_{\Lambda_2}^\dagger) = \alpha(S_{\Lambda_1\setminus \Lambda_2}), \end{equation} 
which is supported within $(\Lambda_1\setminus \Lambda_2)^{+l}$. This indicates that both strings, $S_{\Lambda_1}$ and $S_{\Lambda_2}$, transform under $\alpha$ according to the same $\sigma$ and $\mu$. In fact, their endpoint operators must be identical, denoted $O_R^\dagger(g)$. Thus, the interval can be extended, contracted, or shifted without affecting $\sigma$ and $\mu$. The same reasoning applies to charge strings without modification.

\end{proof}

As an analog to Prop. \ref{prop:Z2auto}, we can further constrain the consistent transformations by examining the commutation relations between two strings, each of which may be either a symmetry string or a charge string. The next result demonstrates that all consistent string transformations can be interpreted as a mapping to a braided auto-equivalence (anyon permutation symmetry) of a 2D ``bulk" topological order, specifically the Drinfeld center $\mathcal{C} = \mathcal{Z}[G]$, \ie a $G$ gauge theory in 2D. For $G = \otimes_j \Z_{n_j}$, a finite Abelian group, the anyon content of $\mathcal{C}$ consists simply of the gauge charges and gauge fluxes of each component, i.e., $\{ e_j, m_j \}$. The braiding and exchange statistics of the anyons are as follows: (1) all charges and fluxes are bosons; (2) an $e_j$ braiding with $m_j$ produces a phase factor $\omega_j = e^{i2\pi/n_j}$ and vice versa, while all other braidings yield a trivial phase.

{\theorem All QCA’s defined on the subalgebra symmetric under a finite Abelian group $G$, denoted by $\B$, form a group under finite composition. There exists a homomorphism from the group of QCA’s to $\au(\mathcal{C})$, the braided auto-equivalence of the bulk topological order, and this homomorphism is surjective.
\label{thm:Gbulk}
}

\begin{proof}
First, note the simple one-to-one correspondence between the strings and anyons of $\mathcal{C}$: the symmetry string $S(g_j) = \prod_{i=m}^n u(g_j)$, where $u(g_j)$ is the onsite symmetry operator generating the subgroup $\Z_{n_j} \subset G$, corresponds to the gauge flux $m_j$ in the bulk, while $S(Z[j]) = Z_m[j] (Z_n[j])^\dagger$, where $Z[j]$ is the generalized Pauli operator carrying a unit of $\Z_{n_j}$ charge, corresponds to the gauge charge $e_j$ in the bulk.

QCA’s on the symmetric subalgebra $\B$, being LP automorphisms, clearly form a group under composition. We claim that the transformation rules for strings given in Prop. \ref{prop:Gtransform} can be mapped to a braided auto-equivalence of $\mathcal{C}$. Specifically:
\begin{itemize}
    \item Two strings of the same type, for example, $S_g$ with the same $g \in G$, can be arranged as shown in Fig. \ref{Fig:exchange}. Denote these two strings by $S_{\Lambda_1}$ and $S_{\Lambda_2}$. According to Prop. \ref{prop:Gtransform}, under a QCA $\alpha$, the right endpoint operators of these strings are identical and thus commute. Since $\alpha$ is a QCA on $\B$, we have 
    \begin{equation} 
    [\alpha(S_{\Lambda_1}), \alpha(S_{\Lambda_2})] = 0, 
    \end{equation} 
    indicating that the left endpoint operator $O_L(g)$, as defined in Prop. \ref{prop:Gtransform}, must commute with the transformed symmetry element $\sigma(g)$. From the bulk perspective, this condition ensures that $\alpha$ preserves the (bosonic) exchange statistics of a gauge flux. The same argument applies to the charge string $S_{\hat{g}}$, and thus holds for all anyon types.
    \item Two strings of distinct types can be arranged as shown in Fig. \ref{Fig:braiding}. Since $\alpha$ is a QCA on $\B$, we find that $\alpha(S_{\Lambda_1})$ commutes with $\alpha(S_{\Lambda_2})$ up to the same phase factor as $S_{\Lambda_1}$ and $S_{\Lambda_2}$. From the bulk perspective, this condition ensures that $\alpha$ preserves the braiding statistics between distinct types of anyons.
\end{itemize}
Additionally, this mapping is a group homomorphism, as for any arbitrary string $S$, $(\alpha \circ \beta)(S) = \alpha[\beta(S)]$.

Finally, to show that this homomorphism is surjective, it suffices to identify a set of QCA's on $\B$ that generate the entire group $\au(\C)$. As discussed in Ref. \cite{2015brauer}, for a finite Abelian group $G$, three types of elementary transformations (defined below), namely the SPT entanglers classified by $H^2(G, \U)$, the outer automorphism group of $G$, denoted by Out$(G)$, and the KW duality associated with each cyclic component of $G$\footnote{For Abelian $G$, $\mathrm{Out}(G)$ is the same as $\mathrm{Aut}(G)$.}, collectively generate the full group $\au(\C)$. 
\end{proof}

{ We remark that it has been shown in Ref.~\cite{2023JonesDHR} (Theorem C) that there exists a homomorphism from the group of QCA (modulo FDC) acting on the quasi-local algebra of operators symmetric under a general fusion category $\mathcal{D}$ to ${\rm BrPic}(\mathcal{D})$, \ie the anyon permutation symmetry in $\ZZ[\mathcal{D}]$. Here, for the more restricted case where $\mathcal{D}$ is a finite Abelian group, we establish that a \emph{surjective} homomorphism actually exists. We expect that the homomorphism of Ref.~\cite{2023JonesDHR} reduces to our construction in this restricted setup (the ``superselection sectors'' there appear to play the same role as string operators here), but the precise connection between the two constructions, and in particular whether the homomorphism is surjective in the more general setup, remain interesting questions for future work.} 

Before proceeding, we define three types of elementary transformations on the $G$-symmetric subalgebra:

\begin{itemize} 
\item KW$_j$: The KW duality for the $j$-th cyclic subgroup of $G$. Recall that the $\Z_n$ KW duality is defined by the following transformation on the generating set: \begin{equation} 
\begin{split} 
\mathrm{KW}: & X_i \mapsto Z_i^\dagger Z_{i+1}, \\ 
& Z_i^\dagger Z_{i+1} \mapsto X_{i+1}, \end{split} 
\label{eq:ZnKW}
\end{equation} 
where $Z$ and $X$ are $\Z_n$ Pauli operators satisfying $Z_i X_i = \omega X_i Z_i$ with $\omega = e^{2\pi i/n}$, and commute on distinct sites.

\item $spt$: An SPT entangler, which is a $G$-symmetric FDC with \emph{asymmetric} gates. Physically, applying $spt$ to a $G$-symmetric product state produces a $G$-SPT state, labeled by an element $\omega\in H^2(G,\U)$. For $G$ finite Abelian, the entangler $spt$ leaves all $Z[j]$ operators unchanged while transforming the $X[j]$ operators as follows. Arranging the regular representations of each cyclic subgroup of $G$ sequentially from left to right at each site, the transformation rule is expressed as 
\begin{equation} 
X_i[j] \mapsto V_L(\omega) X_i[j] V_R(\omega)^\dagger. \label{eq:spten} \end{equation} 
Here, $V_L(\omega)$ is a product of the nearest Pauli-$Z$ operators located to the left of $X_i[j]$ (supported either on the $i$-th or $(i-1)$-th site). The charge of $V_L(\omega)$ under 
$G$ is specified by 
\begin{equation} \chi_V  (\cdot)= \frac{\omega(g_j, \cdot)}{\omega(\cdot, g_j)} \in H^1(G, \U), \end{equation} 
which is the slant product of $\omega$ with $g_j$, the generator of $\Z_{n_j} \subset G$. Similarly, $V_R(\omega)$ is defined as the product of the nearest Pauli-$Z$ operators (either on site $i$ or $(i+1)$) to the right of $X_i[j]$, carrying the same charge $\chi_V$. This transformation rule can be understood intuitively through the decorated domain wall construction of SPT phases \cite{2014decorated}.

\item $out$: An FDC implementing an element of $\mathrm{Out}(G)$, constructed as a product of onsite unitaries. On site $i$, it transforms $u_i(g)$, the onsite symmetry operator, to another onsite symmetry operator according to an element of $\mathrm{Out}(G)$. The onsite charged operators are transformed such that their commutation relations with the symmetry operators are preserved.

\end{itemize}
These three elementary transformations are, by construction, QCA's on the $G$-symmetric subalgebra. The following conclusion describes the structure of QCA's defined on the symmetric subalgebra, stating that, up to FDC's with symmetric gates, they are all compositions of the three elementary transformations and a generalized translation.

{\theorem In the group of QCA's defined on $\B$, the subalgebra symmetric under a finite Abelian group $G = \otimes_j \Z_{n_j}$, the group of FDC's with $G$-symmetric gates forms a normal subgroup. Let $\alpha$ be a QCA on $\B$. Then its action on $\B$ can be implemented in the form: \begin{equation} 
\alpha(O) =  K \circ \mathrm{T} \circ W^\dagger (O),\quad \forall O\in \B, 
\end{equation} where $K$ is a product of the three elementary transformations, \ie {KW$_j$, $out$, $spt$}, $W$ is an FDC with $G$-symmetric gates, and $\mathrm{T}$ represents a generalized translation (specified precisely below). 
\label{thm:Gauto}
}

\begin{proof} The argument illustrated in Fig.\ref{Fig:sFDC} directly establishes that the group of FDC's with $G$-symmetric gates forms a normal subgroup. To complete the proof of the Theorem, we proceed by examining the action of the QCA on a subalgebra of $\B$, which is isomorphic to the full operator algebra on a finite interval, similarly to Theorem \ref{thm:Z2QCA}. Below, we outline the proof and note where differences arise.

Firstly, by Theorem \ref{thm:Gbulk}, $\alpha$ transforms the string operators in accordance with an element of $\au(\C)$. By composing it with $K^{-1}$, which is a product of the three elementary transformations defined earlier in this section, we can arrange for $K^{-1}\circ \alpha$ to implement a trivial anyon permutation, \ie it preserves all string types. We now show that $K^{-1}\circ \alpha$ is thus a generalized translation, composed with an FDC with $G$-symmetric gates. This reasoning is structured into five steps.

\emph{Step 1:} We begin by coarse-graining the system, grouping every $l$ sites into a new unit cell to effectively reduce the spread of the QCA $\alpha$ to 1. We then focus on a subalgebra of $\B$ that is isomorphic to the full operator algebra supported on a finite interval. Specifically, for each single-site $Z[j]$ operator in the unit cells $I_{2n} \cup I_{2n+1}$ (where $Z[j]$ denotes the Pauli-$Z$ operator for the $j$-th cyclic subgroup of $G$), we define an operator 
\begin{equation} 
z_i[j] = Z_i[j] Z_m^\dagger[j],\, \forall j, \label{eq:locallyGcharged} 
\end{equation} 
where $m - i > 2(2l+2)^2+2l$. The set of all such operators with $i \in I_{2n} \cup I_{2n+1}$ is denoted by $A^o_{2n}$. Additionally, we denote the set of single-site $X[j]$ operators in $I_{2n} \cup I_{2n+1}$ for all $j$ as $A_{2n}$. Two key observations regarding these operators will be critical in the subsequent discussion:
\begin{itemize} 
\item The algebra $\B_{2n} = \langle A_{2n}, A^o_{2n} \rangle$ forms a subalgebra of the symmetric subalgebra $\B$, thus possessing a well-defined transformation under $\alpha$. For each generator of $\B_{2n}$, we can also define a \emph{local} $G$ charge by examining the portion supported within $I_{2n} \cup I_{2n+1}$. Specifically, operators in $A_{2n}$ are locally $G$-symmetric, while the local $G$ charge of generators in $A^o_{2n}$ is defined as the $G$ charge of the left component, \ie $Z_i[j]$.

\item The algebra $\B_{2n}$, generated by $A_{2n}^o$ and $A_{2n}$, is isomorphic to the algebra of all operators supported within $I_{2n} \cup I_{2n+1}$, denoted by $\mO_{2n} \otimes \mO_{2n+1}$. This isomorphism is established by the following map: \begin{equation} 
\begin{split} 
att: &Z_i[j] \mapsto z_i[j] = Z_i[j] Z_m^\dagger[j], \\ 
&X_i[j] \mapsto X_i[j],\, \forall j. 
\end{split} 
\label{eq:Gattach} 
\end{equation} 
\end{itemize}

We then extend this construction to nearby unit cells $[I_{2k}, I_{2t+1}]$ ($k+3 < n < t-3$), which are located at a sufficiently large distance from the fixed site $m$. Notably, all operators with the same local $G$ charge—specifically, those formed by the product of a $Z[j]$ operator within $[I_{2k}, I_{2t+1}]$ and another operator located far away—share the same right endpoint $Z_m[j]$. This setup allows us to construct a sequence of mutually commuting algebras, denoted $\B_{2k}, \dots, \B_{2t}$. It is also straightforward to verify that the algebra generated by this sequence,
\begin{equation} 
\mathcal{M} = \langle \B_{2k},\dots, \B_{2t} \rangle, 
\end{equation}
is isomorphic to the full operator algebra supported on the interval $[I_{2k}, I_{2t+1}]$. In particular, $\mathcal{M}$ contains the $G$-symmetric subalgebra with support within $[I_{2k}, I_{2t+1}]$.

\emph{Step 2:} We now construct another sequence of algebras, which can be regarded as the ``image" of the algebra sequence from Step 1 under the QCA $\alpha$. Consider the transformation of a locally charged operator as defined in Eq.\eqref{eq:locallyGcharged}:
\begin{equation} 
K^{-1}\circ\alpha(z_i[j]) = V_L[j] V_R[j]^\dagger, 
\label{eq:Grightend} 
\end{equation}
where $V_L[j]$ and $V_R[j]$ are order-$n_j$ unitaries carrying a unit of charge under $\Z_{n_j}$ and are neutral under the other cyclic subgroups, with supports in $[i-l,i+l]$ and $[m-l,m+l]$, respectively. By arranging two charge strings in the geometry shown in Fig.\ref{Fig:exchange}, and observing that operators in the set $\{ z[j], \forall j \}$ are mutually commuting, we conclude that the right endpoint operators $\{ V_R[j]^\dagger, \forall j \}$ also commute. Thus, we define a new set of locally charged operators as
\begin{equation} 
z'_i[j] = Z_i[j] V_R[j]^\dagger, 
\end{equation}
where the right-hand component is now the transformed operator. Since $V_R$ carries the same charge as $Z_m[j]$, each $z_i'[j]$ lies within $\B$. We then define $\mathcal{C}_{2n-1} = \langle A_{2n-1}, A^o_{2n-1} \rangle$, where $A_{2n-1}$ denotes the set of single-site $\{ X[j] \}$ operators on $I_{2n-1} \cup I_{2n}$, and $A^o_{2n-1} = \{ z'_i[j] | i \in I_{2n-1} \cup I_{2n}, \forall j \}$. Repeating this construction for unit cells $[I_{2k-1}, I_{2t+2}]$ generates a new sequence of algebras, denoted by $\mathcal{C}_{2k-1}, \dots, \mathcal{C}_{2t+1}$.

By construction, the algebras in the sequence $\mathcal{C}_{2k-1}, \dots, \mathcal{C}_{2t+1}$ satisfy the following properties: (1) they commute with each other; (2) the algebra generated by this sequence,
\begin{equation} 
\mathcal{N} = \langle \mathcal{C}_{2k-1}, \dots, \mathcal{C}_{2t+1} \rangle, 
\end{equation}
is isomorphic to the algebra of all operators supported within the interval $[I_{2k-1}, I_{2t+2}]$ and; in particular, (3) $\mathcal{N}$ includes the $G$-symmetric subalgebra fully supported within $[I_{2k-1}, I_{2t+2}]$.

\emph{Step 3:} In this step, we examine the transformation of $\B_{2n}$ under the QCA $\alpha$. The transformed algebra $\alpha(\B_{2n})$ lies within $\mathcal{C}_{2n-1} \otimes \mathcal{C}_{2n+1}$. Accordingly, we define the support algebras of $\alpha(\B_{2n})$ within $\mathcal{C}_{2n-1}$ and $\mathcal{C}_{2n+1}$ as follows:
\begin{equation} 
\begin{split} 
\Le_{2n} & = \sa(\alpha(\B_{2n}), \mathcal{C}_{2n-1}), \\ 
\Ri_{2n} & = \sa(\alpha(\B_{2n}), \mathcal{C}_{2n+1}), 
\end{split} 
\label{eq:definesupport}
\end{equation}
for $n \in [k, t]$. A geometric illustration of these algebras is provided in Fig. \ref{Fig:partition}. Note that all support algebras thus obtained commute with one another, a property that follows from the mutual commutativity of distinct $\B_{2n}$'s.

Now we characterize the support algebras. Following a similar reasoning as in the proof of Theorem \ref{thm:Z2QCA}, we find: (1) Each $\alpha(\B_{2n})$ decomposes as a tensor product of its support algebras, \ie
\begin{equation} 
\alpha(\B_{2n}) = \Le_{2n} \otimes \Ri_{2n}, , n\in[k+1,t-1]. \label{eq:supp1G} 
\end{equation} 
(2) We also have 
\begin{equation} 
\mathcal{C}_{2n-1} = \Ri_{2n-2} \otimes \Le_{2n}, \label{eq:supp2G} \end{equation}
for $n\in [k+1,t-1]$. (3) The support algebras $\Le_{2n}$ and $\Ri_{2n}$ for $n \in [k+1, t-1]$ each have a trivial center, \ie they are central simple algebras, and thus each is isomorphic to a matrix algebra \cite{bratteli2012operator}. (4) Lastly, each support algebra is a $G$-graded algebra, meaning it has a basis in which every element has a well-defined local $G$ charge.

By combining these results from Eqs. \eqref{eq:supp1G} and \eqref{eq:supp2G}, we deduce that for $n \in [k+1, t-1]$,
\begin{equation} 
\begin{split} 
\dim(\B_{2n}) & = \dim(\Le_{2n}) \times \dim(\Ri_{2n}), \\
\dim(\mathcal{C}_{2n-1}) & = \dim(\Le_{2n}) \times \dim(\Ri_{2n-2}), 
\end{split} \label{eq:Gtranslation} 
\end{equation}
which implies that ${ \Le_{2n} }$ are isomorphic to $d_L \times d_L$ matrix algebras, and ${ \Ri_{2k} }$ are isomorphic to $d_R \times d_R$ matrix algebras, with $d_L \times d_R = |G|^{2l}$, as each site carries a regular representation of $G$. Notably, the value of $d_L$ remains constant for all $n \in [k+1, t-1]$, and the same applies to $d_R$.

\emph{Step 4:} We now demonstrate that the QCA $K^{-1}\circ \alpha$ acts on the symmetric subalgebra, up to a symmetric FDC, as a generalized translation. To establish this, first observe that the Hilbert space on each site is a tensor product of regular representations of each cyclic subgroup of $G$. We decompose these onsite regular representations of each cyclic subgroup into prime factor components. For example, if $G$ contains a $\Z_4$ subgroup, we decompose the 4-dimensional regular representation of $\Z_4$ into a tensor product of two 2-dimensional subspaces. Importantly, this decomposition is chosen so that the on-site $\Z_4$ symmetry operator factorizes across these subspaces. For the example of $\Z_4$, consider its 4-dimensional regular representation in which the $X$ operator acts as
\begin{equation} 
X = \begin{pmatrix} 1 & & &\\ & e^{i\frac{\pi}{2}} & & \\ & & e^{i\pi} & \\ & & & e^{i\frac{3\pi}{2}} 
\end{pmatrix}. 
\label{eq:primefactors}
\end{equation}
This can be decomposed as a tensor product of two 2-dimensional subspaces, where $X$ acts as
\begin{equation} 
X = \begin{pmatrix} 1 & \\ & e^{i\frac{\pi}{2}}
\end{pmatrix}\otimes \begin{pmatrix} 1 & \\ & e^{i\pi} \end{pmatrix}. \end{equation}
We apply this decomposition across all cyclic subgroups of $G$ so that the symmetry operator factorizes accordingly. Specifically, for each cyclic subgroup $\Z_{n_j}$, where $n_j = \prod_k p_k$ with $\{ p_k \}$ as the prime factors of $n_j$, we decompose its regular representation such that \begin{equation} 
X[j] = \begin{pmatrix} 1 & & & \\ & e^{i\frac{2\pi}{n_j}} & & \\ & & \dots & \\ & & & e^{i\frac{2\pi (p_1 - 1)}{n_j}} 
\end{pmatrix} \otimes 
\begin{pmatrix} 1 & & & \\ & e^{i\frac{2\pi p_1}{n_j}} & & \\ & & \dots & \\ & & & e^{i\frac{2\pi p_1 (p_2 - 1)}{n_j}} \end{pmatrix} \otimes 
\begin{pmatrix} 1 & & & \\ & e^{i\frac{2\pi p_1 p_2}{n_j}} & & \\ & & \dots & \\ & & & e^{i\frac{2\pi p_1 p_2 (p_3 - 1)}{n_j}} \end{pmatrix} \otimes \dots 
\label{eq:symmetryprimefactors}
\end{equation} 
In this decomposition, the Hilbert space on each site can be expressed as \begin{equation} \Hi_i = \otimes_x V_x, 
\label{eq:decomposehilbertspace}
\end{equation} 
where each $V_x$ represents a $p_x$-dimensional Hilbert space with $p_x$ prime. Distinct $V_x$'s in the decomposition may have the same dimension.

According to Eq. \eqref{eq:Gtranslation}, each $\Le_{2n}$ for $n \in [k+1, t-1]$ forms a rank-$d_L = \prod_x p_x^{q_x}$ matrix algebra, while each $\Ri_{2n}$ is a rank-$d_R = \prod_x p_x^{2l - q_x}$ matrix algebra, where $q_x \in \Z$ ranges within $[0, 2l]$. Following a method analogous to the proof of Theorem \ref{thm:Z2QCA}, we partition each interval $I_{2n-1} \cup I_{2n}$, with $n \in [k+1, t-1]$, into two spatially adjacent segments: the left segment includes $2l - q_x$ copies of $V_x$ for each prime factor $p_x$, while the right segment includes $q_x$ copies of $V_x$. Here, since $q_x$ can vary across different prime factors, ``spatially adjacent” refers to the partitioning of the onsite Hilbert spaces within $I_{2n-1} \cup I_{2n}$ associated with each prime factor into left and right sections. For $V_x$'s of the same dimension, we denote them as $\{ V_x^{(y)} \}$, where the superscript $y$ labels the multiplicity index. The selection of the corresponding $q_x^{(y)}$ is arbitrary, provided that (1) $\sum_y q_x^{(y)}$ is fixed, and (2) either all $q_x^{(y)} \in [0, l]$ or all $q_x^{(y)} \in [l, 2l]$. This condition is always achievable.

We denote the algebra of all operators on the left segment as $\mathcal{T}_{2n-2}^R$ and on the right segment as $\mathcal{T}_{2n}^L$, which correspond to $d_R \times d_R$ and $d_L \times d_L$ matrix algebras, respectively. We claim that for $n\in[k+1,t-1]$, there exists a unitary $W_{2n-1}$, supported entirely within $I_{2n-1} \cup I_{2n}$, that maps the support of $\Le_{2n}$ and $\Ri_{2n-2}$ onto $\mathcal{T}_{2n}^L$ and $\mathcal{T}_{2n-2}^R$, respectively. Here, for locally charged operators in $\Le_{2n}$ and $\Ri_{2n-2}$ (i.e., those with a right component at a distant site), we refer only to the support of the left component.

Indeed, there exists an operator isomorphism $\rho$ given by 
\begin{equation} 
\rho: \mathcal{C}_{2n-1} = \oplus_{\hat{g} \in \hat{G}} (\mO_{2n-1} \otimes \mO_{2n})^{\hat{g}} \otimes V_R^\dagger[\hat{g}] \mapsto \mO_{2n-1} \otimes \mO_{2n} = \oplus_{\hat{g} \in \hat{G}} (\mO_{2n-1} \otimes \mO_{2n})^{\hat{g}}, \quad n \in [k+1, t-1], 
\end{equation} 
where $\mO^{\hat{g}}$ denotes the component with a specific $G$ charge, $\hat{g} \in H^1(G, U(1))$, in a $G$-graded algebra $\mO$. The right endpoint operator $V_R[\hat{g}]$ is a product of operators in ${V_R[j]}$, defined in Eq. \eqref{eq:Grightend}, such that $V_R[\hat{g}]$ carries charge $\hat{g}$. Thus, for $n \in [k+1, t-1]$, we define $\tilde\Le_{2n} = \rho(\Le_{2n})$, which is isomorphic to $\mathcal{T}_{2n}^L$ as both are $d_L \times d_L$ matrix algebras, and similarly, we set $\tilde\Ri_{2n-2} := \rho(\Ri_{2n-2}) \cong \mathcal{T}_{2n-2}^R$. Since the algebras $\tilde\Le_{2n}$ and $\tilde\Ri_{2n-2}$ are now fully supported within $I_{2n-1} \cup I_{2n}$, there exists a local unitary $W_{2n-1}$, supported on $I_{2n-1} \cup I_{2n}$, such that 
\begin{equation} 
\begin{split} W_{2n-1} \tilde\Ri_{2n-2} W_{2n-1}^\dagger &= \mathcal{T}_{2n-2}^R, \\ 
W_{2n-1} \tilde\Le_{2n} W_{2n-1}^\dagger &= \mathcal{T}_{2n}^L, 
\end{split} \end{equation} 
and this construction holds for all $n \in [k+1, t-1]$.

We then apply a generalized translation, denoted by $\mathrm{T}^{-1}$, to shift the algebras $\mathcal{T}_{2n}^L \otimes \mathcal{T}_{2n}^R$ back within the interval $I_{2n} \cup I_{2n+1}$. In particular, we shift each subspace $V_x$ on each site to the right by $q_x - l$ lattice sites. For subspaces with the same prime dimension, note that, under the choice where either all $q_x^{(y)} \in [0, l]$ or all $q_x^{(y)} \in [l, 2l]$, the shifts do not occur in opposite directions. According to Eq.\eqref{eq:supp1G}, we have 
\begin{equation} 
\mathcal{T}_{2n}^L \otimes \mathcal{T}_{2n}^R \cong \tilde\Le_{2n} \otimes \tilde\Ri_{2n} \cong \Le_{2n} \otimes \Ri_{2n} \cong \B_{2n} \cong \mathcal{O}_{2n} \otimes \mathcal{O}_{2n+1}, \label{eq:Gchain} \end{equation} 
where each tensor product in Eq.\eqref{eq:Gchain} is a matrix algebra. Consequently, the transformation \begin{equation} 
\mathrm{T}^{-1} \circ W_{2n-1} \circ W_{2n+1} \circ \rho \circ K^{-1} \circ \alpha \circ att \end{equation} restricts to an automorphism of $\mO_{2n} \otimes \mO_{2n+1}$, which, being an inner automorphism, can be implemented by a local unitary $Q_{2n}^\dagger \in \mO_{2n} \otimes \mO_{2n+1}$ for $n \in [k+1, t-1]$. Thus, we conclude that \begin{equation} W_\Lambda \circ \mathrm{T}^{-1} \circ \rho \circ K^{-1} \circ \alpha \circ att \label{eq
} \end{equation} acts as the identity on $\otimes_{n=k+1}^{t-1} (\mO_{2n}\otimes \mO_{2n+1})$, where \begin{equation} W_\Lambda = \left(\prod_{n=k+1}^{t-1} Q_{2n}\right) \otimes \mathrm{T}^{-1} \left(\prod_{n=k+1}^{t-1} W_{2n-1}\right) \mathrm{T} \end{equation} is a unitary operator with non-trivial action restricted to $[I_{2k}, I_{2t+1}]$. Extending this construction across the entire system, we obtain an FDC acting on the chain, given by \begin{equation} 
W = \left(\prod_n Q_{2n}\right) \otimes \mathrm{T}^{-1} \left(\prod_n W_{2n-1}\right) \mathrm{T}. 
\end{equation} 
Following the same reasoning as in Eqs \eqref{eq:FDCconstruction1} and \eqref{eq:FDCconstruction2}, the transformation $W \circ \mathrm{T}^{-1} \circ K^{-1} \circ \alpha$ acts on the symmetric subalgebra $\B$ as identity.

\emph{Step 5:} Finally, we demonstrate that $W$ is an FDC with $G$-symmetric gates. This follows straightforwardly: {By construction, $K^{-1}\circ \alpha$, and therefore $\mathrm{T} \circ W^\dagger$, implements a trivial anyon permutation (it preserves all string types). Furthermore, as in Eq.~\eqref{eq:primefactors}, we decompose the local Hilbert space such that the symmetry operator factorizes across its prime factors, and $\mathrm{T}$ commutes with the $G$-symmetry and also implements a trivial anyon permutation. As a consequence, the FDC $W$, as a composition of $\mathrm{T}$ and $K^{-1}\circ\alpha$, when acting on $\B$, preserves all string types -- in particular the symmetry string:
\begin{equation}
W^\dagger(S_g) = \mathrm{T}^{-1}\circ K^{-1}\circ \alpha (S_g) = O_L(g) S_g^{-l} O_R(g)^\dagger, \quad \forall g\in G,  
\label{eq:trivialGSPT}
\end{equation}
with $O_L(g)$ and $O_R(g)$ carrying trivial $G$-charge (here we use the same notation as in Prop.~\ref{prop:Gtransform}). This implies that $W^\dagger$ is a $G$-symmetric FDC that entangles a \emph{trivial} $G$-SPT, because when $G$ is finite Abelian, the Künneth formula implies that $H^2(G,U(1))$ is completely specified by the $G$-charge of $O_L(g)$ \cite{2014decorated}, which is trivial. Therefore, $W^\dagger$ is indeed an FDC with $G$-symmetric gates, according to Lemma~\ref{lemma:sFDC}.

}
\end{proof}

Theorem \ref{thm:Gauto} reduces any QCA to a composition of $\mathrm{T}$, $spt$, $out$, and KW, up to an sFDC, which does not alter the index and can therefore be ignored. The possible values of the index can therefore be deduced.  Calculating the index for these operations is straightforward, as they merely translate the generators of the symmetric subalgebra (with $spt$, as defined in Eq.\eqref{eq:spten}, dressing $X_i$ with products of neighboring generators, effectively relabeling without shifting the reference position). The index is then determined by counting the number of common generators in the overlapping region, as defined in Def.\ref{def:ind}.

{\Proposition
For two QCAs defined on $\B$, the index is multiplicative under composition, \ie 
\begin{equation} 
\ind(\alpha \circ \beta) = \ind(\alpha) \ind(\beta). 
\label{eq:composition}
\end{equation}
\label{ind-mul}
}

\begin{proof}
 A rigorous proof of this Proposition will be provided in a forthcoming work \cite{index}, which leverages the connection between the index defined in Def.\ref{def:ind} and the Jones index in von Neumann algebras \cite{2024Jonesindex}. We just mention that  for uQCA, Eq.\eqref{eq:composition} can be demonstrated by stacking two copies of the system and applying $\alpha$ and $\beta$ to each copy independently \cite{2020review}. However, this approach encounters a difficulty in the present case, as the symmetric subalgebra of the stacked system is not simply the tensor product of the symmetric subalgebras of the individual copies, as noted in Remark 2. While for a specific group $G$ one can check Eq. \eqref{eq:composition} manually, the general case requires tools from von Neumann algebra and the proof will be presented in Ref. \cite{index}.
\end{proof}

{\Proposition The index of a QCA defined on the subalgebra $\B$ symmetric under a finite Abelian group $G = \otimes_j \Z_{n_j}$ can take values in 
\begin{equation} 
\ind (\alpha) \in \Big\{ \prod_j (\sqrt{n_j})^{m_j} \prod_x p_x^{q_x} | q_x \in \Z,  m_j = 0,1 \Big\}, \end{equation} 
where $\{ p_x \}$ are the prime factors of $|G| = \prod_j n_j$.  
\label{thm:Gindex}
}

\begin{proof}

As per Theorem \ref{thm:Gauto}, any QCA on $\B$ is a composition of the three elementary transformations defined earlier in this section and a generalized translation, up to an FDC with symmetric gates, which, by Corollary \ref{cor:FDCinvariant}, does not alter the index and can thus be ignored.

We now examine the transformation of the standard generating set of $\B$, which is defined as the union of the standard generating sets of each cyclic subgroup of $G$, under the three elementary transformations and the generalized translation. The KW$_j$ operation defined in Eq.\eqref{eq:ZnKW} can be viewed as a ``half” translation of the standard generating set associated with the $j$-th cyclic subgroup. Meanwhile, $spt$ as defined in Eq.\eqref{eq:spten} and $out$ act by relabeling elements within this set, leaving their spatial positions unchanged. Direct verification shows that both $spt$ and $out$ have index 1, while $\ind(\mathrm{KW}_j) = \sqrt{n_j}$.

Below, we demonstrate that a generalized translation $\mathrm{T}$ has an index $\ind(\mathrm{T}) = \prod_x p_x^{q_x}$. This is achieved by showing that, for any QCA on $\B$ implementing a trivial anyon permutation, the index can be determined from the dimensions of the support algebras defined in Eq.\eqref{eq:definesupport}. Since $\mathrm{T}$ falls within this class of QCA, its index follows directly.

\begin{figure}
\begin{center}
  \includegraphics[width=.45\textwidth]{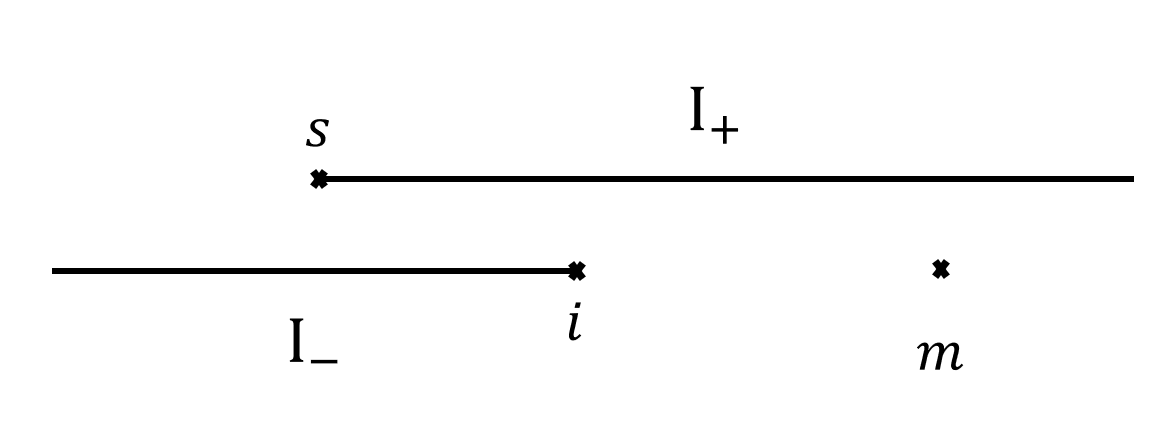} 
\end{center}
\caption{The geometric arrangement of $I_-$, $I_+$, and site $m$.
}
\label{Fig:composition}
\end{figure}

Let $\alpha$ be a QCA with spread $l$ that implements a trivial anyon permutation. Consider the geometry illustrated in Fig. \ref{Fig:composition}, where the distance from site $m$ to $I_-$ exceeds $2(2l+2)^2+2l$, and $m \notin \partial_l I_+$. Let the rightmost site of $I_-$ be site $i$, and denote the standard generating set on $I_-$ by $B$, while the symmetric subalgebra on $I_+$ is denoted as $\B_+$. Define $\B_1$ as the algebra generated by $B$ and $B' = \{ Z_i[j] Z_m[j]^\dagger, \, \forall j \}$. Following the same reasoning as in the proof of Theorem \ref{Thm:global}, we find
\begin{equation}
    \ind(\alpha) = \frac{\eta(\alpha(\B_1),\B_+)}{\eta(\B_1,\B_+)}.
\end{equation}
The right-hand side of this equation can be evaluated as follows. Let the leftmost site of $I_+$, which corresponds to the leftmost site of $I_{2n-1}$ shown in Fig. \ref{Fig:partition}, be denoted by $s$. The algebra $\B_1 \cap \B_+$ is isomorphic to the full operator algebra on the interval $[s,i]$, with a dimension of $|G|^{2(i-s+1)}$. On the other hand, the algebra $\alpha(\B_1) \cap \B_+$ is the tensor product of the support algebras with spatial support in $I_+$, \ie $\{ \Ri_{2a}, \Le_{2b} \mid a \geq n-1, b \geq n \}$. Consequently, by Eq.\eqref{eq:supp1G}, we find
\begin{equation}
    \ind(\alpha) = \frac{\sqrt{\dim(\Ri_{2n-2})}}{\sqrt{|G|^{2l}}},
\end{equation}
where $|G|^{2l}$ in the denominator represents the dimension of the full operator algebra within $I_{2n-1}$. Moreover, for a generalized translation $\mathrm{T}$, the algebra $\Ri_{2n-2}$ is a rank-$d_R=\prod_x p_x^{2l-q_x'}$ matrix algebra, where $\{ p_x \}$ are the prime factors of $|G|$. This establishes that $\ind(\mathrm{T}) = \prod_x p_x^{q_x}$, where $\{ q_x \}$ is a set of integers.
\end{proof}

Furthermore, given the structure of an arbitrary QCA on the symmetric subalgebra as described in Theorem \ref{thm:Gauto}, we claim that the equivalence classes of QCA's are fully characterized by their corresponding anyon permutation and the index.

{\theorem QCA on the $G$-symmetric subalgebra $\B$ are completely classified by their action on string operators and their index. {Specifically, two QCA's are associated with the same element in $\au(\C)$ and have the same index if and only if they differ by an FDC with $G$-symmetric gates.}
\label{thm:completeG}
}

\begin{proof} 
We first prove the ``only if'' direction. From Proposition \ref{ind-mul}, we immediately observe that \begin{equation} \ind(\alpha)\ind(\alpha^{-1}) = 1. \label{eq:indexinverse} 
\end{equation}

For two QCA's, $\alpha$ and $\beta$, that implement the same anyon permutation in $\au(\C)$, Theorem \ref{thm:Gauto} implies that \begin{equation} 
\alpha^{-1}\circ \beta (\B) = W^\dagger \circ \mathrm{T} (\B), 
\end{equation} 
where $\mathrm{T}$ is a generalized translation and $W$ is an FDC with symmetric gates. Additionally, we have 
\begin{equation} \ind(\alpha^{-1} \circ \beta) = \frac{\ind(\beta)}{\ind(\alpha)} = 1, \end{equation} 
indicating that the generalized translation $\mathrm{T}$ must be trivial.\footnote{ Recall that, in the construction in the proof of Theorem \ref{thm:Gauto}, $\mathrm{T}$ does not shift subspaces with the same prime dimension in opposite directions. Consequently, for the index to be trivial, the translation must also be trivial for all subspaces.} Consequently, we find 
\begin{equation} \beta(\B) = W' \circ \alpha(\B), \end{equation} 
where $W' = \alpha(W^\dagger)$ is also an FDC with symmetric gates, since the subgroup of FDC's with symmetric gates is normal within the QCA group on $\B$.  

\add{Now we prove the ``if'' direction. If two QCAs differ by an FDC composed of local symmetric gates, they must act on charge and symmetry strings in the same way, \ie they have the same transformation rules $\sigma$, $\mu$, $\gamma$, and $\nu$ as in Prop.~\ref{prop:Gtransform}. Due to the correspondence between string operators and anyons in $\C$, the two QCAs are associated with the same element in the anyon permutation symmetry $\au(\C)$. They also have the same index by Corollary~\ref{cor:FDCinvariant}.
}
\end{proof}

We note that, even without relying on the multiplicativity of the index under composition (\ie Proposition \ref{ind-mul}), whose rigorous proof requires more advanced mathematical tools, it can still be deduced that the index and the anyon permutation fully specify the QCA on $\B$. Specifically, if $\alpha$ and $\beta$ correspond to the same element in $\au(\C)$ and satisfy $\ind(\alpha^{-1} \circ \beta) = 1$, then they differ only by an FDC with $G$-symmetric gates.

Theorem \ref{thm:completeG} is one of the main results of this work. It states that for a chain with a regular representation of a finite Abelian group at each site, the group of QCA's defined on the symmetric subalgebra is completely classified by (1) the index, which quantifies the flow of the subalgebra, and (2) a homomorphism to $\au(\C)$, which describes how a QCA permutes anyons. The kernel of both homomorphisms is exactly the group of $G$-symmetric FDC's.

Before concluding this section, we present an additional result: any QCA on the symmetric subalgebra can be represented by a matrix product operator (MPO) with finite bond dimension. {MPO representations of the KW transformation have been constructed for $\Z_2$ in Ref. \cite{2024measurespt,2024Seibergshaoseif} and for a general finite $G$ in Ref.~\cite{2025GKW},} and their existence is naturally expected due to the correspondence between fusion categorical symmetries in 1D and anyons in a 2D topological order (specifically, the Drinfeld center of the fusion category \cite{2005stringnet}), which can be represented as an MPO at the virtual level of the 2D tensor network describing the topological order \cite{2020MPSreview,2023MPOsymmetry,2010virtual}, we here provide a proof of this existence from a purely 1D perspective.

{\Proposition Any QCA on $\B$, considered as an operator on the Hilbert space $\otimes_i \Hi_i$, can be represented by an MPO with finite bond dimension.}

\emph{Proof:} According to Theorem \ref{thm:Gauto}, any QCA on $\B$ can be decomposed as a composition of the transformations $spt$, $out$, $\mathrm{T}$, $W$, and the KW transformations for each cyclic subgroup. The first four transformations are all unitary QCAs and can be represented by matrix product unitaries \cite{2004QCA,2017MPUCirac}, so we only need to show that the KW transformation also admits an MPO representation. Let us focus on a single $\Z_n$ subgroup within $G$.

Consider a system of size $L$ with periodic boundary conditions, where each site hosts an $n$-dimensional qudit in the regular representation of $\Z_n$. We attach an additional $n$-dimensional ancilla qudit to each site, creating a copy of the original system denoted as $\otimes_i \Hi_i'$. Define the maximally entangled state between the ancilla and the physical system as $| \Phi \ra$: 
\begin{equation} | \Phi \ra = \otimes_i \left(\frac{1}{\sqrt{n}} \sum_{s_i=1}^n | s_i \ra \otimes | s_i \ra'\right), 
\end{equation} 
where we have 
\begin{equation} X_i |s_i \ra = e^{i2\pi \frac{s_i}{n}} |s_i \ra, \quad X_i' |s_i \ra' = e^{i2\pi \frac{s_i}{n}} |s_i \ra'. \end{equation}

Consider an operator $D$ acting on $\otimes_i \Hi_i$ that implements the $\Z_n$ KW duality. Specifically, we define $D$ as the operator satisfying: \begin{equation} D O = \mathrm{KW}(O) D, \quad \forall O \in \B. \end{equation} The operator $D$ can be represented as an MPO if and only if $|\Psi \rangle = D \otimes \mathbf{1}' | \Phi \rangle$ is a matrix product state (MPS), where $\mathbf{1}'$ is the identity operator acting on the ancilla $\otimes_i \Hi_i'$. This holds because $| \Psi \rangle$ is exactly the Choi state of $D$.

Notice that the maximally entangled state $|\Phi\rangle$ is stabilized by the following commuting operators:
\begin{equation}
    X_i X_i'^\dagger | \Phi\ra = | \Phi \ra,\quad Z_i Z_i' | \Phi\ra = | \Phi \ra.
\end{equation}
Therefore, applying the transformation rule of the $\Z_n$ KW transformation in Eq.\eqref{eq:ZnKW}, we obtain: 
\begin{equation} 
\begin{split} 
& Z_i^\dagger Z_{i+1} X_i'^\dagger | \Psi \rangle = | \Psi \rangle, \\ 
& X_{i+1} Z_i'^\dagger Z_{i+1}' | \Psi \rangle = | \Psi \rangle, 
\end{split} \end{equation} 
for all $i$. Consequently, $| \Psi \rangle$ is the unique gapped ground state of a local, commuting Hamiltonian $H_\Psi$, given by 
\begin{equation} 
H_\Psi = \sum_i -Z_i^\dagger Z_{i+1} X_i'^\dagger - X_{i+1} Z_i'^\dagger Z_{i+1}' +\hc, 
\end{equation} 
namely, $|\Psi\rangle$ is stabilized by $2L$ independent stabilizers, each of order $n$. Therefore, $| \Psi \rangle$ can be represented as an MPS with finite bond dimension \cite{2006MPSrep}. Extending this argument to all cyclic subgroups of $G$ completes the proof of the Proposition. \qed

\section{Example: The $\Z_2\times \Z_2$ symmetric subalgebra}
\label{sec:Z2Z2}

In this section, we provide an explicit example by considering a qubit chain with a $\Z_2^e \times \Z_2^o$ symmetry generated by \begin{equation} 
\eta^e = \prod_{i:\mathrm{even}} X_i,\quad \eta^o = \prod_{i:\mathrm{odd}} X_i, \end{equation} 
which act on all even and all odd sites, respectively. We study QCA's defined on $\B$, the $\Z_2^e \times \Z_2^o$ symmetric subalgebra, generated by the set $\{ Z_{i} Z_{i+2}, X_i \}$ for all sites $i$ on the lattice.

As stated in Theorem \ref{thm:completeG}, QCA's on $\B$ are completely classified by their action on string types and their index. We now focus on the action on string operators, or equivalently, to how QCA's permute anyons, corresponding to an element in $\au(\C)$. Here, the bulk topological order $\C$ is simply a $\Z_2 \times \Z_2$ gauge theory, \ie two copies of the toric code, with anyons given by $\{ e_1, m_1, e_2, m_2 \}$ and their compositions (where $e_1$ and $m_1$ are the gauge charge and flux associated with $\Z_2^e$, and $e_2$ and $m_2$ are those of $\Z_2^o$). It is known that, in this case, $\au(\C) = (S_3 \times S_3) \rtimes \Z_2$.

\begin{figure}
\begin{center}
  \includegraphics[width=.40\textwidth]{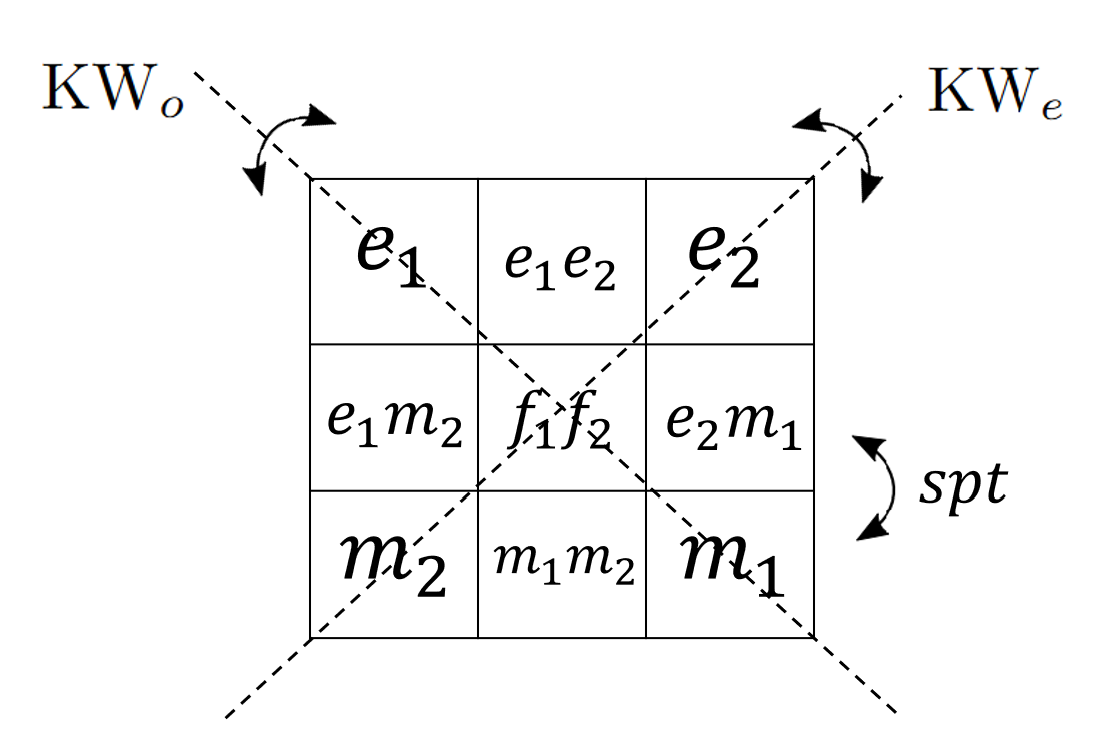} 
\end{center}
\caption{The actions of QCA on the table of bosons. The $out$ operation implements the $S_3$ permutation group on the three columns.}
\label{Fig:anyon}
\end{figure}

A clear way to understand this group of braided auto-equivalences is to examine the symmetry of the boson table in the 2D $\Z_2 \times \Z_2$ gauge theory \cite{2018colorcode}, as illustrated in Fig \ref{Fig:anyon}. The nine nontrivial bosons in this theory are arranged such that the following properties hold: (1) The product of any three bosons in the same row or column is the identity; (2) the braiding of bosons within the same row or column yields 1, and it gives -1 otherwise. Within $\au(\C)$, the two $S_3$ subgroups arise from permuting rows and columns, respectively, while the $\Z_2$ factor corresponds to transposing the boson table.

Now we list the three elementary transformations discussed in the previous section for the example of $G = \Z_2 \times \Z_2$:

\begin{itemize} 
\item The KW transformations: There are two KW transformations associated with the $\Z_2^e$ and $\Z_2^o$ symmetries, respectively. We denote them as KW$_e$ and KW$_o$.
\item $spt$: In 1D, there is a nontrivial SPT phase protected by $\Z_2^e \times \Z_2^o$ symmetry, specifically the cluster chain \cite{2001cluster,2020MPSreview,2013SPTChen}, labeled by the nontrivial element of $H^2(\Z_2 \times \Z_2, \U) = \Z_2$. The entangler $spt$ acts on local operators as
\begin{equation}
    spt: \, X_i \leftrightarrow Z_{i-1} X_i Z_{i+1}.
\end{equation}

\item $out$: The group of outer automorphisms of $\Z_2 \times \Z_2$ is $S_3$, corresponding to the permutation of the three non-identity elements in $\Z_2 \times \Z_2$.

\end{itemize}

The action of these three elementary transformations on anyons (equivalently, on the string types) is also depicted in Fig \ref{Fig:anyon}. Together, they generate the entire group of braided auto-equivalence of the $\Z_2 \times \Z_2$ gauge theory.

We now turn our attention to a specific class of QCA's on $\B$, aiming to connect with the concept of non-invertible symmetries \cite{2015GGS,2023ShaoTasi,2024noninvICTP}, an area of active study. Specifically, we consider a QCA $\alpha$ that meets the following conditions in its action on the symmetric subalgebra:
\begin{itemize}

\item  Under $\alpha$, the image of every pure charge string includes a nontrivial symmetry string. Equivalently, the homomorphism $\gamma$ in Prop.\ref{prop:Gtransform} maps each nontrivial element of $\hat{G}$ to a nontrivial element of $G$.

\item We require that $\alpha^2$ preserves all string types, \ie $\alpha^2$ implements a trivial element in $\au(\C)$.
\end{itemize}
Hereafter, we refer to these two requirements as the ``TY-conditions."

The motivation for considering QCAs of this form stems from the fusion rules of the $\Z_2^e \times \Z_2^o$ Tambara-Yamagami (TY) fusion category \cite{tambara1998tensor}, where a non-invertible transformation $\mathrm{D}$ and the symmetry operators obey the following algebra:
\begin{equation}
    \begin{split}
        & \eta^e \mathrm{D} = \mathrm{D} \eta^e = \eta^o \mathrm{D} = \mathrm{D} \eta^o = \mathrm{D}, \\
        & \mathrm{D}^2 = 1 + \eta^e + \eta^o +\eta^e\eta^o.
    \end{split}
    \label{eq:TYfusion}
\end{equation}
If we interpret $\mathrm{D}$ as a QCA acting on the symmetric subalgebra, which is also a locality-preserving transformation (though not necessarily a QCA) on the full local operator algebra, then the first line in Eq.\eqref{eq:TYfusion} indicates that $\mathrm{D}$ annihilates all charged local operators. (Physically, this operation corresponds to ``gauging.") However, if a charge string associated with an element $\hat{g} \in \hat{G}$ is transformed into a purely charge string, such that
\begin{equation} 
\alpha[V^L_{\hat{g}} (V^R_{\hat{g}})^\dagger] = O^L_{\hat{k}} (O^R_{\hat{k}})^\dagger, \end{equation}
where $V_{\hat{g}}$ and $O_{\hat{k}}$ are local unitary operators carrying charges $\hat{g}$ and $\hat{k}$, respectively, then one can naturally define the transformation of local charged operators by $\alpha(V^L_{\hat{g}}) = O^L_{\hat{k}}$ without being annihilated. Furthermore, any operator $\mathrm{D}$ that squares to the identity on the symmetric subalgebra must satisfy the second TY condition.

It is known that there are four TY fusion categories satisfying the fusion rule in Eq.\eqref{eq:TYfusion}, denoted by $\mathrm{TY}(\Z_2 \times \Z_2, \chi, \epsilon)$ \cite{2019Thorngren}. Here, $\chi: \Z_2\times\Z_2\rightarrow \U$ represents a symmetric, non-degenerate bicharacter, and $\epsilon = \pm 1$ denotes the Frobenius-Schur (FS) indicator. For group-like symmetries, it is established that the topological data—\ie the anomaly—can be derived from the symmetry operator \cite{2014elsenayak,2024Kapustin}, or equivalently, from its action on the local operator algebra, independent of the Hamiltonian. A natural question arises: can a similar analysis be extended to non-invertible symmetries? Specifically, how can we define $\chi$ and $\epsilon$ based on the assumption that $\alpha$ is a QCA on the symmetric subalgebra $\B$ satisfying the two TY conditions, such that their properties align with their counterparts as defined in field theory \cite{2019Thorngren}?

In the following, we provide an algebraic definition of the bicharacter purely based on the transformation of string operators under the QCA $\alpha$, independent of any reference to a symmetric lattice Hamiltonian. We show that it must be symmetric and non-degenerate under the TY conditions. Additionally, we present an example where the FS indicator is derived solely from the transformation of the symmetric subalgebra under a QCA.

We define the bicharacter associated with a QCA on $\B$ as follows:

\begin{equation} S_{\hat{h}}^\dagger \alpha(S_{\hat{g}}) S_{\hat{h}} = \chi({\hat{g}},{\hat{h}})\alpha(S_{\hat{g}}), \label{eq:defbicharac} 
\end{equation}
where $S_{\hat{h}}$ and $S_{\hat{g}}$ (with $\hat{g},\hat{h} \in H^1(G, \U)$) are charge strings of sufficiently large length such that $\partial_l(S_{\hat{g}}) \cap \partial_l(S_{\hat{h}}) = \varnothing$. The relative positioning of these two string operators is illustrated in Fig.\ref{Fig:braiding}. Effectively, $\chi$ quantifies the commutation relation between a charge $\hat{h}$ and the symmetry operator $\gamma(\hat{g})$ as described in Prop. \ref{prop:Gtransform}, and is denoted as:
\begin{equation} \chi({\hat{g}},{\hat{h}}) = \gamma(\hat{g})[\hat{h}], 
\end{equation}
representing the charge of $\hat{h}$ under the symmetry element $\gamma(\hat{g})$. From the results in Prop. \ref{prop:Gtransform}, it follows that:
\begin{equation} \begin{split} & \chi({\hat{g}},\hat{h}_1 \hat{h}_2) = \chi(\hat{g},\hat{h}_1) \chi(\hat{g},\hat{h}_2), \\ & \chi(\hat{g}_1 \hat{g}_2, \hat{h}) = \chi(\hat{g}_1, \hat{h}) \chi(\hat{g}_2, \hat{h}), \end{split} \label{eq:bicharacterc} 
\end{equation}
since $\gamma$ is a homomorphism from $\hat{G}$ to $G$. Equation \eqref{eq:bicharacterc} matches the standard definition of a bicharacter for a finite Abelian group $G\cong \hat{G}$.

Applying $\alpha$ to both sides of Eq.\eqref{eq:defbicharac}, we obtain \begin{equation} 
\alpha(S_{\hat{h}})^\dagger S_{\hat{g}}' \alpha(S_{\hat{h}}) = \chi({\hat{g}},{\hat{h}}) S_{\hat{g}}', 
\end{equation} where $S_{\hat{g}}' = \alpha^2(S_{\hat{g}})$. Since $\alpha^2$ preserves all string types by the TY condition, $S_{\hat{g}}'$ must again correspond to a charge string associated with ${\hat{g}}$. Consequently, we have 
\begin{equation} 
S_{\hat{g}}' \alpha(S_{\hat{h}}) (S_{\hat{g}}')^\dagger = \chi({\hat{g}},{\hat{h}}) \alpha(S_{\hat{h}}). 
\end{equation} 
For $G = \Z_2\times \Z_2$, since all non-identity elements in $H^1(G,U(1))$ are of order 2, it follows that $\chi({\hat{g}},{\hat{h}}) = \chi({\hat{h}},{\hat{g}})$, indicating that the bicharacter is symmetric. Finally, the first TY condition implies that for any ${\hat{g}} \in {\hat{G}}$, there exists at least one ${\hat{h}} \in {\hat{G}}$ such that
\begin{equation} \chi({\hat{g}},{\hat{h}}) \neq 1, \end{equation}
which ensures that $\chi$ is non-degenerate.

For $\hat{G} \cong G = \Z_2 \times \Z_2$, it is known that there exist two symmetric, non-degenerate bicharacters \cite{2019Thorngren}. Let $\hat{a}$ and $\hat{b}$ denote the elements in $\hat{G}$ with nontrivial charge under $\Z_2^e$ (while neutral under $\Z_2^o$) and under $\Z_2^o$ (while neutral under $\Z_2^e$), respectively.

\begin{itemize} 
\item The first possibility is $\chi^{(1)}(\hat{a},\hat{b}) = -1$ and $\chi^{(1)}(\hat{a},\hat{a}) = \chi^{(1)}(\hat{b},\hat{b}) = 1$. A concrete realization of this bicharacter is found in the fusion category Rep(D$_8$) $= \mathrm{TY}(\Z_2 \times \Z_2, \chi^{(1)}, \epsilon = 1)$, where the non-invertible operator  in Eq.\eqref{eq:TYfusion} is given by 
\begin{equation} \mathrm{D}_1 = \mathrm{T}^{-1} \circ \mathrm{KW}_e \circ \mathrm{KW}_o \circ\mathcal{P}, 
\label{eq:repD8}
\end{equation} 
where $\mathrm{T}$ represents a rightward lattice translation by one site,\footnote{In terms of the three elementary operations defined in Sec.\ref{sec:finiteabelian}, $\mathrm{T}$ is a SWAP between sites $2n$ and $2n+1$, an element of $\mathrm{Out}(\Z_2 \times \Z_2)$, followed by a two-site translation of all odd sites.} and $\mathcal{P}$ is the projector to the $\Z_2\times\Z_2$ invariant subspace. On the generators of the symmetric algebra $\B$, its action simply exchanges $X_i$ and $Z_{i-1} Z_{i+1}$. Computing the index of $\mathrm{D}$ gives \begin{equation} 
\ind(\mathrm{D}_1) = 1, 
\end{equation} 
which provides a quantitative characterization of the fact that the non-invertible operator in the Rep(D$_8$) fusion category does not mix with lattice translation \cite{2024seifshao}.

\item The second possibility is $\chi^{(2)}(\hat{a},\hat{a}) = \chi^{(2)}(\hat{b},\hat{b}) = -1$ and $\chi^{(2)}(\hat{a},\hat{b}) = 1$, which is realized in the fusion category Rep(H$_8$) $= \mathrm{TY}(\Z_2 \times \Z_2, \chi^{(2)}, \epsilon = 1)$. On lattice, such a bicharacter is realized by the following QCA: 
\begin{equation}
    \mathrm{D}_2 = \mathrm{KW}_e \circ \mathrm{KW}_o\circ \mathcal{P},
\end{equation}
which implements the transformation $X_i \mapsto Z_i Z_{i+2} \mapsto X_{i+2}$ on the standard generating set of the symmetric subalgebra $\B$. A direct computation yields
\begin{equation}
    \ind(\mathrm{D}_2) = 2.
\end{equation}
Notice that ${\rm D}_2$ mixes with translation, so its algebra is different from the TY fusion rule. 

\end{itemize}

Finally, we outline a first step towards a purely algebraic definition of the FS indicator based on the action of $\alpha$ on $\B$. Consider two QCA’s, $\alpha$ and $\beta$, on $\B$ that satisfy the TY condition. If they further possess the same symmetric, non-degenerate bicharacter (and thus implement the same anyon permutation in $\au(\C)$) and have the same index, then by Theorem \ref{thm:completeG}, their actions on $\B$ differ only by an FDC with symmetric gates, $W$, so that \begin{equation} 
\alpha = W \circ \beta. 
\end{equation} 
If we additionally require that $\alpha^2$ and $\beta^2$ act identically on $\B$, namely, $\alpha^2 = \beta^2$, it follows that 
\begin{equation} 
W \circ W' \circ \beta^2(O) = \beta^2(O),\quad \forall O\in\B, 
\end{equation} 
where $W' = \beta(W)$, which is again an FDC with symmetric gates. This relation indicates that the FDC $WW'$ must be a symmetry operator associated with some element $g \in \Z_2^e \times \Z_2^o$. As an FDC, we define the spread of $WW'$ as $l'$. Consequently, we can truncate $WW'$ to a finite region $\Lambda$ (with a length much larger than $l'$) by retaining gates with nontrivial support in $\Lambda$ and removing the others. We denote this truncated circuit by $U_\Lambda$. Two observations are useful: (1) Deep inside the region $\Lambda$ (i.e., away from $\partial_{l'} \Lambda$, the $l'$ boundary of $\Lambda$), $U_\Lambda$ behaves as $WW'$, and thus acts as the symmetry operator associated with $g$. (2) Deep outside $\Lambda$, $U_\Lambda$ acts as the identity by construction. As a result, we have 
\begin{equation} U_\Lambda = O_L S_g(\Lambda^{-l'}) O_R^\dagger, \end{equation} 
where $O_L$ and $O_R$ are unitary operators supported within a distance $l'$ from the left and right endpoints of $\Lambda$, respectively, and $S_g(\Lambda^{-l'})$ is the symmetry operator associated with $g$, restricted to the interior of $\Lambda$. Since $WW'$ consists of symmetric gates, the truncated circuit $U_\Lambda$ is also symmetric. Thus, the boundary operators $O_L$ and $O_R$ carry a well-defined charge $\lambda \in H^1(\Z_2 \times \Z_2, U(1))$, and their charges must be identical.

As an illustrative example, we show that, in certain cases, the FS indicator of the TY category can be determined (or exchanged) by the sFDC $W$. Specifically, starting from the fusion category $\text{Rep}(\mathrm{D}_8)$, we can obtain the fusion category $\text{Rep}(\mathrm{Q}_8) = \mathrm{TY}(\mathbb{Z}_2 \times \mathbb{Z}_2, \chi^{(1)}, \epsilon = -1)$—distinguished by a different FS indicator—by composing the non-invertible operator $\mathrm{D}$ in Eq.\eqref{eq:repD8} with an FDC $W$, chosen such that the spatial truncation of $WW'$, $U_\Lambda$, satisfies the following conditions: (1) the symmetry string $S_g(\Lambda^{-l'})$ of the truncated circuit $U_\Lambda$ is associated with the diagonal $\Z_2$ element of $\Z_2^e\times \Z_2^o$; (2) $O_L$ and $O_R$ carry nontrivial charges under both $\Z_2^e$ and $\Z_2^o$. As shown in Appendix \ref{app:Rep}, the new operator $\mathrm{D}' = W \circ \mathrm{D}$, together with the symmetry operators, indeed forms the Rep$(\mathrm{Q}_8)$ fusion category.

A lattice realization of the non-invertible operator $\mathrm{D}'$ of a Rep($\mathrm{Q}_8$) fusion category is given by~\cite{li2024non}
\begin{equation} \mathrm{D}' = W \circ \mathrm{T}^{-1} \circ \mathrm{KW}_e \circ \mathrm{KW}_o, 
\end{equation} 
where $W = \prod_i \sqrt{Z_{i-1} X_i Z_{i+1}}$. One can verify explicitly that the truncated unitary $U_\Lambda$ takes the form \begin{equation} U_\Lambda = Z_{m-1} Z_m \left(\prod_{j=m}^{n} X_j \right) Z_n Z_{n+1},
\end{equation}
which satisfies the conditions for the sFDC discussed in the preceding paragraph. This serves as an illustrative example, while a systematic definition and exploration of the FS indicator from the perspective of operator algebras is deferred to future studies.

\section{Discussions}
\label{sec:discussion}
In this work, we have systematically explored the definition, classification, and characterization of quantum cellular automata  on symmetric subalgebras in 1D spin systems, where each site carries a regular representation of a finite Abelian group $G$. We achieved a complete classification framework based on two key components: an index that quantifies the flow of the subalgebra, and the transformation rules of string operators, which correspond to anyon permutation symmetries in a 2D $G$ gauge theory. Notably, we have demonstrated that two QCAs correspond to the same anyon permutation and share the same index if and only if they differ by a $G$-symmetric finite-depth circuit.

Several open questions remain for future investigation:

\noindent\emph{General Symmetric Subalgebras:} It is natural to consider extending our results to more general symmetric subalgebras, which can be tentatively categorized into two classes: \begin{itemize} \item First, one might begin with an ordinary spin Hilbert space at each site, where the local operator algebra is a matrix algebra. We then consider the subalgebra of the full local operator algebra that is symmetric under a categorical symmetry or matrix product operator (MPO) symmetry forming a fusion category \cite{2023MPOsymmetry}. Intuitively, QCAs defined on this symmetric subalgebra should again be characterized by a mapping to an anyon permutation symmetry in a two-dimensional topological order, \ie the Drinfeld center of the fusion category, combined with an index that parametrizes the information flow. How can we make this intuition more precise? Do these considerations provide a complete classification of QCAs on this subalgebra?

\item Second, consider the case of an anyon chain \cite{2007anyonchain,2009anyonchain}, where the Hilbert space at each site is not a spin degree of freedom but an object $X$ in a fusion category. Consequently, the local operator algebra is no longer a tensor product of local matrix algebras (possibly with an additional symmetry constraint), but rather the endomorphism algebra of tensor powers of $X$. Progress has been made in investigating QCAs defined in this setting \cite{2023JonesDHR,2024JW}, and in particular, a generalization of the GNVW index based on the Jones index from subfactor theory has been proposed \cite{2024Jonesindex}. However, a comprehensive understanding of the full landscape of such QCAs remains lacking. Specifically, what set of data can completely classify QCAs defined on an anyon chain? Moreover, in this context, what is the relationship between the index based on the Jones index and the one proposed in our work, which is based on the overlap of two local subalgebras? We will address some of these issues in a forthcoming work \cite{index}.

\end{itemize}

\noindent \emph{Anomalies of Non-Invertible Symmetries:} For any group-like symmetry implemented by unitary QCAs on a quantum spin chain, one can define an anomaly index that takes values in the 3rd group cohomology \cite{2014elsenayak,2024Kapustin}. This anomaly index is purely kinematic, \ie it does not depend on the Hamiltonian, and whenever nontrivial, it prevents any group-symmetric Hamiltonian from supporting a unique gapped ground state. It is known that non-invertible symmetries can similarly impose strong constraints on the low-energy phase diagram of any local symmetric Hamiltonian, leading to anomalies or Lieb-Schultz-Mattis-type obstructions \cite{2024seibergshao,2024Seibergshaoseif,2023ShaoTasi}. Currently, most insights into the anomalies associated with non-invertible symmetries are framed within the paradigm of symmetry topological field theory—where a 1D system with (categorical) symmetry is viewed as the boundary of a 2D topological order described by the Drinfeld center. A natural question thus arises: can we define and characterize the anomalies of non-invertible symmetries from an operator algebra perspective, analogous to the approach used for invertible group-like symmetries \cite{2024Kapustin}?

\noindent\emph{Higher-dimensional Lattice Systems:} Significant progress has been made in recent years on the classification of unitary QCAs in spatial dimensions greater than 1 \cite{Haah:2019fqd, 2020FH,2023Haah}, with nontrivial examples in 3D now constructed \cite{2023FHH}. A natural and intriguing question arises: when considering a symmetric subalgebra in higher dimensions, how does one classify QCAs defined on such a subalgebra? For example, there is a generalization of the KW duality to 3D qubit lattice, known as the Wegner duality (for a recent discussion, see \cite{Gorantla:2024ocs}), which can be viewed as a QCA on the subalgebra symmetric under a $\Z_2$ 1-form symmetry. It would be interesting to investigate the classification of such QCAs for 1-form symmetric subalgebras.  

\begin{acknowledgements}
We thank Lukasz Fidkowski, Michael Levin, Laurens Lootens, Da-Chuan Lu, Shinsei Ryu, Daniel Spiegel, Alex Turzillo and Ruben Verresen for their insightful discussions. RM offers special gratitude to Corey Jones for his invaluable mathematical insights and to Zongping Gong for his guidance in the proof of Lemma \ref{lemma:sFDC}. MC is supported by NSF grant DMR-2424315. RM is supported in part by the Simons Investigator grant (990660) and by the Simons Collaboration on Ultra-Quantum Matter, which is a grant from the Simons Foundation (651442). YL is supported by the U.S. National Science Foundation under Grant No. NSF DMR-2316598.
\end{acknowledgements}

\appendix

\section{Connection to fermionic QCA}
\label{app:fermionize}
In this appendix, we provide a more direct illustration of the connection between QCA's on the $\Z_2$ symmetric subalgebra and those on a Majorana chain. The approach is similar to that in the proof of Theorem \ref{thm:Z2QCA}. However, here we select a subalgebra of the symmetric subalgebra $\B$ that mirrors the structure of a \emph{fermionic} operator algebra. In other words, we demonstrate how to fermionize a $\Z_2$ symmetric subalgebra.

Specifically, in Step 1 of the proof of Theorem, for each single-site $Z$ operator in the unit cells $I_{2n} \cup I_{2n+1}$, we now define an operator \begin{equation} 
\psi_i = Z_i (\otimes_{j = i+1}^{m-1} X_j ) Z_m, \label{eq:fermion} 
\end{equation} 
where $m - i > 2(2l+2)^2+2l$. The set of all such operators with $i \in I_{2n} \cup I_{2n+1}$ is denoted by $F_{2n}$, referred to as the fermionic operators. The reason for this terminology is clear: operators within this set obey a fermionic commutation relation, specifically, $\psi_i$ and $\psi_j$ for $i \neq j$ anti-commute. Additionally, we denote the set of single-site $X$ operators in $I_{2n} \cup I_{2n+1}$ as $A_{2n}$. We make the following observations:
\begin{itemize}
    \item The algebra $\B_{2n} = \langle A_{2n}, F_{2n} \rangle$ forms a subalgebra of the symmetric subalgebra $\B$, thus undergoing a well-defined transformation under $\alpha$. For each generator of $\B_{2n}$, we can define a \emph{local} fermion parity by examining the portion supported within $I_{2n} \cup I_{2n+1}$: operators in $A_{2n}$ are classified as locally even, while those in $F_{2n}$ are classified as \emph{locally} odd. By construction, any operator in $\B_{2n}$ can therefore be expressed as a linear combination of locally even and locally odd operators. Hereafter, we refer to the locally even operators (\ie those without a long tail) as bosonic operators.
    
    \item The algebra $\B_{2n}$ is isomorphic to a fermionic algebra. Specifically, consider a fermionic chain where each site has a two-dimensional Hilbert space, assumed to be uniform across the chain. This local Hilbert space is a $\Z_2$-graded Hilbert space, expressible as a direct sum of fermion parity even and odd components: \begin{equation} \Hi_{i} = \Hi_i^e \oplus \Hi_i^o, \end{equation} with each component having dimension 1. Accordingly, any operator $O$ also decomposes into fermion parity even and odd parts: \begin{equation} O = O^e \oplus O^o, \end{equation} where the even part $O^e$ connects states with the same fermion parity, while $O^o$ transitions between the two subspaces. Importantly, the algebra $\B_{2n}$ is exactly isomorphic to the operator algebra on the fermionic chain within the interval $I_{2n} \cup I_{2n+1}$. Specifically, two fermionic operators defined in Eq. \eqref{eq:fermion} at distinct sites satisfy the anti-commutation relation: \begin{equation} 
    \{ \psi_i, \psi_j \} = 0, \, \forall i \neq j, \label{eq:anticommute} 
    \end{equation} 
    and a fermionic operator commutes with a bosonic operator when they are supported on distinct sites: 
    \begin{equation} [\psi_i, O] = 0, 
    \label{eq:commute} 
    \end{equation} 
    where $O$ is a bosonic operator in $\B_{2n}$. Since $\psi_i$ has support over a large interval, ``disjoint supports” here means $i \notin \text{Supp}(O)$. More generally, for two operators $O_1$ and $O_2$ in $\B_{2n}$ with well-defined $\Z_2$ parity in $I_{2n} \cup I_{2n+1}$, we define the graded commutator as \cite{2019fQCA}: 
    \begin{equation} [O_1, O_2]_g = O_1 O_2 - (-1)^{|O_1|\cdot|O_2|} O_2 O_1, 
    \label{eq:gradedcommutator} 
    \end{equation} 
    where $|\cdot|$ denotes the $\Z_2$ parity within $I_{2n} \cup I_{2n+1}$: 0 for bosonic operators and 1 for fermionic ones. Operators $O_1$ and $O_2$ are said to ``graded commute” if $[O_1, O_2]_g = 0$. Eqs. \eqref{eq:anticommute} and \eqref{eq:commute} thus imply that operators supported on distinct sites graded commute.
\end{itemize}
We then repeat this construction for nearby unit cells $[I_{2k}, I_{2t+1}]$ ($k+3 < n < t-3$), all of which are sufficiently distant from site $m$. In this manner, we obtain a sequence of algebras, denoted $\B_{2k}, \dots, \B_{2t}$. By construction, fermionic operators within distinct $\B_{2n}$'s anti-commute with each other. Therefore, the algebra generated by this sequence, \begin{equation} \mathcal{M} = \langle \B_{2k}, \dots, \B_{2t} \rangle, \end{equation} is isomorphic to the full operator algebra supported on the interval $[I_{2k}, I_{2t+1}]$ on a fermionic chain. Moreover, $\mathcal{M}$ includes the $\Z_2$ symmetric subalgebra supported within $[I_{2k}, I_{2t+1}]$ as its bosonic part.

We now construct another sequence of algebras, which can be regarded as the ``image" of the preceding algebra sequence under the QCA $\alpha$. According to Prop.\ref{prop:Z2auto}, the fermionic operators are always mapped to other fermionic operators, specifically, \begin{equation} \alpha(\psi_i) = V_L S_X^{-l} V_R, \end{equation} where $V_L$ and $V_R$ are $\Z_2$-odd unitaries supported on $[i-l,i+l]$ and $[m-l,m+l]$, respectively, each chosen to be Hermitian. Since all fermionic operators defined above share the right endpoint operator $Z_m$, Prop.\ref{prop:Z2auto} implies that the right endpoint component of $\alpha(\psi_i)$, namely $V_R$, is identical for all these operators. This observation motivates the definition of a new set of fermionic operators as 
\begin{equation} 
\psi'_i = Z_i \Big(\prod_{j=i+1}^{m-l-1} X_j\Big) V_R, 
\end{equation} where the right component is now the transformed operator. We then define $\mathcal{C}_{2n-1} = \langle A_{2n-1}, F_{2n-1} \rangle$, where $A_{2n-1}$ denotes the set of single-site $X$ operators on $I_{2n-1} \cup I_{2n}$, and $F_{2n-1} = \{ \psi'_i | i \in I_{2n-1} \cup I_{2n} \}$. This construction is similarly repeated for nearby unit cells $[I_{2k-1}, I_{2t+2}]$.

By construction, the algebras in the sequence $\mathcal{C}_{2k-1}, \dots, \mathcal{C}_{2t+1}$ exhibit the following properties: (1) they graded-commute with each other; and (2) the algebra generated by this sequence, 
\begin{equation} 
\mathcal{N} = \langle \mathcal{C}_{2k-1}, \dots, \mathcal{C}_{2t+1} \rangle, 
\end{equation} 
is isomorphic to the algebra of all operators on a fermionic chain, supported within the interval $[I_{2k-1}, I_{2t+2}]$. Moreover, it contains the $\Z_2$ symmetric subalgebra within $[I_{2k-1}, I_{2t+2}]$ as its bosonic part.

The remaining steps in the proof of Theorem \ref{thm:Z2QCA} proceed in a similar manner. Since $\alpha(\B_{2n}) \subset \C_{2n-1} \otimes \C_{2n+1}$ (where $\otimes$ here denotes the graded tensor product, meaning that fermionic operators on distinct sites anti-commute), we can define the support algebras of $\alpha(\B_{2n})$ within $\mathcal{C}_{2n-1}$ and $\mathcal{C}_{2n+1}$, respectively: \begin{equation} 
\begin{split} \Le_{2n} & = \sa(\alpha(\B_{2n}), \mathcal{C}_{2n-1}), \\ 
\Ri_{2n} & = \sa(\alpha(\B_{2n}), \mathcal{C}_{2n+1}), 
\end{split} \end{equation} 
for $n \in [k, t]$. A geometric illustration of these algebras is shown in Fig. \ref{Fig:partition}. By the same argument as that presented in the proof of Theorem \ref{thm:Z2QCA} and Appendix A of Ref.\cite{2019fQCA}, these support algebras have the following properties:
\begin{itemize} 
\item We have 
\begin{equation} \begin{split} 
&\alpha(\B_{2n}) = \Le_{2n} \otimes \Ri_{2n}, \\ &\mathcal{C}_{2n-1} = \Ri_{2n-2} \otimes \Le_{2n}, \end{split} \label{eq:fsupp} \end{equation} 
for $n \in [k+1, t-1]$.
\item The support algebras $\Le_{2n}$ and $\Ri_{2n}$ for $n \in [k+1, t-1]$ are $\Z_2$-graded algebras, \ie they decompose as a direct sum of locally odd and locally even components. Furthermore, they have a trivial graded center, meaning that any element in $\Le_{2n}$ (or $\Ri_{2n}$) that graded commutes with all other elements in $\Le_{2n}$ (or $\Ri_{2n}$) must be proportional to the identity. Such algebras with a trivial graded center are referred to as central simple graded algebras.
\end{itemize}

It is known that a central simple $\Z_2$-graded algebra is classified by the $\Z_2$-graded Brauer group over the complex numbers \cite{knus1969algebras}, denoted Br$_{\Z_2}(\mathbb{C}) = \Z_2$. Specifically, such an algebra is isomorphic to one of the following types: 
\begin{itemize} \item End$_\mathbb{C}(V)$, where $V$ is a $\Z_2$-graded Hilbert space, corresponding to the trivial element in Br$_{\Z_2}(\mathbb{C}) = \Z_2$. As an ungraded algebra, this is simply a matrix algebra. In a fermionic chain, this is the support algebra associated with an FDC \cite{2019fQCA}.
\item The Clifford algebra, $Cl_1$, which corresponds to the nontrivial element in Br$_{\Z_2}(\mathbb{C}) = \Z_2$. In a fermionic chain, this is precisely the support algebra associated with a Majorana translation \cite{2019fQCA}. 
\end{itemize}

As an illustrative example, we examine the support algebras for the identity QCA and KW, the KW duality defined in Eq.\eqref{eq:KWdual}. Since the spread of both QCAs is upper bounded by $l=1$, each unit cell in Fig.\ref{Fig:partition} can be chosen to contain a single site. For the identity, we find the support algebra 
\begin{equation} \Le_{2n} = \langle X_{2n}, \psi_{2n} \rangle, 
\end{equation} 
which is isomorphic to a matrix algebra on a two-dimensional $\Z_2$-graded Hilbert space, \ie $\Le_{2n} \cong \mathrm{End}_\mathbb{C}(\mathbb{C}^{1|1})$. On the other hand, for the KW transformation, we find the support algebra 
\begin{equation} \Le_{2n} = \langle Z_{2n}(\prod_{j=2n+1}^{m-1} X_j) Y_{m} \rangle, \end{equation} 
which is generated by a single order-2 element. This can be interpreted as a single Majorana mode and corresponds precisely to the support algebra of a Majorana translation \cite{2019fQCA}.

\section{Comparison of the non-invertible operators in $\mathrm{Rep}(\mathrm{D}_8)$ and $\mathrm{Rep}(\mathrm{Q}_8)$}

\label{app:Rep}

In this appendix, we demonstrate that $\mathrm{D}'$, the non-invertible operator in the Rep(Q$_8$) fusion category, can be derived from $\mathrm{D}$, the non-invertible operator in the Rep(D$_8$) fusion category, via an FDC $W$ composed of $\Z_2 \times \Z_2$ symmetric gates. The spatial truncation of $WW'$, denoted $U_\Lambda = O_L S_g(\Lambda^{-l'})O_R^\dagger$ over a finite interval $\Lambda$, satisfies the following properties: (1) $S_g(\Lambda^{-l'})$ is a symmetry string associated with the diagonal $\Z_2$ element $ab$; (2) both $O_L$ and $O_R$ carry nontrivial charges under the symmetries $a$ and $b$.

To distinguish between the two non-invertible operators, we sequentially gauge a $\Z_2^3$ global symmetry in 2D, employing two distinct types of Dijkgraaf-Witten twists. The resulting 2D bulk theories correspond to the Drinfeld centers, $\mathcal{Z}(\mathrm{Rep}(\mathrm{D}_8))$ and $\mathcal{Z}(\mathrm{Rep}(\mathrm{Q}_8))$, respectively. By examining the bulk anyon lines, one can infer the difference between the two non-invertible symmetry operators acting in 1D.

The two Dijkgraaf-Witten twists are specified as Type III and Type II+II+III \cite{2015JCWang}, respectively: 
\begin{equation} 
\begin{split} 
&\mathcal{L}_1 = \frac{1}{2} A_1 \cup A_2 \cup A_3, \\ 
&\mathcal{L}_2 = \frac{1}{2}(A_1 \cup A_1 \cup A_3 + A_2 \cup A_2 \cup A_3 + A_1 \cup A_2 \cup A_3), 
\end{split} \end{equation} 
where $A_1$, $A_2$, and $A_3$ are background gauge fields associated with the three $\Z_2$ symmetries. We first promote $A_1$ and $A_2$ to dynamical gauge fields, denoted as $a_1$ and $a_2$. In both resulting theories, a remaining $\Z_2$ global symmetry persists, generated by two topological surface operators: 
\begin{equation} \begin{split} &U_1 = \exp(i\pi\int_M a_1 \cup a_2), \\ 
&U_2 = \exp(i\pi\int_M (a_1 \cup a_1 + a_2 \cup a_2 + a_1 \cup a_2)), 
\end{split} \label{eq:D8_Q8_sym_op}\end{equation} 
where the integrals are taken over $M$, the worldsheet of the symmetry operators. Three key observations will be useful for the subsequent discussion:
\begin{itemize} \item When $M$ is a closed surface with no boundary, we have $U_1 = U_2$, and both operators square to the identity. This aligns with the fact that they generate a $\Z_2$ symmetry.

\item Due to the presence of the $a_1 \cup a_2$ term, both surface operators carry a gauged $1+1$-dimensional $\Z_2 \times \Z_2$ SPT phase on their worldsheet. Thus, if we label the gauge charge of the $a_1$ ($a_2$) field as $e_1$ ($e_2$), and the corresponding gauge fluxes as $m_1$ ($m_2$), the symmetry operators act nontrivially on the anyons, implementing the following permutation:
\begin{equation}
    \begin{split}
        & e_1 \mapsto e_1, \quad m_1 \mapsto e_2 m_1, \\
        & e_2 \mapsto e_2, \quad m_2 \mapsto e_1 m_2.
    \end{split}
    \label{eq:thirdZ2}
\end{equation}

\item When the surface $M$ has a boundary, the two symmetry operators become distinct:
\begin{equation}
    U_2 = U_1 \exp(i\pi\int_{\partial M} \frac{1}{2}a_1 + \frac{1}{2}a_2),
    \label{eq:halfanyon}
\end{equation}
which can be understood by noting that, for a $\Z_2$ gauge field, $a_1 \cup a_1 = \frac{1}{2} d a_1$ and similarly for $a_2$. Importantly, along the boundary, the two operators differ by a ``half" anyon line associated with $e_1 e_2$.

\end{itemize}

We further note that when surface $M$ has a boundary, the symmetry operators in Eq.\eqref{eq:D8_Q8_sym_op} are not gauge invariant. To restore the gauge invariant, we can introduce ancillary degrees of freedom on $\partial M$, and write the symmetry operators as follows,
\begin{equation}
    D_i = U_i \cdot \sum_{\phi_1,\phi_2 }\exp(i\pi \int_{\partial M}\phi_1\cup \phi_2+\phi_1\cup a_2+a_1\cup \phi_2),
\end{equation}
where $\phi_1$ and $\phi_2$ are $\Z_2$-valued ancillary fields. Under a gauge transformation $a_i\rightarrow a_i+\delta\lambda_i$, these fields transform as $\phi_i\rightarrow\phi_i+\lambda_i$. Thus, we sacrifice the unitarity of symmetry operators in order to restore the gauge invariance. By straightforward calculation, we see that
\begin{equation}
    D_i^2=1+\exp(i\pi \int_{\partial M}a_1)+\exp(i\pi \int_{\partial M}a_2)+\exp(i\pi \int_{\partial M}a_1+a_2).
\end{equation}
This result agrees with the fusion rule of TY-categories in Eq.\eqref{eq:TYfusion} on the boundary, after a relabeling of anyons as $m_1 \rightarrow \tilde{e}_2$, $m_2 \rightarrow \tilde{e}_1$, $e_1 \rightarrow \tilde{e}_1\tilde{m}_2$, and $e_2 \rightarrow \tilde{e}_2\tilde{m}_1$.

Now we gauge the remaining $\mathbb{Z}_2$ global symmetry. After the above relabeling, when anyons braid around a gauge flux of the third $\mathbb{Z}_2$ symmetry (denoted by $\mu_1$ and $\mu_2$ for the two distinct types of Dijkgraaf-Witten twists), they are permuted according to the anyon permutation rule in Eq.\eqref{eq:thirdZ2}:
\begin{equation} \tilde{e}_1 \leftrightarrow \tilde{m}_2,\quad \tilde{e}_2 \leftrightarrow \tilde{m}_1. 
\end{equation}
Due to the ``PEM Theorem,"\cite{2015PEM,2020PEM} the 2D topological orders associated with the Dijkgraaf-Witten twists $\mathcal{L}_1$ and $\mathcal{L}_2$ are $\mathcal{Z}(\mathrm{Rep}(\mathrm{D}_8))$ and $\mathcal{Z}(\mathrm{Rep}(\mathrm{Q}_8))$, respectively. Consequently, the boundary descendants of $\tilde{m}_1$, $\tilde{m}_2$, and $\mu_1$ ($\mu_2$) form the fusion categories $\mathrm{Rep}(\mathrm{D}_8)$ ($\mathrm{Rep}(\mathrm{Q}_8)$). The boundary descendants of $\mu_1$ (or $\mu_2$), $\tilde{m}_1$, and $\tilde{m}_2$ correspond to the non-invertible operators $\mathrm{D}$ (or $\mathrm{D}'$), as well as the $\Z_2^e$ and $\Z_2^o$ symmetries, respectively. Furthermore, according to Eq.\eqref{eq:halfanyon}, if we fuse two $\mu_1$ and $\mu_2$ in the bulk, we find that the fusion results differ by an anyon $\tilde{e}_1\tilde{m}_1\tilde{e}_2\tilde{m}_2$. This leads to the difference between the two non-invertible operators $\mathrm{D}$ and $\mathrm{D}'$ stated at the beginning of this Appendix.

\section{Classification of $G$-symmetric FDCs}
\label{app:sFDC}

First let us briefly review the definition of $H^2$ cohomology class for a $G$-symmetric uQCA, using the matrix product unitary (MPU) representation. It has been shown in Ref.\cite{2017MPUCirac,2018MPUChen} that in 1D, any uQCA can be written as a MPU.  By blocking a finite number $k$ of sites, the MPU can always be brought into the ``standard form" \cite{2017MPUCirac} (\ie the ``two-layer construction”), as shown in Fig.\ref{Fig:twolayer}, where $u$ and $v$ are unitaries. Each coarse-grained site (represented by the black vertical lines in Fig.\ref{Fig:twolayer}) then carries $\rho_g = R_g^{\otimes k}$, \ie a $k$-th power of the regular representation of $G$, and the tensors $u$ and $v$ transform under the $G$ symmetry as shown in Figs \ref{fig:utrans} and \ref{fig:vtrans}. Here $x_g$ and $y_g$ are projective representations of $G$. The second cohomology class of the uQCA is defined to be the one associated with the projective representation $x_g$ (or $y_g^*$) \cite{2021Gonginfoflow}. One can show that the cohomology classes multiply under the composition of uQCAs. In addition, it is easy to see that a FDC with $G$-symmetric gates carries the trivial class.

Thus, to prove the classification in Theorem \ref{thm:GsymFDC}, it suffices to show that a $G$-symmetric MPU with a trivial GNVW index and trivial group cohomology is always an FDC with $G$-symmetric gates, if each site is in a regular representation of $G$.

\begin{figure}
\begin{center}
  \includegraphics[width=.45\textwidth]{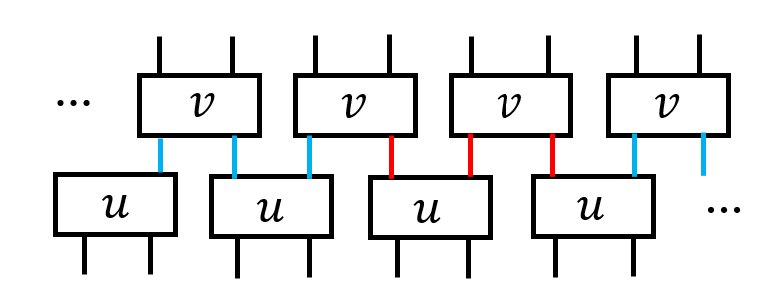} 
\end{center}
\caption{The standard form of a matrix product unitary (MPU).}
\label{Fig:twolayer}
\end{figure}

\begin{lemma}
    Let $G$ be a finite Abelian group, and suppose each site has a local Hilbert space in the regular representation of $G$. Then a $G$-symmetric FDC associated with the trivial element in $H^2(G, \U)$ is an FDC with $G$-symmetric gates.
    \label{lemma:sFDC}
\end{lemma} 

\begin{proof}
 Here, (1) $x_g$ and $y_g$ are linear representations of $G$, since the MPU has trivial group cohomology \cite{2018sMPU}, and (2) $x_g$ and $y_g$ both have the same dimension as $\rho_g$, since the MPU has a trivial GNVW index \cite{2017MPUCirac}. We now show that, with further blocking, the MPU becomes an FDC with symmetric gates.

\begin{figure}
     \centering
     \begin{subfigure}[b]{0.3\textwidth}
         \centering
         \includegraphics[width=\textwidth]{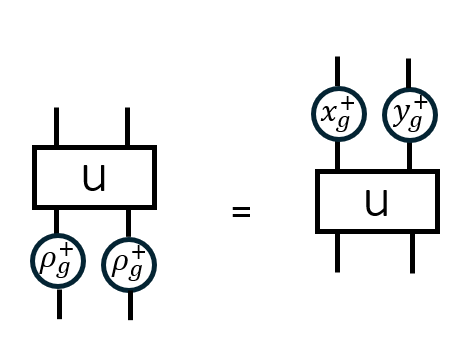}
         \caption{}
         \label{fig:utrans}
     \end{subfigure}
     \quad \quad \quad \quad 
     \begin{subfigure}[b]{0.3\textwidth}
         \centering
         \includegraphics[width=\textwidth]{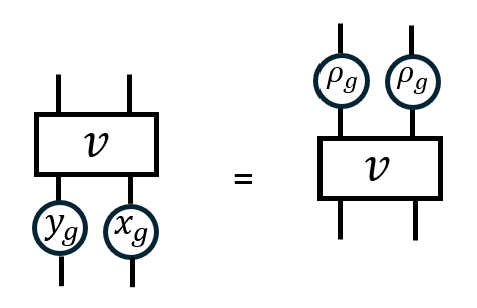}
         \caption{}
         \label{fig:vtrans}
     \end{subfigure}
     \caption{Transformation of the building blocks of the standard form under the action of the symmetry}
\end{figure}

Let us combine three $v$ gates into a coarse-grained gate, denoted by $v'$. Specifically, we combine three consecutive blue legs in Fig.\ref{Fig:twolayer}, carrying representations $y_g$, $x_g$, and $y_g$, into one leg. Given that $y_g \otimes x_g = v^\dagger (\rho_g \otimes \rho_g) v$ (the transformation of $v$ shown in Fig.\ref{fig:vtrans}), we see that $y_g \otimes x_g$ is equivalent to a $2k$-fold tensor product of the regular representation. Consequently, the combined leg carries a representation equivalent to a $3k$-fold tensor product of the regular representation, 
\begin{equation} y_g \otimes x_g \otimes y_g = U_a R_g^{\otimes 3k} U_a^\dagger, 
\end{equation} 
where $U_a$ is a unitary gauge transformation. This can be seen by checking the character: \begin{equation} 
\tr(y_g \otimes x_g \otimes y_g) = \tr(\rho_g \otimes \rho_g \otimes y_g) = \delta_{e,g} \dim(R_g)^{3k}, 
\end{equation} 
where $e$ is the identity element in $G$, and we use the fact that the regular representation has the property $\tr(R_g) = \delta_{e,g} \dim(R_g)$. Similarly, we combine three red legs into a single leg, which then carries a representation equivalent to a $3k$-fold tensor product of the regular representation, 
\begin{equation} 
x_g \otimes y_g \otimes x_g = U_b R_g^{\otimes 3k} U_b^\dagger. \end{equation} 
Thus, after a gauge transformation $\tilde{v} = v'(U_a \otimes U_b)$, the coarse-grained gate $\tilde{v}$, as shown in Fig.\ref{Fig:vtilde} is symmetric under an onsite $3k$-fold regular representation. Similarly, by combining three $u$ gates into a new coarse-grained gate and performing a gauge transformation, we obtain $\tilde{u} = (U_b^\dagger \otimes U_a^\dagger) u$, which is also symmetric under an onsite $3k$-fold regular representation. This completes the proof that, on a chain where each site carries a regular representation, a $G$-symmetric FDC with a trivial group cohomology is an FDC with symmetric gates, given by $\tilde{v}$ and $\tilde{u}$. 
\end{proof}

\begin{figure}
\begin{center}
  \includegraphics[width=.25\textwidth]{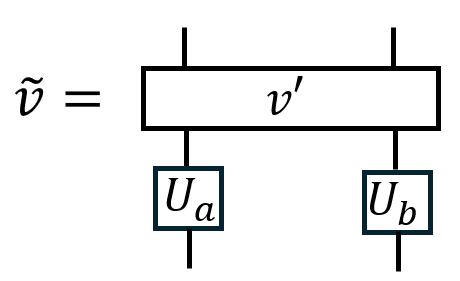} 
\end{center}
\caption{Building blocks of the sFDC: $v'$ and $u'$ (not shown) are symmetric local gates. Each leg in this diagram represents a combination of three legs from Fig. \ref{Fig:twolayer}.
}
\label{Fig:vtilde}
\end{figure}

\bibliographystyle{quantum}
\bibliography{Ref.bib} 
\end{document}